\documentclass[aps,twocolumn,superscriptaddress,nofootinbib,a4paper]{revtex4-2}

\usepackage{times}
\usepackage{epsfig}
\usepackage{amsfonts}
\usepackage{amsmath}
\usepackage{amsthm}
\usepackage{amssymb}					
\usepackage{amsthm}
\usepackage{dsfont}
\usepackage{bm}
\usepackage{mathtools}
\usepackage{color}
\usepackage{multirow}
\usepackage[normalem]{ulem}
\newcommand{\stkout}[1]{\ifmmode\text{\sout{\ensuremath{#1}}}\else\sout{#1}\fi}
\usepackage{latexsym}
\usepackage{mathrsfs}
\usepackage{natbib}
\usepackage{verbatim}
\usepackage[T1]{fontenc}
\usepackage{float}
\usepackage{graphicx}
\usepackage{subcaption}
\usepackage{xcolor}
\usepackage{physics}
\usepackage{soul}
\usepackage{caption}
\usepackage{subcaption}

\usepackage{centernot}
\usepackage[font=small,labelfont=bf]{caption}

\usepackage{xspace}

\newtheorem{lemma}{Lemma}
\newtheorem{property}{Property}

\newtheorem{result}{Result} 

\newtheorem{definition}{Definition}

\newcommand{\floor}[1]{\lfloor #1 \rfloor}

\newcommand{\bracket}[3]{\langle#1|#2|#3\rangle}
\newcommand{\expect}[1]{\langle#1\rangle}

\usepackage{amssymb}
\usepackage{amsthm}
\usepackage{mathtools} 
\usepackage[customcolors]{hf-tikz} 
\usetikzlibrary{patterns}
\usetikzlibrary{matrix,decorations.pathreplacing}

\pgfkeys{tikz/mymatrixenv/.style={decoration={brace},every left delimiter/.style={xshift=8pt},every right delimiter/.style={xshift=-8pt}}}
\pgfkeys{tikz/mymatrix/.style={matrix of math nodes,nodes in empty cells,left delimiter={[},right delimiter={]},inner sep=1pt,outer sep=1.5pt,column sep=2pt,row sep=2pt,nodes={minimum width=20pt,minimum height=10pt,anchor=center,inner sep=0pt,outer sep=0pt}}}
\pgfkeys{tikz/mymatrixbrace/.style={decorate,thick}}

\usepackage{ytableau}
\usepackage{youngtab}
\ytableausetup{boxsize=1.1em}

\usepackage[colorlinks=true,linkcolor=blue,citecolor=magenta,urlcolor=blue]{hyperref}

\begin{document}
	

\title{The superb supersinglet}
	
\author{Gabriele Cobucci}
\affiliation{Physics Department and NanoLund, Lund University, Box 118, 22100 Lund, Sweden.}
\author{Armin Tavakoli}	
\affiliation{Physics Department and NanoLund, Lund University, Box 118, 22100 Lund, Sweden.}
\author{Alexander Bernal}
\affiliation{Physics Department and NanoLund, Lund University, Box 118, 22100 Lund, Sweden.}
\author{Shishir Khandelwal}
\affiliation{Physics Department and NanoLund, Lund University, Box 118, 22100 Lund, Sweden.}

\begin{abstract}
Multipartite systems of high-dimensional particles can host forms of quantum entanglement inaccessible to qubit systems. A particularly fascinating class is formed by the fully antisymmetric states, also known as supersinglets, which possess a remarkable symmetry: they are invariant under arbitrary collective local unitaries. Here, we explore their entanglement properties and establish their quantum information applications. Using convex optimisation, we show that central entanglement features of supersinglets are remarkably robust to noise, substantially surpassing the robustness of several paradigmatic entangled states under certain relaxations. We then develop witness methods, based on collective-spin or randomised measurements, that certify both the depth and dimensionality of their entanglement. Finally, we identify supersinglets as a natural resource for high-dimensional gradient sensing and prove that they enable Heisenberg scaling. Together, these results reveal supersinglets as an exceptionally robust and versatile class of multipartite entangled states, combining striking symmetry and strong entanglement with concrete potential for quantum sensing.
\end{abstract}

\date{\today}	
\maketitle

\section{Introduction}

Entanglement becomes dramatically richer when more than two particles are involved. While bipartite pure-state entanglement admits a unique maximally entangled state up to local unitaries, already three qubits exhibit two inequivalent classes of genuine tripartite entanglement \cite{Dur2000}, and four qubits give rise to nine continuous families of entanglement classes \cite{Verstraete2002}. Against this enormous landscape of qualitatively different states available to $n$ particles of local dimension $d$, it is striking that quantum information science repeatedly returns to a relatively small cast of paradigmatic states. Among them are Greenberger--Horne--Zeilinger (GHZ) states, rooted in quantum nonlocality; cluster states, originating in measurement-based quantum computation, and their generalisation to graph states; Dicke states of qubits; and absolutely maximally entangled states, whenever they exist. This raises a natural question --- what other multipartite high-dimensional states deserve a place in this cast? Particularly compelling candidates are states distinguished by powerful symmetries, unusual entanglement structures, and the possibility of new quantum-information applications.

One such candidate is the fully antisymmetric state, whose origins long predate quantum information. The antisymmetric wavefunction of $n$ identical fermions occupying distinct orbitals is one of the fundamental constructions of quantum mechanics \cite{Slater1929,Weyl1928,Weyl1931}. Viewed instead as the state of $n$ distinguishable $n$-dimensional quantum systems, it takes the form
\begin{equation}\label{supersinglet}
	\ket{S_n}=\frac{1}{\sqrt{n!}}\sum_{i_1,\ldots,i_n=0}^{n-1}
	\epsilon_{i_1\ldots i_n}\ket{i_1,\ldots,i_n},
\end{equation}
where $\epsilon_{i_1\ldots i_n}\in\{-1,0,+1\}$ is the Levi--Civita tensor. It vanishes whenever two indices coincide and changes sign under any odd permutation,
$\epsilon_{i_{\pi(1)}\ldots i_{\pi(n)}}=\operatorname{sgn}(\pi)\epsilon_{i_1\ldots i_n}$.

What makes this state especially remarkable is its symmetry. It is invariant under an arbitrary collective special unitary transformation,
\begin{equation}
	U^{\otimes n}\ket{S_n}=\ket{S_n},
	\qquad \forall\, U\in SU(n).
\end{equation}
Thus, no matter which $SU(n)$ transformation is applied identically to every subsystem, the state remains unchanged. For $n=2$, $\ket{S_n}$ reduces to the familiar two-qubit singlet. For arbitrary $n$, it provides its natural high-dimensional counterpart. Although the state is often called the fully antisymmetric state, we adopt the more evocative terminology of Ref.~\cite{Cabello2003} and call $\ket{S_n}$ a \textit{supersinglet}.\footnote{The terminology requires some care. Ref~\cite{Cabello2003} uses ``supersinglet'' for states invariant under collective SU$(2)$ transformations. Here, we require invariance under collective SU$(d)$ transformations.}

This combination of simplicity and symmetry has repeatedly brought supersinglets into quantum information theory. Early works proposed using $\ket{S_3}$ to generate eigenstates of unknown qutrit unitaries \cite{Hillery2001} and to solve the detectable Byzantine agreement problem \cite{Fitzi2001}. More generally, Refs~\cite{Cabello2002,Jex2003} identified several information-processing tasks naturally suited to $\ket{S_n}$. More recently, supersinglets have appeared in seemingly disparate questions across the field: they can be device-independently self-tested \cite{Saha2025,Konderak2025}, provide useful constructions for studying multipartite entanglement \cite{Steinberg2024,Bernards2024}, and, remarkably, $\ket{S_3}$ can be transformed by LOCC into locally inequivalent fully entangled pure states \cite{Li2024}, a possibility that is exceptional in the multipartite setting \cite{Gour2017,Sauerwein2018}. In parallel, a series of works has proposed schemes for realising supersinglets in different physical platforms \cite{Jin2005,Lin2007,Shao2010,Lu2012,Chen2016,IloOkeke2022,Piccolini2025}. Nevertheless, a supersinglet has yet to be experimentally realised in a quantum information setting.

These observations suggest that supersinglets are much more than an elegant consequence of antisymmetry. Their unusually strong symmetry, intrinsically high-dimensional structure, and recurring appearance in quantum information problems invite a systematic investigation of what makes their entanglement special and what it can be used for. In this work, we pursue these questions from three complementary perspectives: the robustness of supersinglet entanglement, its experimental certification, and its use as a resource for quantum sensing.

First, we investigate the robustness of the genuinely high-dimensional multipartite entanglement \cite{Cobucci2024} carried by supersinglets. Using convex-programming relaxations \cite{Tavakoli2024} together with symmetrisation techniques, we obtain benchmarks for the noise robustness of their entanglement properties up to $n=7$, corresponding to a Hilbert-space dimension exceeding eight hundred thousand. We find a remarkable robustness to noise, which goes beyond that exhibited by the GHZ and cluster states in the cases considered in our comparison.

Second, we develop experimentally accessible methods for witnessing this entanglement. We first construct a witness based on randomised basis measurements and show that it can certify the genuine dimensionality of genuine multipartite entanglement \cite{Cobucci2024} with a practical number of measurement samples. We then turn to a considerably more economical setting requiring only collective-spin measurements. We show that these limited measurement resources are sufficient to witness both the dimensionality and the depth of entanglement.

Finally, we uncover a direct metrological role for supersinglets. We consider the problem of sensing Hamiltonian gradients in high-dimensional systems and identify the supersinglet as a natural resource for this task. Suitably constructed macroscopic supersinglet states attain Heisenberg scaling and are optimal for gradient sensing under broad conditions. Moreover, we connect this enhanced sensitivity directly to their multipartite and high-dimensional entanglement structure.

Taken together, these results reveal a rather unusual combination of properties. Supersinglets possess an exceptionally strong symmetry, remarkably robust high-dimensional multipartite entanglement, experimentally accessible signatures, and an uncontrived application in quantum sensing. They thereby emerge not merely as elegant generalisations of the familiar two-qubit singlet, but as a distinctive class of multipartite states worthy of attention in their own right.
\section{Elementary properties}\label{Properties}
We begin with a list of elementary properties of supersinglets. These properties can all be found in the quantum information literature, but they are often presented without proof or argument. Here, we give a complete introduction to these properties and discuss their associated implications. 

To this end, we begin by formally defining a supersinglet.
\begin{definition}\label{supersinglet_definition}
A pure $n$-particle state, $\ket{\psi}$, with local dimension $d$, is a supersinglet when
\begin{equation}
U^{\otimes n} \ket{\psi}=\ket{\psi}, \qquad \forall\ U\in SU(d).
\end{equation}
\end{definition} 
Thus, supersinglets are the states that remain invariant if we were to apply any given local unitary to every particle. A natural question is whether there exists a state that satisfies the above definition. We list the answer as our first property.
\begin{property}\label{Property1}
	Supersinglets exist if and only if $n=md$ for an integer $m$. In particular, for $n=d$ the supersinglet is unique (up to local unitaries) and given by Eq~\eqref{supersinglet}. For $m>1$, every supersinglet takes the form 
	\begin{equation}
		\ket{\psi} = \sum_{T \in \text{SYT}(d,m)} c_{T} \, P_{T} \ket{S_{d}}^{\otimes m},
	\end{equation}
	where $\ket{S_{d}}$ is the $d$-dimensional $d$-partite supersinglet, $c_{T}$ are complex coefficients and $\lbrace T\rbrace$ denote the permutations (and $P_T$ the corresponding permutation operator) among all possible Standard Young Tableaux (SYT) of $d$ rows and $m$ columns.
\end{property}
The proof is given in Appendix~\ref{app:unique}. A central consequence of Property~\ref{Property1} is that for all $(n,d)$ for which a supersinglet exists, it is a linear combination of copies and sub-system permutations of the supersinglet corresponding to $n=d$. Therefore, we will focus our analysis on this $n=d$ elementary building block and use the name supersinglet to refer directly to the specific class of states defined in Eq~\eqref{supersinglet}. An immediate consequence of this definition is that the supersinglet is globally antisymmetric in the following operational sense: \textit{if all particles are measured in the same basis, then all outcomes are distinct}.


The elementary supersinglet $\ket{S_n}$ can also be seen as an appropriate superposition of smaller supersinglets. This is formalised by the following property. 
\begin{property}\label{Property2}
For any $T\subset [n]$ with $|T|=k$, let  $\ket{S_k(T)}$ be the $k$-particle supersinglet involving the level labels in $T$ (in increasing order). Then 	
	\begin{equation}\label{supersinglet_dec}
		\ket{S_n}=\frac{1}{\sqrt{\binom{n}{k}}} \sum_{\substack{T\subset [n]\\ |T|=k}} \text{sgn}(T\oplus \bar{T}) \ket{S_k(T)} \otimes\ket{S_{n-k}(\bar{T})}.
	\end{equation}
	where we have defined $\ket{S_1(\{i\})}=\ket{i}$,  $\bar{T}=[n]\setminus T$ and $\oplus$ as the ordered union of sets.	
\end{property}
See Appendix~\ref{app:prop2} 
for a proof. An immediate corollary of this property is that all reduced states of the supersinglet are maximally mixed over the antisymmetric space of the remaining particles. As Eq~\eqref{supersinglet_dec} is a Schmidt decomposition with equal coefficients, and the states $\ket{S_k(T)}$ form an orthonormal basis of the $k$-particle antisymmetric subspace, tracing out the other particles leaves an equal mixture of all these basis states, i.e., the maximally mixed state on that antisymmetric subspace.

\begin{property}\label{Property3}
	Let $I\subset [n]$ be any subset of the particle indices. The reduced state of the supersinglet on the particles $I$ is
	\begin{equation}\label{supersingletreduced}
		\rho_I\equiv \tr_{\bar{I}} \left( \ketbra{S_n}{S_n}\right) = \frac{1}{\binom{n}{|I|}} \Pi^{\text{asym}}_I,
	\end{equation}
	where $\bar{I}=[n]\setminus I$  and $\Pi^{\text{asym}}_I$ is the projector onto the fully antisymmetric space of the particles indexed in $I$,
	\begin{align}
		\Pi_I^{\text{asym}} = \sum_{i \in \binom{[n]}{|I|}} \ketbra{S_i}{S_i}.
	\end{align}
\end{property}

See Appendix \ref{app:prop3} for a proof. Property 3 lets us think about the supersinglet as an antisymmetric counterpart to the absolutely maximally entangled state. The absolutely maximally entangled state  is defined by all its reductions being maximally mixed. The supersinglet has the same feature, but only over the antisymmetric subspace. Furthermore, the space defined by $\Pi^\text{asym}$ is called genuinely entangled \cite{Demianowicz_2018}, meaning that it is inhabited only by genuinely multipartite entangled (GME) states. This can be seen by defining the complement space $V = \lbrace \ket{\varphi}: \, \ket{\varphi} = \ket{\phi_1}_{I} \otimes \ket{\phi_2}_{\bar{I}}\rbrace$ for some bipartitions $(I|\bar{I})$ and verifying that it cannot contain antisymmetric states. An interesting consequence is that the entanglement of the supersinglet is robust to the loss of arbitrarily many particles, even though the reduced state is often highly mixed. It was shown in Ref~\cite{Christandl2012} that any reduction $\rho_I$ has an entanglement cost bounded away from zero, independently of $n$. 


\begin{figure}
	\centering
	\includegraphics[width=0.8\columnwidth]{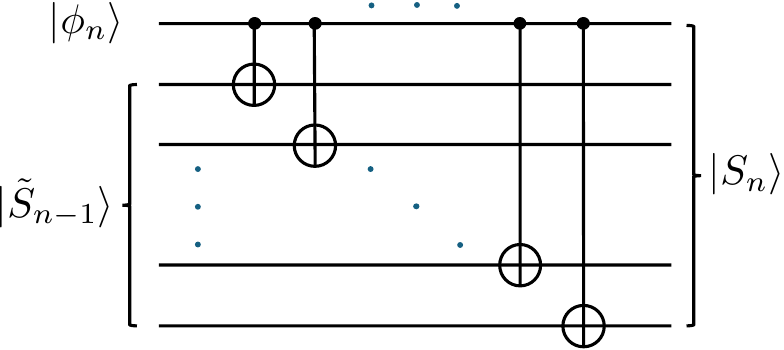}
	\caption{A quantum circuit to construct a $n$-dimensional supersinglet using an $n-1$-dimensional supersinglet, Fourier state and $n-1$ CNOT gates.}
	\label{fig:circ}
\end{figure}

The next property shows how  supersinglets can be constructed recursively through a simple quantum circuit.  

\begin{property}\label{Property4}
	 There exists a circuit that takes as input  $\ket{S_{n-1}}$ and an ancilla and outputs $\ket{S_n}$.  It  uses $n-1$ controlled two-qudit gates. 
\end{property}	 
 If $n>2$, let $|\tilde S_{n-1}\rangle$ be a $(n-1)$-dimensional supersinglet embedded in the levels $\{1,2,\dots,n-1\}$ and if $n=2$, let $|\tilde S_n\rangle=\ket{1}$. With the $n$-dimensional Fourier transform $F_n$, we define $\ket{\phi_{n}} = F_n\ket{0}$ for odd $n$ and $\ket{\phi_{n}} = F_n\ket{n/2}$ for all even $n$. The $n$-dimensional CNOT gate is defined as $C_{ij} \ket{a}_i \ket{b}_j = \ket{a}_i\ket{a+b\, \text{mod}\, n }_j$. Then the following property is satisfied,
\begin{align}
		\ket{S_n} \equiv \left(\prod_{j=2}^{n} C_{1j} \right)\ket{\phi_n}\otimes  |\tilde S_{n-1}\rangle,
\end{align}up to a global phase. See Appendix~\ref{app:circ} 
for a proof. The circuit is shown in Fig.~\ref{fig:circ}. This is a generalisation of the scheme to produce a qutrit supersinglet state presented in Ref~\cite{Jex2003}. This property can then be utilised to produce supersinglets of any dimension in a recursive manner, starting with a separable state.  Starting with $n=2$, every step of the recursion then requires the circuit, with exactly $n-1$ two-qudit CNOT gates, apart from Fourier state preparations. Overall, the procedure therefore requires $n(n-1)/2$ CNOT gates.

\section{Entanglement robustness}\label{entanglement_properties}

We will now turn to the study of the entanglement structure of supersinglets. In particular, we focus on two of the main aspects that characterise high-dimensional multi-particle states, namely the depth of entanglement and its dimensionality. Let us begin by briefly introducing these concepts.

A pure state $\ket{\psi}$ is called $k$-producible if it can be factored as $\ket{\psi}=\ket{\phi_1}_{I_1}\otimes\ldots\otimes \ket{\phi_m}_{I_m}$ where $\{I_i\}_i$ is a partition of $\{1,\ldots,n\} $ and $|I_i|\leq k$ for all $i$.
This means that the state is composed of $m$ separate clusters and that none of these clusters entangles more than $k$ particles. In particular, $k=1$ corresponds to full separability, $k=n-1$ to biseparability and $k=n$ to genuine multipartite entanglement (GME).  Similarly, a mixed state $\rho$ is called $k$-producible if it can be written as a convex combination of pure $k$-producible states. We say that the smallest $k$ for which the state is $k$-producible is the depth of the state's entanglement \cite{Guhne_2009,Seevinck_2001,Guhne_2005,Guhne_2006}.

Consider now that a state is shown to have a maximal entanglement depth, i.e.~it is genuinely multipartite entangled. Since each particle has a dimension larger than two, it is relevant to ask whether the GME itself represents a genuinely high-dimensional phenomenon or whether it can be realised using fewer levels. This notion is captured by the GME-dimension \cite{Cobucci2024}. Specifically, let $\{I,\bar{I}\}$ be a partition of the $n$ particles and write
\begin{equation}
	\rho=\sum_{I| \bar{I}} q_{I|\bar{I}} \sigma_{I|\bar{I}},
\end{equation}
where $\{q_{I|\bar{I}}\}$ is a probability distribution and $ \sigma_{I|\bar{I}}$ is a state whose Schmidt number across the bisection defined by $\{I|\bar{I}\}$ is no larger than $r$. If such a decomposition of $\rho$ exists then $\rho$ can be realised by mixing states whose bipartite entanglement dimension is no larger than $r$. The case of $r=1$ reduces to standard GME. We call the smallest $r$ for which such a decomposition exists the GME-dimension of $\rho$ and denote it as $\mathcal{D}_{\text{GME}}(\rho)$. For fixed $(n,d)$, we define as $\mathcal{G}_{r}$ the set of states for which $\mathcal{D}_{\text{GME}}(\rho)\leq r$.

\subsection{Robustness of GME-dimension}
Assume that a state $\rho \notin \mathcal{G}_{r}$, i.e. $\mathcal{D}_{\text{GME}}(\rho) > r$. A natural way to quantify the robustness of its GME-dimension is through the application of the depolarising channel $\Phi_{v}(\rho) = v\rho + \frac{1-v}{d^n}\openone$ to determine the largest visibility, $v \in [0,1]$, for which $\Phi_{v}(\rho) \in \mathcal{G}_{r}$. This corresponds to the largest amount of noise that the state can resist before losing entanglement among some of its levels. Computing this value for a given state is in general not an easy task. However, upper bounds on it can be found using the following semidefinite programming (SDP) relaxation \cite{Cobucci2024}:
\begin{equation}\label{SDP_main}
	\begin{aligned}
		\max_{v,\tilde{\sigma}} & \quad v\\
		\text{s.t.}&\quad v\rho+\frac{1-v}{d^{n}}\openone =  \sum_{ I|\bar{I}} \tilde{\sigma}_{I|\bar{I}},\\
		& \quad \tilde{\sigma}_{I|\bar{I}} \succeq 0, \quad \forall (I|\bar{I}),\\
		& \quad (\Lambda_I\otimes \openone_{\bar{I}})[\tilde{\sigma}_{I|\bar{I}}] \succeq 0, \quad \forall (I|\bar{I}),
	\end{aligned}
\end{equation}
where $\tilde{\sigma}_{I|\bar{I}}$ are unnormalised states, and $\Lambda$ is an $r$-positive trace-preserving map, i.e. it is positive when applied to a bipartite state with Schmidt number at most $r$. In particular, we choose $\Lambda(X) = \tr(X)\openone - \frac{1}{r}X$, which, however, may also be positive for some states with larger Schmidt number \cite{Terhal_2000,Tomiyama1985OnTG}. Therefore, depending on whether it is applied on subsystem $I$ or on subsystem $\bar{I}$, this map generates different relaxations of $\mathcal{G}_{r}$. A more accurate relaxation can be obtained by applying it on both subsystems.

Since the dimension of the global Hilbert space scales as $d^n$, for $n=d$, the program quickly becomes computationally intensive, allowing evaluations only up to $n=4$ on a standard machine. Due to the highly symmetric structure of the state, this issue can be partially solved when evaluating \eqref{SDP_main} for the supersinglet. Indeed, since the state is permutationally invariant, only $\left\lfloor \frac{n}{2}\right\rfloor$ of all possible $\tilde{\sigma}_{I|\bar{I}}$ are independent, while the rest can be obtained as permutations of the former. Let $\sigma_{m}$ for $m = 1,\dots,\left\lfloor \frac{n}{2}\right\rfloor$ denote the independent operators. Then, we label as $\lbrace g_{I} \rbrace$ the permutations that allow to recover the rest of the operators as $\tilde{\sigma}_{I|\bar{I}} = P_{g_{I}} \sigma_m P_{g_{I}}^{\dagger}$. Additionally, since the supersinglet is invariant under local unitaries, we can use Schur-Weyl duality and Schur's lemma \cite{fulton1991representation} to express each $\sigma_m$ as a linear combination with real coefficients $\lbrace x_{\pi}^{(m)}\rbrace$ of components generated from a direct sum of some matrices $\rho_{\lambda}(\pi)$. Specifically, $\rho_{\lambda}(\pi)$ denotes the matrix representation of the permutation group in a specific irreducible representation of the group $SU(n)$, labeled by $\lambda$. Using these symmetries, the SDP in \eqref{SDP_main} reads as:
\begin{equation}\label{SDP_symm_main}
	\small
	\begin{aligned}
		\max_{v,x} & \quad v\\
		\text{s.t.}&\quad \frac{v}{n!}\text{sgn}(\pi) + \frac{1-v}{d^n} \delta_{e,\pi} = \sum_{m=1}^{\floor{\frac{n}{2}}} \sum_{g_{I}} x_{g_{I}^{-1} \circ \pi \circ g_{I}}^{(m)}, \quad \forall \pi \in \mathcal{S}_n,\\
		& \quad \sum_{\pi \in \mathcal{S}_n} x_{\pi}^{(m)}\rho_{\lambda}(\pi) \succeq 0, \quad \forall m,\lambda,\\
		& \quad \sum_{\pi \in \mathcal{S}_n} x_{\pi}^{(m)} \left[n^{c_{I}(\pi)}\rho_{\lambda}(\tilde{\pi}) - \frac{1}{r}\rho_{\lambda}(\pi) \right] \succeq 0, \quad \forall m,\lambda,
	\end{aligned}
\end{equation}
where $\pi$ denotes a generic permutation of the symmetric group of degree $n$, $m$ labels the $\left\lfloor \frac{n}{2}\right\rfloor$ independent bipartitions, and $c_{I}(\pi)$ is the number of full cycles from the cyclic decomposition of $\pi$ contained in the subset $I$. The full derivation of \eqref{SDP_symm_main} can be found in Appendix \ref{SDP_symmetr_app}, while our implementation is available at \cite{CobucciSupersingletGME}.

In this formulation, the number of independent variables in the problem scales as $\mathcal{O}(n (n!))$, allowing evaluations up to $n = 7$ on a standard machine (see Appendix \ref{SDP_symmetr_app}). This goes beyond the current capabilities in the analysis of GME-dimension robustness for other relevant classes of states. Indeed, while similar symmetry arguments can also be used for GHZ and cluster states, two of the most relevant classes of high-dimensional multiparticle states, they only allow the program in \eqref{SDP_main} to be evaluated up to $(n,d) = (5,5)$ for GHZ and $(n,d) = (4,4)$ for cluster states \cite{Cobucci2024}. In Table \ref{table:visibility_tab}, we compared the results for GHZ, cluster and supersinglet state in the case of $r = n-1$. They were all obtained by applying the generalised reduction map on both subsystems for each bipartition. Although no definitive conclusions can be drawn from these data, as they relate to a relaxation of the set $\mathcal{G}_{r}$, they nevertheless provide an interesting and practically useful insight to the noise robustness of the GME-dimension in these states. Interestingly, the supersinglet appears to exhibit remarkable noise resistance, as the upper bounds on its critical visibility are always lower than those of the GHZ and the cluster state.


\begin{table}[t!]
	\begin{tabular}{|c|c|c|c|c|}
		\hline
		$\, n \,$ & $\, r \,$ & $\,$  $v_{S_n}$  $\,$ & $\,$  $v_{\text{ghz}}$ $\,$ & $v_{\text{cluster}}$  \\
		\hline
		3 & 2 & $0.5345$ & $0.5909$ & $*$\\
		\hline
		4 & 3 & $0.4854$ & $0.6503$ & $0.534$\\
		\hline
		5 & 4 & $0.2668$ & $0.5364$ & $-$ \\
		\hline
		6 & 5 & $0.1338$ & $-$ & $-$ \\
		 \hline
		7 & 6 & $0.0470$ & $-$ & $-$ \\
		\hline
	\end{tabular}
	\caption{Upper bounds on the critical visibility $v$ for supersinglet, GHZ and cluster state, using the SDP relaxation in \eqref{SDP_main} and its symmetrised version in \eqref{SDP_symm_main} for $n=d \leq 7$ and $r = n-1$. The "$*$" for $d=3$ denotes that GHZ and cluster states are local-unitary equivalent, while "$-$" indicates the cases in which it was not possible to evaluate the program on a standard machine.}
	\label{table:visibility_tab}
\end{table}


\section{Entanglement witnesses}
In the previous section, we analysed the robustness of entanglement properties in supersinglets  using full knowledge of the density matrix. However, practical detection of entanglement, commonly achieved through witness methods, should be based on much more modest measurement resources. Here, we construct  witnesses of relevant entanglement properties tailored for supersinglets. We discuss witnesses based on randomised measurements and then  witnesses based on collective spin measurements. These witnesses detect both the depth of entanglement and the GME-dimension.

It is relevant to note that a simple witness for the GME-dimension can be constructed based on the fidelity between the lab state $\rho$ and the target supersinglet $\ket{S_n}$. It is known that for any target state \cite{Cobucci2024}
\begin{equation}\label{fidcriterion}
	F(\psi,\rho)\equiv \bracket{\psi}{\rho}{\psi}\leq \max_{I|\bar{I}} \sum_{i=1}^r \lambda_i(\psi_I),
\end{equation}
where $\lambda_i(\psi_I)$ are the decreasingly ordered eigenvalues of the reduced state $\psi_I=\tr_{\bar{I}} (\psi)$. Selecting $\ket{\psi}=\ket{S_n}$, it follows from Property~\ref{Property2} that $\sum_{i=1}^r \lambda_i=r/\binom{n}{|I|}$. This is maximised when choosing $|I|=1$. Thus, the fidelity witness is 
\begin{equation}\label{FidelityWitness}
	F(\psi,\rho) \leq \frac{r}{d}.
\end{equation} 
In fact, Ref~\cite{Cobucci2024} pointed out that this criterion, when applied to $n$-partite $d$-dimensional GHZ states, linear cluster states and AME states (should they exist) all give the same fidelity bound as above. Now, also supersinglets are found to admit the same delimitation. It means that this method is rather crude, in the sense that it ignores the different entanglement properties of these states.

\subsection{Randomised measurements}\label{random_witness}
By definition, supersinglets always yields distinct outcomes when all particles are measured in the same basis. Therefore, a seemingly reasonable ansatz for an  entanglement witness is to consider such anticorrelated outcomes for every possible basis. This corresponds to the observable
	\begin{align}\nonumber \label{WitnessRandomised}
&W=\int dU \ U^{\otimes n} D \ (U^\dagger)^{\otimes n}, \\
&\text{where} \quad D= \sum_{k_1,\ldots, k_n=0}^{n-1}\ketbra{k_1}\otimes\ldots \otimes \ketbra{k_n}\delta_{\vec{k}},
\end{align}
where $\delta_{\vec{k}}=1$ if all elements in $\vec{k}=(k_1,\ldots,k_n)$ are different and $\delta_{\vec{k}}=0$ otherwise. The integral over the unitaries is normalised. Hence, by construction, we have
$\expect{W}_{S_n}=1$. We now show how the observed value of $W$ reveals the GME-dimension of the state.

\begin{result}
	Every $n$-partite $n$-dimensional state $\rho$ with GME-dimension at most $r$ satisfies
	\begin{equation}\label{witnessRandomised}
		\expect{W}_\rho \leq \frac{r+1}{n+1}
	\end{equation}
\end{result} 
\begin{proof}
	
As we prove in Appendix \ref{App_GME-dimension_witness_bound}, $W$ can be bound by
\begin{equation}
	W \preceq S_n + \frac{\openone-S_n}{n+1}.
\end{equation}
Once this is achieved, let  $\psi$ be any pure state of GME-dimension  $r$. Using the fidelity relation Eq~\eqref{FidelityWitness} for supersinglets, it follows that $\bracket{\psi}{W}{\psi} \leq \frac{r+1}{n+1}$. By convexity, the largest value of $\expect{W}_\rho$ over any state $\rho$ with GME-dimension $r$ is achieved by a pure state. Hence, $\expect{W}_\rho \leq \expect{W}_\psi$, which gives the result in Eq~\eqref{witnessRandomised}.
\end{proof}

It is interesting to analyse the detection efficiency of this witness considering the mixture of supersinglet with white noise. Using the bound in \eqref{witnessRandomised}, the state is detected to have a GME-dimension larger than $r$ whenever
\begin{equation}
	v > v_r:=\frac{n^n (r+1)-(n+1)!}{(n+1)(n^n - n!)}.
\end{equation}
Notably, in the limit of large $n$, using Stirling's approximation of the factorial, $n!\approx \sqrt{2\pi n}\left(n/e \right)^n$, to simplify the bound, we get $v\gtrsim r/n$. This is essentially the same result obtained from the fidelity criterion in Eq~\eqref{FidelityWitness}.

While an exact evaluation appears to require measurements in every possible basis, the result can be effectively approximated by randomly sampling from $SU(n)$. Specifically, we can draw a list of $N$ independent Haar random unitaries and define the associated witness
\begin{equation}\label{eq:fs_list}
	W_N=\frac1N\sum_{i=1}^N U_i^{\otimes n}D(U_i^\dagger)^{\otimes n}.
\end{equation}
Then, we measure the multipartite state $\rho$ in one of the bases, chosen uniformly and independently from the list, and we check whether the outcomes of all particles are distinct. Denoting by $T$ the number of times we repeat this operation and by $K$ the number of trials where the distinct-outcomes condition occurs, the observed witness value is $W_{N}^{\text{obs}} = \frac{K}{T}$. In Appendix \ref{app:finite_settings}, we show how to use finite statistics to bound the probability that, under the null hypothesis $\rho \in \mathcal{G}_r$, the experiment produces an observed witness value $W_{N}^{\text{obs}}$ at least as large as the bound in \eqref{witnessRandomised}, taking into account both the finite sampling of measurement settings $N$ and the finite number of trials $T$ used to estimate the expectation value $\expect{W_N}$. For example, for $n=8$ and an experiment with $N \approx 50$ measurement settings which violates \eqref{witnessRandomised} by a value $\delta \approx 0.07$, a number of trials $T \approx 10^4$ would certify GME-dimension $r > 7$ with a $\text{p-value} \lesssim 10^{-4}$.

\subsection{Collective spin measurements}
While the randomised measurement witness requires access to potentially any basis, we now construct a much simpler witness that only requires the ability to perform collective spin measurements on all the particles. To this end, let $\{\lambda_j\}_{j=1}^{n^2-1}$ be a basis of traceless Hermitian  $n\times n $ matrices. These can for instance be chosen as the Gell-Mann matrices, satisfying $\tr(\lambda_j\lambda_k)=2\delta_{jk}$. For simplicity, write $\lambda_j^{(i)}=\mathds{1}\otimes \mathds{1}\otimes \lambda_j\otimes \mathds{1}\ldots\mathds{1}$, where $\lambda_j$ appears in position $i$. The collective spin observable in the $j$'th direction is then written $R_j=\sum_{i=1}^n \lambda_j^{(i)}$. The total spin corresponds to the Casimir operator
\begin{equation}
	\mathcal{C}=\sum_{j=1}^{n^2-1} R_j^2.
\end{equation}
We can use the expectation value of $\mathcal{C}$ as a witness to detect entanglement properties ideally tailored for states similar to the supersinglet. The reason for this is that the supersinglet is the unique eigenvector of the Casimir operator corresponding to its zero-eigenvalue. In other words, it is the unique total spin zero state.

\begin{lemma}\label{lemma1}
	The Casimir operator $\mathcal{C}$ has a single zero-mode which corresponds to the supersinglet, i.e. $\mathcal{C}\ket{S_n}=0$.
\end{lemma}
\begin{proof}
	Since the operators $\lbrace R_{j} \rbrace_{j=1}^{n^2-1}$ are a valid set of generators for the subgroup $\lbrace U^{\otimes n}: \, U \in SU(n)\rbrace$, any such unitary can be expressed as $U^{\otimes n} = e^{-i\sum_{j}R_{j}\theta_j}$. Applying the collective unitary on the supersinglet state $\ket{S_n}$ and differentiating with respect to $\theta_j$, we get that $R_{j}\ket{S_n} = 0$, for each $j$, from which the lemma follows. See Appendix~\ref{app:uniqueground} for further details.
\end{proof}

This lemma strongly suggests that the measured value of $\mathcal{C}$ can be used to detect entanglement properties for states in the vicinity of the supersinglet. As we now show, this intuition can be formalised into a criterion for the depth of the multipartite entanglement.

\begin{result}[Entanglement depth from collective spin measurements]
Let $\rho$ be a $k$-producible state of $n$ 
particles with $n$ local levels each. It then holds that
\begin{equation}\label{CasimirDepth}
	\expect{\mathcal{C}}_\rho \geq 2\left(n^2-\left\lfloor \frac{n}{k}\right\rfloor k^2-s^2\right),
\end{equation}
where $s=n \mod{k}$. Thus, a violation certifies an entanglement depth of at least $k+1$. The inequality is tight. 
\end{result}
\begin{proof}
Here we provide a sketch of the proof (see Appendix \ref{entdeptwit_app} for further details). Using Schur-Weyl duality and Schur's lemma \cite{fulton1991representation}, we can express the Casimir operator as $\mathcal{C} = 2n(n-1)\openone + 4 \sum_{i < j} \mathbf{S}_{ij}$, where $\mathbf{S}_{ij}$ acts as the SWAP operator between subsystems $i$ and $j$, and as the identity on the remaining $(n-2)$ subsystems. Restricting to pure $k$-producible states, $\ket{\psi} = \bigotimes_{p} \ket{\phi}_{S_{p}}$ with $|S_{p}| \leq k$, we can bound $\sum_{i < j}\expect{\mathbf{S}_{ij}} \geq 2(n^2 - \sum_{p}|S_{p}|^{2})$. Then, the result in \eqref{CasimirDepth} follows from maximising the sum of the squared lengths of each partition.
\end{proof}

Consider the strongest case, in which we set $k=n-1$. Then, Eq~\eqref{CasimirDepth} reduces to $\expect{\mathcal{C}}\geq 4(n-1)$ and its violation becomes a certificate of GME. However, this does not reveal whether the GME is a genuinely high-dimensional phenomenon. It turns out that the value of the Casimir operator can be used to upgrade this GME-criterion into a criterion for the GME-dimension of the state $\rho$.
\begin{result}[GME-dimension from collective spin measurements]\label{GME-dimension_witness}
Let $\rho$ be a state of $n$ particles with $n$ local levels each, and with a GME-dimension no larger than $r$. It then holds that 
\begin{equation}\label{GMEdimCasimir}
	\expect{\mathcal{C}}_\rho \geq 4\left(n-r\right).
\end{equation}
Thus, a violation certifies a GME-dimension of at least $r+1$. The inequality is tight.
\end{result}
\begin{proof}
The result follows directly if we assume that
\begin{equation}\label{Casimir_bound_GME-dim}
\mathcal{C}\succeq 4n\left(\mathds{1}-S_n\right).
\end{equation}
Indeed, multiplying both sides with a state $\rho$ of GME-dimension $r$ and taking the trace gives  $F\equiv \bracket{S_n}{\rho}{S_n}\geq 1-\frac{\expect{\mathcal{C}}}{4n}$. Using the fidelity bound from Eq~\eqref{FidelityWitness}, we get $\frac{r}{n}\geq 1-\frac{\expect{\mathcal{C}}}{4n}$ which can be rearranged into Eq~\eqref{GMEdimCasimir}.
Now, we prove the assumption \eqref{Casimir_bound_GME-dim}. Here, we only sketch the proof and refer to Appendix \ref{App_GME-dim_Casimir} for further details. Using Schur-Weyl duality and Schur's lemma, the Casimir operator can be decomposed as $\mathcal{C} = \bigoplus_{\lambda \vdash n} c_{\lambda} I_{\lambda}$, where $\lambda$ denote the integer partitions of $n$ with length $l(\lambda) \leq n$ and $I_{\lambda}$ is the identity on the Schur-Weyl sector. The coefficients $c_{\lambda}$ for each representation are known \cite{mkrtchyan2025casimireigenvaluesadotimesk,GotoGrosshans1978,DiFrancescoMathieuSenechal1997}. By expressing them in terms of the geometry of the Young tableaux, we get the full spectrum of the Casimir operator.
\end{proof}

\section{Supersinglets for gradient metrology}

In this section, we consider a natural application of the supersinglet state: the estimation of spatial gradients. Spatially varying fields arise naturally in magnetometry, gradiometry, imaging and distributed sensing, where the quantity of interest is often not the absolute field strength but its variation across a spatially extended probe. From a metrological perspective, gradient estimation is particularly appealing since a homogeneous contribution to the field may be regarded as a nuisance parameter or common-mode fluctuation. Quantum strategies for gradient sensing have therefore received considerable attention across several platforms, to measure gradients ranging from magnetic fields \cite{UrizarLanz2013,Altenburg2017,Apellaniz_2018,Schaffner2024,Newman2024} to gravity \cite{Stray2025,Thomas_2025}. In particular, multipartite singlet states have been proposed for gradient magnetometry because they are insensitive to homogeneous magnetic fields while remaining sensitive to spatial variations \cite{UrizarLanz2013}. More general treatments have derived quantum-Fisher-information bounds for gradient estimation, including arbitrary-spin atomic ensembles \cite{Apellaniz_2018}, identified optimal entangled and separable probes and decoherence-free strategies against fluctuations of the homogeneous field \cite{Altenburg2017}, and demonstrated that suitably chosen multipartite states can in principle attain Heisenberg scaling \cite{Zhang2014}.

Most theoretical works, however, have largely considered magnetic fields coupled through a fixed spin component, so that the internal degrees of freedom are governed by an $SU(2)$ spin representation. Even when higher-spin particles are considered, the increased local dimension enters through the spin-$j$ representation of $SU(2)$ rather than through the full $SU(d)$ structure of a $d$-level system. This leaves open how well $SU(d)$-invariant states perform as probes of spatial variations, and what role the local dimension and the high-dimensional entanglement of such states play in the attainable precision. Supersinglets are by Definition \ref{supersinglet_definition}, invariant under the application of $U^{\otimes n}$, $\forall\,\,U\in SU(d)$, and are therefore naturally suited for gradient metrology. We note that supersinglets are only a subclass of states with this property. However, the optimal state for gradient metrology under relevant considerations (see Results \ref{th:op1} and \ref{th:2}) turns out to be a supersinglet, as we will see below.

Consider a field $f(x) = f_0 + g x$, where $f_0$ is a constant, $x$ is the spatial position, and $g$ is some gradient that needs to be estimated. In an experimental scenario, $f_0$ may represent a large background field to be implemented, while the gradient may be a small error that arises due to device imperfections. Consider the same observable $\xi$ for each particle $i$ and let its interaction with the field be captured by the Hamiltonian,
\begin{align}
	H = \sum_{i=1}^d (f_0+gx_i)\xi^{(i)} \equiv H_0 + H_g,
\end{align}where $H_0$ and $H_g$ are the Hamiltonians due to the background field and gradient field, respectively. Without loss of generality, the observables can be chosen to be traceless, i.e.~$\tr(\xi)=0$, as a non-zero trace only leads to a constant shift in the Hamiltonian and an irrelevant global phase to the state.  The unitary corresponding to the field Hamiltonian can be written as 
\begin{align}
e^{-it H} = e^{-itf_0\sum_i\xi^{(i)}} e^{-itg\sum_i x_i\xi^{(i)}} \equiv U_0^{\otimes n} U_g,
\end{align}with $U_0 =e^{-itf_0\xi}$ and $U_g = e^{-itg\sum_i x_i\xi^{(i)}}$. The estimation probe consists of $n$ qudits (such that $n=d$)  placed on a 1-dimensional chain at spatial locations $x_i$ ($i=1,\dots,d$). We let the probe be initially in a supersinglet state $\ket{S_n}$. 

Since $U_0$ acts uniformly on all qudits, it leaves the supersinglet invariant. So we have that $U^{\otimes n}_0U_g\ket{S_n}=U_g\ket{S_n}=e^{-itH_g}\ket{S_n}$. Therefore, the supersinglet state is invisible to the homogeneous field and is only acted upon non-trivially by the gradient field that is to be estimated. Importantly, the supersinglet state is the unique pure state with this property, i.e., it is the only state which rejects background noise, while allowing to exclusively estimate  the gradient.  For specific background Hamiltonians, one may find other states which are invisible under their action. However, the supersinglet is the unique state which is invisible to all background Hamiltonians.

The encoded state is given by $\ket{\psi_g} = U_g\ket{S_n}$. We define the rescaled parameter $\theta = gt$ and Hamiltonian $H'=H_g/g$. 

\begin{result}[Quantum Fisher information (QFI) for gradient metrology with a single supersinglet state] The QFI for estimating $\theta$ with $\ket{S_d}$ is given by 
\begin{align}
	F_Q[\ket{S_d},H'] = 4 \Delta H'^2 = 4 \frac{d}{d-1}\tr(\xi^2)\sigma_x^2,
\end{align}where $\sigma_x^2 = \frac{1}{n}\sum_i (x_i-\bar x)^2$ and $\bar x = \frac{1}{n}\sum_i x_i$. \label{res:sing}
\end{result}
For equally-spaced particles with distance $x_0$, the spatial variance is $\sigma_x^2 = (d^2-1)x_0^2/12$, giving
\begin{align}
	F_Q = \frac{1}{3}d(d+1)\tr(\xi^2)x_0^2 = \frac{1}{3} \frac{d(d+1)}{(d-1)^2}\tr(\xi^2)L^2,
\end{align}with fixed length of the chain $L=(d-1)x_0$. Therefore, for large $d$,  and apart from possibly entering in the trace, the dimensional factor disappears and $F_Q$ grows only with the size of the chain as $\sim L^2$. The proof of Result \ref{res:sing} can be found in Appendix~\ref{app:metsingle}.

\subsection{Macroscopic supersinglet states}

\begin{figure}
	\centering
	\includegraphics[width =0.6 \columnwidth]{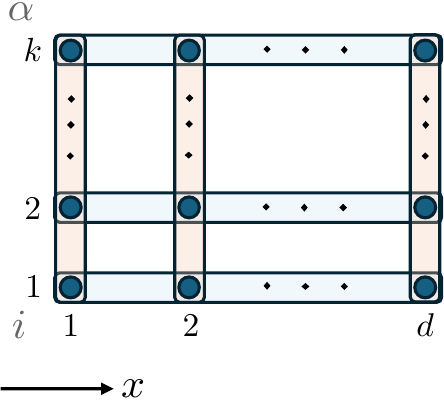}
	\caption{A 1D macroscopic probe for gradient metrology with $k$ qudits at each of the $d$ spatial sites. $d$ qudits in every layer (blue boxes) are in supersinglet states. The symmetrised supersinglet state is constructed by symmetrising $k$ qudits at every site (orange boxes). }
	\label{fig:probe}
\end{figure}

Let us now consider this estimation problem but with  the probe consisting of $n=kd$ qudits, as shown in Fig. \ref{fig:probe}. Therefore, we consider a layered construction with $d$ spatial sites and $k$ qudits at every site, indexing sites with $i=1,\dots,d$ and particles with $\alpha=1,\dots,k$. The background and gradient Hamiltonians can now be expressed as 
\begin{align}\label{eq:hp1}
H_0 = \sum_{i=1}^{d}\sum_{\alpha=1}^k f_0 \xi^{(i,\alpha)},\, H' = \sum_{i=1}^{d}\sum_{\alpha=1}^k x_i \xi^{(i,\alpha)},
\end{align}with $H_g = gH'$, as defined previously.
Let us first consider a tensor product over supersinglet layers, i.e., a $d-$particle supersinglet state at every layer, 
\begin{align}
	\ket{S_d^k } = \bigotimes_{\alpha=1}^k \ket{S_d}.
\end{align}As the QFI is additive under tensor products, 
\begin{align}
	F_Q [\ket{S_d^k },H'] = kF_Q[\ket{S_d},H']=4\frac{kd}{(d-1)} \tr(\xi^2) \sigma_x^2.
\end{align}Therefore, for fixed $d$, the QFI scales linearly in the number of layers $k$ or the number of particles $n$, which corresponds to a shot-noise scaling. This is expected; while the layers (i.e., the supersinglet blocks) themselves are entangled, there is no entanglement between different layers, which leads to a linear scaling.

To achieve better than linear scaling, we must create entanglement between the layers. To this end, let us now consider symmetrisation along the sites. At site $i$, let $V_{\pi}^{(i)}$ permute $k$ layers and define 
\begin{align}
	P_i  = \frac{1}{k!} \sum_{\pi \in S_k} V_\pi^{(i)}, \quad \text{and}\quad P = \bigotimes_{i=1}^dP_i.
\end{align}Therefore, each operator $P_i$ projects all $k$ particles at the corresponding site to their completely symmetric subspace. The collective operator $P$ applies these projections at all sites. We now consider the state
\begin{align}
	 \ket{\mathds S_d^k}  = \frac{P\ket{S_d^k}}{\sqrt{\bra{S_d^k}P\ket{S_d^k}}},
\end{align}and refer to it as the symmetrised supersinglet state. The GME-dimension of this state is exactly $d$ (see Appendix \ref{app:GME-dim_macroscopsupersinglet}). We note that it has already been introduced \cite{Arovas2008} in condensed matter theory, although without noting its entanglement properties.
\begin{result}[Heisenberg scaling with the symmetrised supersinglet state] The QFI for estimating $\theta$ with $\ket{\mathds S_d^k}$ is given by
	\begin{align}\label{eq:fqsss}
		F_Q[\ket{\mathds S_d^k},H'] = 4\frac{kd(d+k)}{(d^2-1)}\tr(\xi^2)\sigma_x^2. 
	\end{align}
\end{result}See App. \ref{app:symsing} for a proof of the above. Therefore, for fixed $d$, the QFI scales as $\sim k^2$ or equivalently $\sim n^2$, corresponding to a Heisenberg scaling.

\subsection{Optimal states for gradient metrology}

We now ask whether $\ket{\mathds S_d^k}$ is optimal for gradient metrology among global SU($d$) singlets, i.e., among states invisible to the action of homogeneous fields. In this regard, it is important to note the role of the spatial configuration. Some spatial configurations will clearly be better than others and in general, the optimal probe should be tailored to it. However, the symmetrised supersinglet state is the optimal pure state for gradient metrology for any spatial configuration.

\begin{result}\label{th:op1}
Let $\ket{\psi}$ be an arbitrary supersinglet and $H'$ be given by Eq. \eqref{eq:hp1}. Then, the following relation is satisfied,
\begin{equation}
	F_Q(\ket{\psi},H'(\boldsymbol x))
	\leq
4	\frac{kd(d+k)}{d^2-1}\tr(\xi^2)\sigma_x^2, \quad \forall\,\boldsymbol x \in \mathds R^d .
	\label{eq:pure_optimal_bound}
\end{equation}
The state $\ket{\mathds S_d^k}$ saturates the bound for every spatial profile
$\{x_i\}_i$ (Eq. \eqref{eq:fqsss}).
\end{result}The proof can be found in App. \ref{app:th1}. Theorem \ref{th:op1} holds for any given spatial configuration. Therefore, the symmetrised supersinglet state can be seen as a general optimal pure state for gradient metrology. An intuitive reason behind this observation is that the QFI $F_Q = 4\Delta H'^2$; this depends directly on the fluctuations of the local collective observables $L_i = \sum_\alpha \xi^{(i,\alpha)}$, through $\expval{L_i^2}$. The symmetrised supersinglet state, by virtue of symmetrisation, acts as coherently as possible at every site and maximises $\expval{L_i^2}$, while keeping the correct anticorrelations between the sites to ensure homogeneous-field invisibility. Let us finally emphasise that the theorem does not claim that $\ket{\mathds S_d^k}$ is the unique pure state that saturates the upper bound. For specific spatial configurations, there may be other states that perform equally well.
  
  Next, we consider the role of the spatial configuration explicitly. Let us consider the quantity $F_Q[\rho,H'(\boldsymbol{x})]/\sigma_x^2$. As $H'$ is linear in the positions, this is the natural normalisation of the QFI in a scenario where the spatial configuration $\boldsymbol{x}$ is allowed to vary freely up to an overall scale. The normalised QFI can be interpreted as the amount of quantum Fisher information a quantum state $\rho$ generates per unit spatial variance of the array. In other words, it removes the trivial enhancement in the QFI  obtained simply by tuning the spatial extent of the array. We refer to the spatial configuration that minimises the normalised QFI as "worst-case". The corresponding sensitivity generated by a state $\rho$ is then
  \begin{align}
  	\min_{\boldsymbol{x};\sigma_x>0} \frac{F_Q[\rho,H'(\boldsymbol{x})]}{\sigma_x^2}
  \end{align}This brings us to the second theorem: the symmetrised supersinglet state is optimal for gradient metrology for worst-case spatial configurations. 
 \begin{result}\label{th:2}
Let $n=kd$ and $\rho$ be an arbitrary state satisfying $U^{\otimes n}\rho (U^{\otimes n})^\dagger=\rho$
$\forall \, U\in SU(d)$. For any operator $H'$ given by Eq. \eqref{eq:hp1}, the following holds 
\begin{align}
	\max_{\rho}
	\;
	\min_{\boldsymbol{x};\sigma_x>0}
	\frac{F_Q[\rho,H'(\boldsymbol{x})]}{\sigma_x^2}
	=
	4\frac{kd(d+k)}{(d^2-1)}\tr(\xi^2).
	\label{eq:robust_optimality}
\end{align}$\rho=\ketbra{\mathds S_d^k}{\mathds S_d^k}$ is the unique state attaining this bound.
 \end{result}The proof can be found in App. \ref{app:th2}. Note that the result does not assume that $\rho$ is pure and therefore a supersinglet; yet the optimal state turns out to be one.  Intuitively, this result is based on two properties of the symmetrised supersinglet state. First, as described above, it acts as coherently as possible at all sites, maximizing local fluctuations. And second, as it is invariant under site permutations, these gradient-sensitive fluctuations are equally distributed (up to $\sigma_x^2$) for all non-uniform spatial configurations. Crucially, these properties make it the unique optimal state for worst-case configurations.

\subsection{Optimal measurement for gradient metrology}

We have obtained above the QFI for gradient metrology with supersinglet and macroscopic supersinglet states.  Now, we comment upon the corresponding optimal measurements. In full generality, the optimal measurement corresponds to one in the eigenbasis of the symmetric logarithmic derivative (SLD) \cite{Paris2009}, which can be complicated and experimentally inaccessible. Moreover, even though the QFI  is independent of the true value of the estimation parameter, the optimal measurement, in general, depends strongly on it. Another option, for locally optimal parameter estimation with pure states, is the survival measurement \cite{Tommaso2016}. However, in our case, this would require resolving a particular supersinglet measurement, which may be experimentally impractical. An elegant way around this in our case, is a binary Casimir measurement. Consider the Casimir $\mathcal C = \sum_{j=1}^{d^2-1}R_j^2$, with $R_j = \sum_{i=1}^n\lambda_j^{(i)}$. Let $P_0$ be the projector onto its zero eigenspace, i.e., the global $SU(d)$ singlet sector. Then, the following theorem holds
\begin{result}
Let $\ket{\psi}$ be an arbitrary supersinglet. Then, the binary Casimir measurement $\{P_0,\mathds 1-P_0\}$ is locally optimal at $\theta=0$,
\begin{align}
	\lim_{\theta\to0} F_\text{Cas} = F_Q[\ket{\psi},H'].
\end{align}
\end{result}For mixed states, a corresponding result holds if the state is entirely supported on the global singlet sector, i.e., if $P_0\rho P_0=\rho$. The proof can be found in App. \ref{app:optmeas}. The theorem considers an initial state which is a global $SU(d)$ singlet, i.e, $P_0 \ket{\psi}=\ket{\psi}$. A small, spatially non-uniform field drives the state out of the singlet sector. The binary Casimir measurement, which measures the probability that the perturbed state is in the singlet sector, allows reaching the QFI exactly in the limit $\theta\to0$.  Therefore, the optimal sensitivity can be reached through a symmetry-sector measurement rather than a general, parameter-dependent SLD basis measurement. For qubits, such  a measurement is simply the measurement of total angular momentum of the state \cite{Apellaniz_2018}.

\section{Conclusion}

In this work, we have investigated supersinglets as an outstanding form of high-dimensional multipartite entangled states. By uncovering their fundamental structural properties, we have highlighted how their defining symmetry gives rise to distinctive entanglement features. In particular, we found their genuinely high-dimensional multipartite entanglement to be remarkably robust to noise. Drawing directly on these symmetries, we developed simple, general, and experimentally favourable witnesses capable of certifying both the depth and dimensionality of multipartite entanglement.

Beyond their remarkable entanglement structure, supersinglets also emerge naturally as resources for gradient sensing. Their defining invariance under collective unitaries renders them insensitive to homogeneous fields while retaining sensitivity to spatial variations. This makes supersinglets fundamental building blocks for constructing states that attain Heisenberg scaling and are optimal for gradient metrology under broad conditions. Moreover, the optimal states are both genuinely multipartite and genuinely high-dimensionally entangled, establishing a direct connection between these forms of entanglement and enhanced metrological performance.

Supersinglets have appeared repeatedly, yet somewhat scatteredly, across quantum information theory. Our work aims to finally establish supersinglets a household state in the field. Their combination of robust high-dimensional entanglement and metrological utility demonstrates that their symmetry is not merely elegant, but can be directly exploited for useful quantum-information tasks. We hope that this will motivate a first experimental demonstration of supersinglets in quantum information science, most immediately in trapped-ion systems, where such technology is already within the state-of-the-art \cite{Ringbauer2022}, and potentially also in photonics, where their preparation remains a technological challenge. Looking further ahead, the development of realistic schemes for preparing large supersinglet states would open the door to exploring what further sensing and information-processing applications their distinctive symmetry might enable.

\section*{Acknowledgements}
We warmly thank Ad\'an Cabello for sharing his deep and infectious enthusiasm for supersinglets. We also thank Poya Modaresi for contributions to an early stage of this project. This work is financially supported by the Swedish Foundation for Strategic Research, the Knut and Alice Wallenberg Foundation through the Wallenberg Center for Quantum Technology (WACQT), the Crafoord Foundation, the Krapperup Foundation, the Royal Physiographic Society in Lund  and  the Swedish Research Council under Contract No.~2023-03498. S.K.~acknowledges support from the Swiss National Science Foundation Grant No. P500PT\textunderscore222265.\newline
\textbf{AI contribution ---} The bound in Result 3 was provided to A.T by Claude Opus 4.6, but then derived independently by G.C. The same AI model advised technical questions for the SDP-symmetrisation procedure in section \ref{entanglement_properties}. Results 6 and 7 were proven by ChatGPT 5.6 Sol, and were simplified by G.C. and S.K. Result 8 was discovered by the same model, but proven independently by S.K.

\bibliography{biblio_SS.bib}

@article{Thomas_2025,
	title={Tensor gravity gradiometry with a single-axis atom gradiometer},
	volume={112},
	ISSN={2469-9934},
	url={http://dx.doi.org/10.1103/vqjd-8qwh},
	DOI={10.1103/vqjd-8qwh},
	number={5},
	journal={Physical Review A},
	publisher={American Physical Society (APS)},
	author={Thomas, Ryan J. and Legge, Samuel and Close, John D.},
	year={2025},
	month=Nov }

@article{Schaffner2024,
	title   = {Quantum Sensing in Tweezer Arrays: Optical Magnetometry on an Individual-Atom Sensor Grid},
	author  = {Sch{\"a}ffner, Dominik and Schreiber, Tobias and Lenz, Fabian and
	Schlosser, Malte and Birkl, Gerhard},
	journal = {PRX Quantum},
	volume  = {5},
	pages   = {010311},
	year    = {2024},
	doi     = {10.1103/PRXQuantum.5.010311}
}

@article{Newman2024,
	title   = {Tensor gradiometry with a diamond magnetometer},
	author  = {Newman, A. J. and Graham, S. M. and Edmonds, A. M. and
	Twitchen, D. J. and Markham, M. L. and Morley, G. W.},
	journal = {Phys. Rev. Applied},
	volume  = {21},
	pages   = {014003},
	year    = {2024},
	doi     = {10.1103/PhysRevApplied.21.014003}
}

@article{Stray2025,
	title   = {Quantum gravity gradiometry for future mass change science},
	author  = {Stray, Ben and Bosch-Lluis, Xavier and Thompson, Robert and others},
	journal = {EPJ Quantum Technology},
	volume  = {12},
	pages   = {35},
	year    = {2025},
	doi     = {10.1140/epjqt/s40507-025-00338-1}
}

@article{Altenburg2017,
	title = {Estimation of gradients in quantum metrology},
	author = {Altenburg, Sanah and Oszmaniec, Micha{\l} and
	W{\"o}lk, Sabine and G{\"u}hne, Otfried},
	journal = {Phys. Rev. A},
	volume = {96},
	pages = {042319},
	year = {2017},
	doi = {10.1103/PhysRevA.96.042319}
}

@article{Zhang2014,
	title = {Fitting magnetic field gradient with Heisenberg-scaling accuracy},
	author = {Zhang, Yong-Liang and Wang, Huan and Jing, Li and
	Mu, Liang-Zhu and Fan, Heng},
	journal = {Scientific Reports},
	volume = {4},
	pages = {7390},
	year = {2014},
	doi = {10.1038/srep07390}
}

@article{Tommaso2016,
	title = {Loschmidt echo for quantum metrology},
	author = {Macr\`{\i}, Tommaso and Smerzi, Augusto and Pezz\`e, Luca},
	journal = {Phys. Rev. A},
	volume = {94},
	issue = {1},
	pages = {010102(R)},
	numpages = {6},
	year = {2016},
	month = {Jul},
	publisher = {American Physical Society},
	doi = {10.1103/PhysRevA.94.010102},
	url = {https://link.aps.org/doi/10.1103/PhysRevA.94.010102}
}

@article{Liu_2019,
	title={Quantum Fisher information matrix and multiparameter estimation},
	volume={53},
	ISSN={1751-8121},
	url={http://dx.doi.org/10.1088/1751-8121/ab5d4d},
	DOI={10.1088/1751-8121/ab5d4d},
	number={2},
	journal={Journal of Physics A: Mathematical and Theoretical},
	publisher={IOP Publishing},
	author={Liu, Jing and Yuan, Haidong and Lu, Xiao-Ming and Wang, Xiaoguang},
	year={2019},
	month=Dec, pages={023001} }

@Article{Ringbauer2022,
	author={Ringbauer, Martin
	and Meth, Michael
	and Postler, Lukas
	and Stricker, Roman
	and Blatt, Rainer
	and Schindler, Philipp
	and Monz, Thomas},
	title={A universal qudit quantum processor with trapped ions},
	journal={Nature Physics},
	year={2022},
	month={Sep},
	day={01},
	volume={18},
	number={9},
	pages={1053-1057},
	issn={1745-2481},
	doi={10.1038/s41567-022-01658-0},
	url={https://doi.org/10.1038/s41567-022-01658-0}
}

@article{Paris2009,
	author = {PARIS, MATTEO G. A.},
	title = {QUANTUM ESTIMATION FOR QUANTUM TECHNOLOGY},
	journal = {International Journal of Quantum Information},
	volume = {07},
	number = {supp01},
	pages = {125-137},
	year = {2009},
	doi = {10.1142/S0219749909004839},
	
	URL = { 
	
	https://doi.org/10.1142/S0219749909004839
	
	
	
	},
	eprint = { 
	
	https://doi.org/10.1142/S0219749909004839
	
	
	
	}
}

@article{Arovas2008,
	title = {Simplex solid states of $\mathrm{SU}(N)$ quantum antiferromagnets},
	author = {Arovas, Daniel P.},
	journal = {Phys. Rev. B},
	volume = {77},
	issue = {10},
	pages = {104404},
	numpages = {14},
	year = {2008},
	month = {Mar},
	publisher = {American Physical Society},
	doi = {10.1103/PhysRevB.77.104404},
	url = {https://link.aps.org/doi/10.1103/PhysRevB.77.104404}
}

@article{UrizarLanz2013,
	title={Macroscopic singlet states for t magnetometry},
	volume={88},
	ISSN={1094-1622},
	url={http://dx.doi.org/10.1103/PhysRevA.88.013626},
	DOI={10.1103/physreva.88.013626},
	number={1},
	journal={Physical Review A},
	publisher={American Physical Society (APS)},
	author={Urizar-Lanz, Iñigo and Hyllus, Philipp and Egusquiza, Iñigo Luis and Mitchell, Morgan W. and Tóth, Géza},
	year={2013},
	month=July }

@article{Hillery2001,
	title = {Singlet states and the estimation of eigenstates and eigenvalues of an unknown controlled-$U$ gate},
	author = {Hillery, Mark and Bu\ifmmode \check{z}\else \v{z}\fi{}ek, Vladim\'{\i}r},
	journal = {Phys. Rev. A},
	volume = {64},
	issue = {4},
	pages = {042303},
	numpages = {5},
	year = {2001},
	month = {Sep},
	publisher = {American Physical Society},
	doi = {10.1103/PhysRevA.64.042303},
	url = {https://link.aps.org/doi/10.1103/PhysRevA.64.042303}
}

@Article{Christandl2012,
	author={Christandl, Matthias
	and Schuch, Norbert
	and Winter, Andreas},
	title={Entanglement of the Antisymmetric State},
	journal={Communications in Mathematical Physics},
	year={2012},
	month={Apr},
	day={01},
	volume={311},
	number={2},
	pages={397-422},
	issn={1432-0916},
	doi={10.1007/s00220-012-1446-7},
	url={https://doi.org/10.1007/s00220-012-1446-7}
}

@article{Dur2000,
	title = {Three qubits can be entangled in two inequivalent ways},
	author = {D\"ur, W. and Vidal, G. and Cirac, J. I.},
	journal = {Phys. Rev. A},
	volume = {62},
	issue = {6},
	pages = {062314},
	numpages = {12},
	year = {2000},
	month = {Nov},
	publisher = {American Physical Society},
	doi = {10.1103/PhysRevA.62.062314},
	url = {https://link.aps.org/doi/10.1103/PhysRevA.62.062314}
}

@article{Tavakoli2024,
	title = {Semidefinite programming relaxations for quantum correlations},
	author = {Tavakoli, Armin and Pozas-Kerstjens, Alejandro and Brown, Peter and Ara\'ujo, Mateus},
	journal = {Rev. Mod. Phys.},
	volume = {96},
	issue = {4},
	pages = {045006},
	numpages = {68},
	year = {2024},
	month = {Dec},
	publisher = {American Physical Society},
	doi = {10.1103/RevModPhys.96.045006},
	url = {https://link.aps.org/doi/10.1103/RevModPhys.96.045006}
}

@article{Cobucci2024,
	author = {Gabriele Cobucci  and Armin Tavakoli },
	title = {Detecting the dimensionality of genuine multiparticle entanglement},
	journal = {Science Advances},
	volume = {10},
	number = {38},
	pages = {eadq4467},
	year = {2024},
	doi = {10.1126/sciadv.adq4467},
	URL = {https://www.science.org/doi/abs/10.1126/sciadv.adq4467},
	eprint = {https://www.science.org/doi/pdf/10.1126/sciadv.adq4467}}

@article{Apellaniz_2018,
	title={Precision bounds for gradient magnetometry with atomic ensembles},
	volume={97},
	ISSN={2469-9934},
	url={http://dx.doi.org/10.1103/PhysRevA.97.053603},
	DOI={10.1103/physreva.97.053603},
	number={5},
	journal={Physical Review A},
	publisher={American Physical Society (APS)},
	author={Apellaniz, Iagoba and Urizar-Lanz, Iñigo and Zimborás, Zoltán and Hyllus, Philipp and Tóth, Géza},
	year={2018},
	month=May }

@article{Li2024,
	doi = {10.22331/q-2024-02-29-1270},
	url = {https://doi.org/10.22331/q-2024-02-29-1270},
	title = {Identifying families of multipartite states with non-trivial local entanglement transformations},
	author = {Li, Nicky Kai Hong and Spee, Cornelia and Hebenstreit, Martin and de Vicente, Julio I. and Kraus, Barbara},
	journal = {{Quantum}},
	issn = {2521-327X},
	publisher = {{Verein zur F{\"{o}}rderung des Open Access Publizierens in den Quantenwissenschaften}},
	volume = {8},
	pages = {1270},
	month = feb,
	year = {2024}
}

@article{Gour2017,
	author  = {Gour, Gilad and Kraus, Barbara and Wallach, Nolan R.},
	title   = {Almost all multipartite qubit quantum states have trivial stabilizer},
	journal = {Journal of Mathematical Physics},
	volume  = {58},
	number  = {9},
	pages   = {092204},
	year    = {2017},
	month   = sep,
	doi     = {10.1063/1.5003015},
	eprint  = {1609.01327},
	archivePrefix = {arXiv},
	primaryClass  = {quant-ph},
}

@article{Sauerwein2018,
	title = {Transformations among Pure Multipartite Entangled States via Local Operations are Almost Never Possible},
	author = {Sauerwein, David and Wallach, Nolan R. and Gour, Gilad and Kraus, Barbara},
	journal = {Phys. Rev. X},
	volume = {8},
	issue = {3},
	pages = {031020},
	numpages = {23},
	year = {2018},
	month = {Jul},
	publisher = {American Physical Society},
	doi = {10.1103/PhysRevX.8.031020},
	url = {https://link.aps.org/doi/10.1103/PhysRevX.8.031020}
}

@article{Steinberg2024,
	title = {Finding maximal quantum resources},
	author = {Steinberg, Jonathan and G\"uhne, Otfried},
	journal = {Phys. Rev. A},
	volume = {110},
	issue = {6},
	pages = {062428},
	numpages = {18},
	year = {2024},
	month = {Dec},
	publisher = {American Physical Society},
	doi = {10.1103/PhysRevA.110.062428},
	url = {https://link.aps.org/doi/10.1103/PhysRevA.110.062428}
}

@article{Bernards2024,
	author  = {Bernards, Fabian and G{\"u}hne, Otfried},
	title   = {Multiparticle singlet states cannot be maximally entangled for the bipartitions},
	journal = {Journal of Mathematical Physics},
	volume  = {65},
	number  = {1},
	pages   = {012201},
	year    = {2024},
	month   = jan,
	doi     = {10.1063/5.0159105},
	eprint  = {2211.03813},
	archivePrefix = {arXiv},
	primaryClass  = {quant-ph},
}

@misc{Konderak2025,
	title={Robust Self-Testing of Multiqudit Supersinglet Slater States via Constant Number of Binary Measurements}, 
	author={Arturo Konderak and Wojciech Bruzda and Remigiusz Augusiak},
	year={2025},
	eprint={2508.15546},
	archivePrefix={arXiv},
	primaryClass={quant-ph},
	url={https://arxiv.org/abs/2508.15546}, 
}

@article{Saha2025,
	title = {Self-testing supersinglets with perfect quantum strategies},
	author = {Saha, Debashis and Cabello, Ad\'an},
	journal = {Phys. Rev. A},
	volume = {112},
	issue = {1},
	pages = {012606},
	numpages = {10},
	year = {2025},
	month = {Jul},
	publisher = {American Physical Society},
	doi = {10.1103/3ymw-sn2z},
	url = {https://link.aps.org/doi/10.1103/3ymw-sn2z}
}

@article{Verstraete2002,
	title = {Four qubits can be entangled in nine different ways},
	author = {Verstraete, F. and Dehaene, J. and De Moor, B. and Verschelde, H.},
	journal = {Phys. Rev. A},
	volume = {65},
	issue = {5},
	pages = {052112},
	numpages = {5},
	year = {2002},
	month = {Apr},
	publisher = {American Physical Society},
	doi = {10.1103/PhysRevA.65.052112},
	url = {https://link.aps.org/doi/10.1103/PhysRevA.65.052112}
}

@article{Fitzi2001,
	title = {Quantum Solution to the Byzantine Agreement Problem},
	author = {Fitzi, Matthias and Gisin, Nicolas and Maurer, Ueli},
	journal = {Phys. Rev. Lett.},
	volume = {87},
	issue = {21},
	pages = {217901},
	numpages = {4},
	year = {2001},
	month = {Nov},
	publisher = {American Physical Society},
	doi = {10.1103/PhysRevLett.87.217901},
	url = {https://link.aps.org/doi/10.1103/PhysRevLett.87.217901}
}

@article{Cabello2002,
	title = {$N$-Particle $N$-Level Singlet States: Some Properties and Applications},
	author = {Cabello, Ad\'an},
	journal = {Phys. Rev. Lett.},
	volume = {89},
	issue = {10},
	pages = {100402},
	numpages = {4},
	year = {2002},
	month = {Aug},
	publisher = {American Physical Society},
	doi = {10.1103/PhysRevLett.89.100402},
	url = {https://link.aps.org/doi/10.1103/PhysRevLett.89.100402}
}

@article{Cabello2003,
	author = {Adán Cabello},
	title = {Supersinglets},
	journal = {Journal of Modern Optics},
	volume = {50},
	number = {6-7},
	pages = {1049--1061},
	year = {2003},
	publisher = {Taylor \& Francis},
	doi = {10.1080/09500340308234551},
	
	
	URL = { 
	
	https://doi.org/10.1080/09500340308234551
	
	
	
	},
	eprint = { 
	
	https://doi.org/10.1080/09500340308234551
	
	
	
	}
	
}

@article{IloOkeke2022,
	author = {Ebubechukwu O. Ilo-Okeke and Yangxu Ji and Ping Chen and
	Yuping Mao and Manikandan Kondappan and Valentin Ivannikov and
	Yanhong Xiao and Tim Byrnes},
	title = {Deterministic preparation of supersinglets with collective spin projections},
	journal = {Physical Review A},
	volume = {106},
	pages = {033314},
	year = {2022},
	doi = {10.1103/PhysRevA.106.033314}
}

@article{Jex2003,
	author = {Jex, I. and Alber, G. and Barnett, S.M. and Delgado, A.},
	title = {Antisymmetric multi-partite quantum states and their applications},
	journal = {Fortschritte der Physik},
	volume = {51},
	number = {2-3},
	pages = {172-178},
	doi = {https://doi.org/10.1002/prop.200310021},
	url = {https://onlinelibrary.wiley.com/doi/abs/10.1002/prop.200310021},
	eprint = {https://onlinelibrary.wiley.com/doi/pdf/10.1002/prop.200310021},
	year = {2003}
}

@article{Piccolini2025,
	doi = {10.1088/2058-9565/ad8214},
	url = {https://doi.org/10.1088/2058-9565/ad8214},
	year = {2024},
	month = {oct},
	publisher = {IOP Publishing},
	volume = {10},
	number = {1},
	pages = {015013},
	author = {Piccolini, Matteo and Karczewski, Marcin and Winter, Andreas and Lo Franco, Rosario},
	title = {Robust generation of N-partite N-level singlet states by identical particle interferometry},
	journal = {Quantum Science and Technology}
}

@Article{Chen2016,
	author={Chen, Zhen
	and Chen, Ye-Hong
	and Xia, Yan
	and Song, Jie
	and Huang, Bi-Hua},
	title={Fast generation of three-atom singlet state by transitionless quantum driving},
	journal={Scientific Reports},
	year={2016},
	month={Mar},
	day={02},
	volume={6},
	number={1},
	pages={22202},
	issn={2045-2322},
	doi={10.1038/srep22202},
	url={https://doi.org/10.1038/srep22202}
}

@article{Shao2010,
	doi = {10.1088/1367-2630/12/2/023040},
	url = {https://doi.org/10.1088/1367-2630/12/2/023040},
	year = {2010},
	month = {feb},
	publisher = {},
	volume = {12},
	number = {2},
	pages = {023040},
	author = {Shao, Xiao-Qiang and Wang, Hong-Fu and Chen, Li and Zhang, Shou and Zhao, Yong-Fang and Yeon, Kyu-Hwang},
	title = {Converting two-atom singlet state into three-atom singlet state via quantum Zeno dynamics},
	journal = {New Journal of Physics}
}

@article{Lu2012,
	title={Driving three atoms into a singlet state in an optical cavity via adiabatic passage of a dark state},
	volume={46},
	ISSN={1361-6455},
	url={http://dx.doi.org/10.1088/0953-4075/46/1/015502},
	DOI={10.1088/0953-4075/46/1/015502},
	number={1},
	journal={Journal of Physics B: Atomic, Molecular and Optical Physics},
	publisher={IOP Publishing},
	author={Lu, Mei and Xia, Yan and Song, Jie and Song, He-Shan},
	year={2012},
	month=Dec, pages={015502} }

@article{Lin2007,
	title = {Generation of the singlet state for three atoms in cavity QED},
	author = {Lin, Gong-Wei and Ye, Ming-Yong and Chen, Li-Bo and Du, Qian-Hua and Lin, Xiu-Min},
	journal = {Phys. Rev. A},
	volume = {76},
	issue = {1},
	pages = {014308},
	numpages = {4},
	year = {2007},
	month = {Jul},
	publisher = {American Physical Society},
	doi = {10.1103/PhysRevA.76.014308},
	url = {https://link.aps.org/doi/10.1103/PhysRevA.76.014308}
}

@article{Jin2005,
	title = {Generation of a supersinglet of three three-level atoms in cavity QED},
	author = {Jin, Guang-Sheng and Li, Shu-Shen and Feng, Song-Lin and Zheng, Hou-Zhi},
	journal = {Phys. Rev. A},
	volume = {71},
	issue = {3},
	pages = {034307},
	numpages = {3},
	year = {2005},
	month = {Mar},
	publisher = {American Physical Society},
	doi = {10.1103/PhysRevA.71.034307},
	url = {https://link.aps.org/doi/10.1103/PhysRevA.71.034307}
}

@book{Weyl1931,
	author    = {Hermann Weyl},
	title     = {The Theory of Groups and Quantum Mechanics},
	publisher = {Dover Publications},
	address   = {New York},
	year      = {1931}
}

@book{Weyl1928,
	author    = {Hermann Weyl},
	title     = {Gruppentheorie und Quantenmechanik},
	publisher = {S. Hirzel},
	address   = {Leipzig},
	year      = {1928}
}

@article{Slater1929,
	author    = {J. C. Slater},
	title     = {The Theory of Complex Spectra},
	journal   = {Physical Review},
	volume    = {34},
	number    = {10},
	pages     = {1293--1322},
	year      = {1929},
	doi       = {10.1103/PhysRev.34.1293}
}

@article{Demianowicz_2018,
   title={From unextendible product bases to genuinely entangled subspaces},
   volume={98},
   ISSN={2469-9934},
   url={http://dx.doi.org/10.1103/PhysRevA.98.012313},
   DOI={10.1103/physreva.98.012313},
   number={1},
   journal={Physical Review A},
   publisher={American Physical Society (APS)},
   author={Demianowicz, Maciej and Augusiak, Remigiusz},
   year={2018},
   month=July }

@book{fulton1991representation,
  title     = {Representation Theory: A First Course},
  author    = {Fulton, William and Harris, Joe},
  series    = {Graduate Texts in Mathematics},
  volume    = {129},
  year      = {1991},
  publisher = {Springer-Verlag},
  address   = {New York},
  isbn      = {978-0-387-97495-8}
}

@book{JamesKerber1984,
  author    = {Gordon James and Adalbert Kerber},
  title     = {The Representation Theory of the Symmetric Group},
  series    = {Encyclopedia of Mathematics and its Applications},
  volume    = {16},
  publisher = {Cambridge University Press},
  year      = {1984},
  isbn      = {9780521302364}
}

@book{Stanley_2023, place={Cambridge}, edition={2}, series={Cambridge Studies in Advanced Mathematics}, title={Enumerative Combinatorics}, publisher={Cambridge University Press}, author={Stanley, Richard}, year={2023}, collection={Cambridge Studies in Advanced Mathematics}}

@article{Guhne_2009,
   title={Entanglement detection},
   volume={474},
   ISSN={0370-1573},
   url={http://dx.doi.org/10.1016/j.physrep.2009.02.004},
   DOI={10.1016/j.physrep.2009.02.004},
   number={1-6},
   journal={Physics Reports},
   publisher={Elsevier BV},
   author={Gühne, Otfried and Tóth, Géza},
   year={2009},
   month=Apr, pages={1–75} }

@article{Seevinck_2001,
   title={Sufficient conditions for three-particle entanglement and their tests in recent experiments},
   volume={65},
   ISSN={1094-1622},
   url={http://dx.doi.org/10.1103/PhysRevA.65.012107},
   DOI={10.1103/physreva.65.012107},
   number={1},
   journal={Physical Review A},
   publisher={American Physical Society (APS)},
   author={Seevinck, Michael and Uffink, Jos},
   year={2001},
   month=Dec }

@article{Guhne_2005,
   title={Multipartite entanglement in spin chains},
   volume={7},
   ISSN={1367-2630},
   url={http://dx.doi.org/10.1088/1367-2630/7/1/229},
   DOI={10.1088/1367-2630/7/1/229},
   journal={New Journal of Physics},
   publisher={IOP Publishing},
   author={Gühne, Otfried and Tóth, Géza and Briegel, Hans J},
   year={2005},
   month=Nov, pages={229–229} }

@article{Guhne_2006,
   title={Energy and multipartite entanglement in multidimensional and frustrated spin models},
   volume={73},
   ISSN={1094-1622},
   url={http://dx.doi.org/10.1103/PhysRevA.73.052319},
   DOI={10.1103/physreva.73.052319},
   number={5},
   journal={Physical Review A},
   publisher={American Physical Society (APS)},
   author={Gühne, Otfried and Tóth, Géza},
   year={2006},
   month=May }

@article{Terhal_2000,
   title={Schmidt number for density matrices},
   volume={61},
   ISSN={1094-1622},
   url={http://dx.doi.org/10.1103/PhysRevA.61.040301},
   DOI={10.1103/physreva.61.040301},
   number={4},
   journal={Physical Review A},
   publisher={American Physical Society (APS)},
   author={Terhal, Barbara M. and Horodecki, Paweł},
   year={2000},
   month=Mar }

@article{Tomiyama1985OnTG,
  title={On the geometry of positive maps in matrix algebras. II},
  author={Jun Tomiyama},
  journal={Linear Algebra and its Applications},
  year={1985},
  volume={69},
  pages={169-177},
  url={https://api.semanticscholar.org/CorpusID:121140908}
}

@book{Georgi:1999wka,
    author = "Georgi, Howard",
    title = "{Lie algebras in particle physics}",
    edition = "2nd ed.",
    publisher = "Perseus Books",
    address = "Reading, MA",
    volume = "54",
    year = "1999"
}

@inbook{Fulton_1996, place={Cambridge}, series={London Mathematical Society Student Texts}, title={Representations of the general linear group}, booktitle={Young Tableaux: With Applications to Representation Theory and Geometry}, publisher={Cambridge University Press}, author={Fulton, William}, year={1996}, pages={104–126}, collection={London Mathematical Society Student Texts}}

@misc{mkrtchyan2025casimireigenvaluesadotimesk,
      title={The Casimir eigenvalues on $ad^{\otimes k}$ of SU(N) are linear on N}, 
      author={R. L. Mkrtchyan},
      year={2025},
      eprint={2506.13062},
      archivePrefix={arXiv},
      primaryClass={math-ph},
      url={https://arxiv.org/abs/2506.13062}, 
}

@misc{CobucciSupersingletGME,
  author       = {Cobucci, Gabriele},
  title        = {{Supersinglet GME dimension}},
  howpublished = {\url{https://github.com/GabrieleCobucci/Supersinglet-GME-dimension}},
  note         = {GitHub repository}
}

@book{GotoGrosshans1978,
  author    = {Goto, Morikuni and Grosshans, Frank D.},
  title     = {Semisimple Lie Algebras},
  series    = {Lecture Notes in Pure and Applied Mathematics},
  volume    = {38},
  publisher = {Marcel Dekker},
  address   = {New York and Basel},
  year      = {1978},
  isbn      = {0-8247-6744-6}
}

@book{DiFrancescoMathieuSenechal1997,
  author    = {Di Francesco, Philippe and Mathieu, Pierre and S{\'e}n{\'e}chal, David},
  title     = {Conformal Field Theory},
  series    = {Graduate Texts in Contemporary Physics},
  publisher = {Springer},
  address   = {New York},
  year      = {1997},
  doi       = {10.1007/978-1-4612-2256-9},
  isbn      = {978-0-387-94785-3}
}

@book{huang2006dirac,
  title     = {Dirac Operators in Representation Theory},
  author    = {Huang, Jing-Song and Pand{\v{z}}i{\'c}, Pavle},
  year      = {2006},
  publisher = {Birkh{\"a}user},
  address   = {Boston},
  series    = {Mathematics: Theory \& Applications}
}

@article{FS_Tropp2012,
  author  = {Joel A. Tropp},
  title   = {User-Friendly Tail Bounds for Sums of Random Matrices},
  journal = {Foundations of Computational Mathematics},
  volume  = {12},
  number  = {4},
  pages   = {389--434},
  year    = {2012},
  doi     = {10.1007/s10208-011-9099-z}
}

@article{FS_Chernoff1952,
  author  = {Chernoff, Herman},
  title   = {A Measure of Asymptotic Efficiency for Tests of a Hypothesis
             Based on the Sum of Observations},
  journal = {The Annals of Mathematical Statistics},
  volume  = {23},
  number  = {4},
  pages   = {493--507},
  year    = {1952},
  doi     = {10.1214/aoms/1177729330}
}

\onecolumngrid
\appendix

\section{Proof of elementary properties}

In this appendix, we provide proofs of the elementary properties of supersinglets.

\subsection{Proof of Property \ref{Property1}}
\label{app:unique}

Let us start by proving that for $n=d$, there exists a unique state such that
\begin{equation}\label{invariance_ss_app}
	U^{\otimes n} \ket{\psi} = \ket{\psi}, \quad \forall U \in SU(n).
\end{equation}
Define $\lambda = (\lambda_1,\lambda_2,\dots,\lambda_k)$ an integer partition of $n$, i.e. $\lambda \vdash n$, and represent it via a Young tableau with $k$ rows of $\lambda_i$ elements.
From Schur-Weyl duality \cite{fulton1991representation}, the vector space on which the full unitary operators $U^{\otimes n}$ act can be decomposed as
\begin{equation}
	\mathbb{C}^{n} \otimes \mathbb{C}^n \otimes \dots \otimes \mathbb{C}^n = \bigoplus_{\substack{\lambda \vdash n\\ l(\lambda) \leq n}} \mathcal{Q}_{\lambda} \otimes \mathcal{P}_{\lambda},
\end{equation}
where $\lambda$ runs over all possible integer partitions of $n$ with length $l(\lambda) \leq n$, i.e. each tableau has at most $n$ rows, $\mathcal{Q}_{\lambda}$ is an irreducible representation of $SU(n)$ and $\mathcal{P}_{\lambda}$ carries the corresponding irreducible representation of the symmetric group $\mathcal{S}_n$. The dimension of $\mathcal{Q}_{\lambda}$ is \cite{Georgi:1999wka}
\begin{equation}\label{Qlambda_dim}
	\text{dim}\, \mathcal{Q}_{\lambda} = \prod_{(i,j)\in \lambda} \frac{n+j-i}{h_{ij}},
\end{equation}
where $h_{ij} = \lambda_i - i + \lambda'_{j}-j+1$ is the Hook length for a box at row $i$ and column $j$ in the Young tableaux associated with the partition $\lambda$ (i.e. the number of cells to the right plus the number of cells below plus $1$) and $\lambda'$ denotes the conjugate partition of $\lambda$ (i.e. a vector with the number of elements in each column). Similarly, the dimension of $\mathcal{P}_{\lambda}$ is 
\begin{equation}\label{Plambda_dim}
	\text{dim}\,\mathcal{P}_{\lambda} = \frac{n!}{\prod_{(i,j)\in \lambda} h_{ij}},
\end{equation}
Let $\lbrace \ket*{\varphi_i^{(\lambda)}} \rbrace_{i}$ for $i = 1,\dots,\text{dim} \, \mathcal{Q}_{\lambda}$ be the basis spanning the vector space on the irreducible representation $\mathcal{Q}_{\lambda}$. Then, the action of $U^{\otimes n}$ is such that
\begin{equation}
	U^{\otimes n}\ket*{\varphi_i^{(\lambda)}} = \sum_{j=1}^{\text{dim}\, \mathcal{Q}_{\lambda}} c_{j} \ket*{\varphi_{j}^{(\lambda)}},
\end{equation}
for a set of complex coefficients $c_{j}$. Therefore, the condition in \eqref{invariance_ss_app} occurs for all possible unitaries $U$ only if $\text{dim}\, \mathcal{Q}_{\lambda} = 1$. By looking at all possible Young tableaux of $SU(n)^{\otimes n}$, we get:
\begin{equation}\label{young_tab_nn}
	\lambda(SU(n)^{\otimes n}) = \lbrace (n), (n-1,1), (n-2,2), (n-2,1,1),\dots, (1,1,\dots,1) \rbrace.
\end{equation}
Note that, since the total number of boxes in a generic tableau is $n$, we have that $h_{ij} = \lambda_i + \lambda'_j - i - j + 1 \leq n - i - j + 1$. Therefore, each term in \eqref{Qlambda_dim}
\begin{equation}
	\frac{n+j-i}{h_{ij}} \geq \frac{n+j-i}{n-i-j+1} \geq 1.
\end{equation}
In particular, if the tableau has more than one column, i.e. $j > 2$, the inequality is strict. Therefore, the only case in which $\text{dim} \, \mathcal{Q}_{\lambda} = 1$ is $\lambda = (1,1,\dots,1)$. This tableau corresponds to a fully antisymmetric tensor with $n$ indices, $\epsilon_{i_1\dots i_n}$, i.e. the supersinglet state, and is unique since $\text{dim}\,\mathcal{P}_{\lambda} = 1$.

Now we can prove that for $n=md$ with $m > 1$, the unique state which satisfies \eqref{invariance_ss_app} is a linear combination of copies and sub-system permutations of the $n=d$ supersinglet. Use again Schur-Weyl duality to decompose the full vector space as
\begin{equation}
	\underbrace{\mathbb{C}^{d} \otimes \mathbb{C}^d \otimes \dots \otimes \mathbb{C}^d}_{n=m d \text{ times}} = \bigoplus_{\substack{\lambda \vdash n\\ l(\lambda) \leq d}} \mathcal{Q}_{\lambda} \otimes \mathcal{P}_{\lambda},
\end{equation}
where now $\lambda$ runs over all possible integer partitions of $n$ with length $l(\lambda) \leq d$. The dimension of the irreducible representations is
\begin{equation}
	\text{dim}\, \mathcal{Q}_{\lambda} = \prod_{(i,j)\in \lambda} \frac{d+j-i}{h_{ij}}, \qquad \text{dim}\,\mathcal{P}_{\lambda} = \frac{n!}{\prod_{(i,j)\in \lambda} h_{ij}}.
\end{equation}
For each tableau $\lambda$, the dimension $\text{dim} \, \mathcal{Q}_{\lambda}$ corresponds to the number of possible semistandard Young tableaux of shape $\lambda$ \cite{Fulton_1996}. Such tableaux can be constructed by filling the boxes in all possible ways in which the numbers are strictly increasing along each column and non-decreasing along each row (with repetitions allowed). As in the previous case, in order to satisfy \eqref{invariance_ss_app}, we need the space $\mathcal{Q}_{\lambda}$ to be $1$-dimensional. The unique shape which can be filled only in one way, leading to $\text{dim}\, \mathcal{Q}_{\lambda} = 1$, is $\lambda = (m,m,\dots,m)$, i.e. the tableau with $d$ rows and $m$ columns. Indeed, for any other tableau for which there is an index $i$ such that $\lambda_i > \lambda_{i+1}$, i.e. a row has more boxes than the next, we can fill the last box of row $i$ with all possible values $j \leq i$, generating different semistandard Young tableaux.

The tableau $\lambda = (m,m,\dots,m)$ corresponds to a tensor with $n = md$ indices grouped into $m$ symmetric sets, each containing $d$ fully antisymmetric indices, i.e. $\epsilon_{i_{1,1}\dots i_{d,1}} \otimes \dots \otimes \epsilon_{i_{1,m}\dots i_{d,m}}$. This corresponds to a state which is composed by $m$ copies of a $d$-dimensional supersinglet. Unlike for $n=d$, in this case there is not a unique tableau of this type in the Schur-Weyl duality decomposition. Indeed, since $h_{ij} = \lambda_i - i + \lambda'_{j} - j + 1 = m - i + d - j + 1$ for all $(i,j)$, we have that
\begin{equation}
	\prod_{i=1}^{d} \prod_{j=1}^{m}h_{ij} = \prod_{i=1}^{d} \frac{(m+d-i)!}{(d-i)!} \quad \implies \quad \text{dim}\, \mathcal{P}_{\lambda} = \frac{n!}{\prod_{(i,j)\in \lambda} h_{ij}} = n! \prod_{i=1}^{d} \frac{(d-i)!}{(m+d-i)!} > 1.
\end{equation}
Denote each of these tableaux as $T \in \text{SYT}(d,m)$, where $\text{SYT}(d,m)$ is the set of the standard Young tableaux with $d$ rows and $m$ columns. Each of them defines a state $\ket{\psi_{T}} = P_{T}\ket{S_{d}}^{\otimes m}$, i.e. a permutation of $m$ copies of $\ket{S_d}$, which satisfies \eqref{invariance_ss_app} in its own space. Therefore, the linear combination of all these states,
\begin{equation}
	\ket{\psi} = \sum_{T \in \text{SYT}(d,m)} c_{T} \, P_{T} \ket{S_d}^{\otimes m},
\end{equation}
is the most general state that satisfies \eqref{invariance_ss_app} when $n = md$ with $m > 1$.

\subsection{Proof of Property \ref{Property2}}
\label{app:prop2}
Let $G_n$ denote a set of $n$ particles and call the associated symmetric group $\mathcal{S}_{G_n}$. Denote by $\pi$ its elements and let sgn$(\pi)$ be their parity. Then, express the supersinglet as
\begin{equation}\label{supersinglet_perm}
	\ket{S_n} = \frac{1}{\sqrt{n!}}\sum_{\pi \in \mathcal{S}_{G_n}} \text{sgn}(\pi) \ket{\pi[G_n])},
\end{equation}
where $\ket{\pi[G_n]} = \ket{\pi(1)}\otimes \dots\otimes \ket{\pi(n)}$, with $\pi(i)$ denoting the permuted $i$'th element in the set. Using this notation, \eqref{supersinglet_dec} becomes
\begin{equation}\label{supersinglet_dec_exp}
	\begin{aligned}
		\ket{S_n}&=\frac{1}{\sqrt{\binom{n}{k}}} \sum_{\substack{T\subset [n]\\ |T|=k}} \text{sgn}(T\oplus \bar{T}) \ket{S_k(T)} \otimes\ket{S_{n-k}(\bar{T})}\\
		&= \frac{1}{\sqrt{\binom{n}{k}}} \sum_{\substack{T\subset [n]\\ |T|=k}} \text{sgn}(T\oplus \bar{T}) \frac{1}{\sqrt{k!}}\sum_{\sigma \in \mathcal{S}_{T_k}} \text{sgn}(\sigma) \ket{\sigma[T_k]} \otimes \frac{1}{\sqrt{(n-k)!}}\sum_{\sigma' \in \mathcal{S}_{\bar{T}_{n-k}}} \text{sgn}(\sigma') \ket{\sigma'[\bar{T}_{n-k}]}\\
		&= \frac{1}{\sqrt{n!}}\sum_{\substack{T\subset [n]\\ |T|=k}} \sum_{\sigma \in \mathcal{S}_{T_k}}\sum_{\sigma' \in \mathcal{S}_{\bar{T}_{n-k}}}\text{sgn}(T\oplus \bar{T})\,\text{sgn}(\sigma)\text{sgn}(\sigma') \ket{\sigma[T_k]\sigma'[\bar{T}_{n-k}]},
	\end{aligned}
\end{equation}
where $T_{k}$ (and $\bar{T}_{n-k}$) denote a set of $k$ (and $n-k$) particles.

Let the tuple $(T,\sigma,\sigma')$ denote each element of $H = \bigcup_{\substack{T\subset [n]\\ |T|=k}} \mathcal{S}_{T_k} \times \mathcal{S}_{\bar{T}_{n-k}}$. Now we show that $H$ is isomorphic to $\mathcal{S}_{G_n}$. First, define the function $F: H \longrightarrow \mathcal{S}_{G_n}$. Then,
\begin{itemize}
	\item $F$ is injective
	
	Assume that $F(T,\sigma,\sigma') = F(\tilde{T},\tilde{\sigma},\tilde{\sigma}')$ for two arbitrary tuples $(T,\sigma,\sigma')$ and $(\tilde{T},\tilde{\sigma},\tilde{\sigma}')$. Since $F(T,\sigma,\sigma')$ and $F(\tilde{T},\tilde{\sigma},\tilde{\sigma}')$ are equal, their first $k$ elements are the same. As they are uniquely identified by the choice of $T$ and $\tilde{T}$, it follows that $T = \tilde{T}$. By the same argument, $\sigma = \tilde{\sigma}$ and $\sigma' = \tilde{\sigma}'$. Therefore, $(T,\sigma,\sigma') = (\tilde{T},\tilde{\sigma},\tilde{\sigma}')$.
	
	\item $F$ is surjective
	
	Write a generic element of $\mathcal{S}_{G_n}$ as $\pi = (\pi(1),\dots,\pi(n))$ and select a subset of length $k$. We can always define $T = \lbrace \pi(1),\dots,\pi(k) \rbrace$ and associate it with a unique permutation $\sigma \in \mathcal{S}_{T_k}$ (similarly for $\sigma'$). Therefore, we can always express any $\pi \in \mathcal{S}_{G_n}$ as $\pi = (T,\sigma,\sigma')$.
\end{itemize}

Since $F$ is bijective, $H$ is isomorphic to $\mathcal{S}_{G_n}$, i.e. the choice of a specific tuple $(T,\sigma,\sigma')$ corresponds to a unique element $\pi$ of $\mathcal{S}_{G_n}$. As a consistency check, each $\mathcal{S}_{T_k} \times \mathcal{S}_{\bar{T}_{n-k}}$ contains $k!(n-k)!$ elements, from which $|H| = {n \choose k}k!(n-k)! = n!$.

This implies that
\begin{equation}
	\sum_{\substack{T\subset [n]\\ |T|=k}} \sum_{\sigma \in S_k}\sum_{\sigma' \in S_{n-k}}\ket{\sigma[T_k]\sigma'[\bar{T}_{n-k}]} = \sum_{\pi \in \mathcal{S}_{G_n}} \ket{\pi[G_n])}.
\end{equation}
Using the multiplication rule for signs, $\text{sgn}(T\oplus \bar{T})\,\text{sgn}(\sigma)\text{sgn}(\sigma') = \text{sgn}(\pi)$,
proving the equality between \eqref{supersinglet_dec_exp} and \eqref{supersinglet_perm}.


\subsection{Proof of Property \ref{Property3}}
\label{app:prop3}
Suppose we keep the first $k$ subsystems and trace out the remaining $n-k$.
\begin{equation}
	\begin{aligned}
		\tr_{k+1,\dots,n}(\rho_n) &= \frac{1}{n!}\sum_{a_1\dots,a_n=0}^{n-1}\sum_{b_1\dots,b_n=0}^{n-1} \varepsilon_{a_1\dots a_n}\varepsilon_{b_1\dots b_n}\tr_{k+1,\dots,n} (\ketbra{a_1\dots a_n}{b_1\dots b_n}) \\
		& = \frac{1}{n!}\sum_{a_1\dots,a_k=0}^{n-1}\sum_{b_1\dots,b_k=0}^{n-1} \sum_{c_{k+1}\dots,c_n=0}^{n-1} \varepsilon_{a_1\dots a_k c_{k+1}\dots c_{n}}\varepsilon_{b_1\dots b_k c_{k+1}\dots c_n} \ketbra{a_1\dots a_k}{b_1\dots b_k}\\
		& = \frac{(n-k)!}{n!} \sum_{a_1\dots,a_k=0}^{n-1}\sum_{b_1\dots,b_k=0}^{n-1} \delta^{b_1\dots b_k }_{a_1\dots a_k}\ketbra{a_1\dots a_k}{b_1\dots b_k}\\
		& = \frac{1}{{n\choose k}} \Pi^{(k)}_\text{asym},
	\end{aligned}
\end{equation}where in the second-to-last line, we used the Levi-Civita contraction identity and in the last line identified the binomial coefficient. Note that here we have defined the antisymmetric projector as
\begin{equation}
	\Pi_{\text{asym}}^{(k)} = \frac{1}{k!} \sum_{a_1,\dots,a_k=0}^{n-1}\sum_{b_1,\dots,b_k=0}^{n-1} \delta^{b_1\dots b_k }_{a_1,\dots a_k}\ketbra{a_1\dots a_k}{b_1\dots b_k},
\end{equation}
with the generalised Kronecker delta
\begin{equation}
	\delta^{b_1\dots b_k }_{a_1\dots a_k} = \sum_{\pi} (-1)^{\pi} \delta_{a_1,b_{\pi(1)}}\dots \delta_{a_k,b_{\pi(k)}} = \begin{cases}
		1,\quad &\text{if } b_{1},\dots,b_{k} \text{ are distinct and an even permutation of } a_{1},\dots,a_{k},\\
		-1,\quad &\text{if } b_{1},\dots,b_{k} \text{ are distinct and an odd permutation of } a_{1},\dots,a_{k},\\
		0,\quad & \text{otherwise}.
	\end{cases}
\end{equation}
This is exactly the projector defined as
\begin{equation}\label{antisym_proj}
	\Pi_{\text{asym}}^{(S)} = \frac{1}{|S|!}\sum_{k \in \mathcal{L}_{S}} \sum_{\pi,\sigma} (-1)^{\pi + \sigma} \ketbra{\pi(S_{k})}{\sigma(S_{k})},
\end{equation}
where $k$ runs over the set $\mathcal{L}_{S}$ of all possible $d \choose S$ subsets of size $|S|$, denoted as $S_{k}$. Indeed, here we are summing over all strings of $\lbrace a_1,\dots,a_k \rbrace$ with distinct elements and adding a minus sign all the times the string $\lbrace b_{1},\dots,b_{k} \rbrace$ is an odd permutation of the former.

%

\subsection{Proof of Property \ref{Property4}}
\label{app:circ}
Consider the supersinglet state $\ket{S_n}$
\begin{align}
	\ket{S_n} = \frac{1}{\sqrt{n!}} \sum_\pi \text{sign}(\pi) \ket{\pi(n)} \equiv \frac{1}{\sqrt{n!}} \sum_\pi (-1)^\pi \ket{\pi_1,\dots,\pi_n},
\end{align} and let $\ket{\tilde\psi_{n-1}}$ be a $(n-1)$-dimensional supersinglet embedded in the levels $\{1,2,\dots,n-1\}$,
\begin{align}
	\ket{\tilde \psi_{n-1}} = \frac{1}{\sqrt{(n-1)!}} \sum_{\pi}\text{sign}(\pi) \ket{\pi (n-1)}.
\end{align}Furthermore, let 
\begin{align}
	\ket{\phi_n} = \frac{1}{\sqrt{n}} \sum_{k=0}^{n-1}\alpha_k \ket{k}, \quad  \alpha_k = \begin{cases}
		1,\,n \,\text{odd}\\
		(-1)^k,\,n\, \text{even}
	\end{cases}. 
\end{align}Note that 
\begin{align}
	\ket{\phi_n} = \begin{cases}
		F_n \ket{0},\,n\,\text{odd}\\
		F_n \ket{n/2},\,n\,\text{even}
	\end{cases},
\end{align}where $F_n$ is the $n$-dimensional Fourier transform. Define the $n$-dimensional CNOT gate as $C_{ij} \ket{a}_i \ket{b}_j = \ket{b+a\, \text{mod}\, n }_i\ket{b}_j$. Then the following holds true up to a global phase
\begin{align}\label{eq:construct}
	\ket{S_n} \equiv \left(\prod_{j=1}^{n-1} C_{jn} \right) \ket{\tilde \psi_{n-1}} \otimes \ket{\phi_n}
\end{align}

\textit{Proof}: We first show that the RHS of \eqref{eq:construct} is a superposition of all permutation basis states. Expanding, we obtain
\begin{align}
	\frac{1}{\sqrt{n!}} \sum_{\pi}\sum_{k=0}^{n-1} \alpha_k  \text{sign}(\pi) \ket{\pi_1 + k,\pi_2+k\dots,\pi_{n-1} + k,k},
\end{align}where all additions are understood modulo \(n\). It can be easily seen that $(\pi_1+k,\pi_2+k,\dots,\pi_{n-1}+k,k)$ is a permutation of $(0,1,\dots,n-1)$. Suppose upon the application of the CNOT gates, we have the permutation $\sigma = (\sigma_1,\dots,\sigma_n)$. $k=\sigma_n $ remains invariant as it is the control. Then the pre-CNOT states are uniquely (bijectively) fixed through $\pi_j = \sigma_j-k\,\text{mod}\,n$. Therefore, every permutation $\sigma$ comes from a unique pair $(\pi,k)$. Then, summing over all permutations $\pi$ and states $k$, the final decomposition contains all permutation basis states.

Next, we show that the the state contains the correct local phases (i.e., is totally antisymmetric). Let us first note that the supersinglet can be equivalently expressed in terms of a Levi-Civita symbol as 
\begin{align}
	\ket{S_n} = \frac{1}{\sqrt{n!}} \sum_{a_1\dots,a_n=0}^{n-1} \varepsilon_{a_1\dots a_n} \ket{a_1\dots a_n}, \quad \varepsilon_{a_1\dots a_n} = \begin{cases}
		\text{sign}(a_1,\dots,a_n) \,\text{if}\,(a_1,\dots,a_n)\,\text{is a permutation of }\, (0,1,\dots,n-1)\\
		0,\,\text{otherwise}  
	\end{cases}.
\end{align}Let $\sigma = (\pi_1+k,\dots,\pi_{n-1}+k,k)$ be a post-CNOT state and $u= (\pi_1,\dots,\pi_{n-1},0)$ be the corresponding unshifted tuple. Defining the modular translation $T_k(x) = x+k\,,\text{mod}\,n$, we have $\sigma=T_k (u)$. Say our target coefficient is $\varepsilon_{\sigma_1,\dots,\sigma_n} = \varepsilon_{\pi_1+k,\pi_2+k,\dots,\pi_{n-1}+k,k}$. Since $T_k$ is a permutation, we have the property $\varepsilon_{\sigma_1\dots \sigma_n} = \text{sign}(T_k) \varepsilon_{\pi_1\dots \pi_{n-1},0}$. Moreover, $T_1$ corresponds to an $n$-cycle and therefore $\text{sign}(T_1)=(-1)^{n-1}$. Further note that $T_k$ corresponds to the application of $T_1$ $k$ times. So, sign$(T_k) = (-1)^{k(n-1)}$. Next, sign$(\pi) =\text{sign}(\pi_1,\dots,\pi_{n-1}) = \text{sign}(0,\pi_1,\dots,\pi_{n-1}) $. Since the transformation $(0,\pi_1,\dots,\pi_{n-1})\to(\pi_1,\dots,\pi_{n-1},0)$ requires $n-1$ adjacent swaps, we have that $\varepsilon_{\pi_1\dots\pi_{n-1},0} = (-1)^{n-1}\text{sign}(\pi)$. Overall, we have our target coefficient
\begin{align}
	\varepsilon_{\sigma_1\dots \sigma_n} = (-1)^{(k+1)(n-1)}\text{sign}(\pi).
\end{align}On the other hand, the construction gives us the coefficient $\alpha_k \text{sign}(\pi)$. 

To show that these are equal, let us first assume that $n$ is odd. In this case, $(-1)^{(k+1)(n-1)}=1$ and $\alpha_k=1$ $\forall$ $k$. Therefore, $\varepsilon_{\sigma_1\dots\sigma_n} = \alpha_k\text{sign}(\pi) = \text{sign}(\pi)$. Next, let us assume that $n$ is even. Here, $(-1)^{(k+1)(n-1)} = (-1)^{k+1}$. Therefore, $\varepsilon_{\sigma_1\dots\sigma_n} = -(-1)^{k}\text{sign}(\pi) = -\alpha_k\text{sign}(\pi)$. The overall minus sign corresponds to a global phase and is irrelevant.

\section{SDP symmetrisation}\label{SDP_symmetr_app}
In this appendix, we show that the semidefinite program (SDP) introduced in \cite{Cobucci2024} to compute the white noise robustness of GME-dimension can be reduced for the cases of the supersinglet state, the GHZ state and the Cluster states. Choosen a $k$-positive map, which we select to be the generalised reduction map $\Lambda(X) = \tr(X)\openone - \frac{1}{k}X$ \cite{Terhal_2000}, the SDP is
\begin{equation}\label{SDP}
	\begin{aligned}
		\max_{v,\sigma} & \quad v\\
		\text{s.t.}&\quad v\rho+\frac{1-v}{d^{n}}\openone =  \sum_{ S|\bar{S}} \sigma_{S|\bar{S}},\\
		& \quad \sigma_{S|\bar{S}} \succeq 0 \quad \forall (S|\bar{S}),\\
		& \quad (\Lambda_S\otimes \openone_{\bar{S}})[\sigma_{S|\bar{S}}] \succeq 0 \quad \forall (S|\bar{S}),
	\end{aligned}
\end{equation}
where $\rho \in \mathcal{H}^{\otimes n}$ is a multipartite state of local dimension $d$, and $\lbrace \sigma_{S|\bar{S}} \rbrace_{S|\bar{S}}$ are unnormalised states defined for every possible partition of the $n$-particles into two non-overlapping sets $\lbrace S,\bar{S} \rbrace$. The critical visibility $v$ serves as a quantifier of the GME-dimension of $\rho$, since obtaining $v < 1$ for a specific choice of $k$ implies that $\mathcal{D}_{\text{GME}}(\rho) > k$.

Fixing $\rho = \ketbra{S_n}$ (and consequently $d=n$), the noisy state $v \, S_{n} + \frac{1-v}{d^n} \openone$ in \eqref{SDP} shows two types of symmetries:
\begin{enumerate}
	\item Invariance under the permutation group $\mathcal{S}_n$.
	
	This follows from the uniqueness of the supersinglet state $\ket{S_n}$. Since this is the unique linear independent state spanning the space of the $n$-dimensional rank-$n$ antisymmetric tensors (see Appendix \ref{app:unique}), we have that
	\begin{equation}\label{supersinglet_permutation}
		\ketbra{S_n}{S_n} = \frac{1}{n!} \sum_{\pi \in \mathcal{S}_n} \text{sgn}(\pi) P_{\pi},
	\end{equation}
	where $\pi$ labels the generic element in the permutation group $\mathcal{S}_n$ and $P_{\pi} = \sum_{i_1,\dots,i_n} \ketbra{\pi(i_1,\dots,i_n)}{i_1,\dots,i_n}$ is the associated permutation operator. Therefore,
	\begin{equation}
		P_{\tau} \ketbra{S_n}{S_n} P_{\tau}^{\dagger} = \frac{1}{n!} \sum_{\pi \in \mathcal{S}_n} \text{sgn}(\pi) P_{\tau}P_{\pi}P_{\tau}^{\dagger} = \frac{1}{n!} \sum_{\sigma \in \mathcal{S}_n} \text{sgn}(\tau^{-1}\circ\sigma\circ \tau) P_{\sigma} = \frac{1}{n!} \sum_{\sigma \in \mathcal{S}_n} \text{sgn}(\sigma) P_{\sigma} = \ketbra{S_n}{S_n}.
	\end{equation}
	
	\item Invariance under the unitary group $SU(n)^{\otimes n}$.
	
	This comes directly from the definition of the supersinglet state.
\end{enumerate}
Because of the $\mathcal{S}_n$ invariance, we can twirl the first constraint in \eqref{SDP} as
\begin{equation}
	\frac{1}{n!} \sum_{\pi \in \mathcal{S}_n} P_{\pi} \left(v \, S_n + \frac{1-v}{d^n} \openone \right) P_{\pi}^{\dagger} = \frac{1}{n!} \sum_{\pi \in \mathcal{S}_n} P_{\pi} \left(\sum_{S|\bar{S}} \sigma_{S|\bar{S}}\right) P_{\pi}^{\dagger} = \sum_{S|\bar{S}} \frac{1}{n!} \sum_{\pi \in \mathcal{S}_n} P_{\pi} \sigma_{\pi^{-1}(S)|\pi^{-1}(\bar{S})} P_{\pi}^{\dagger} = \sum_{S|\bar{S}} \tilde{\sigma}_{S|\bar{S}},
\end{equation}
where we have defined the new variables $\tilde{\sigma}_{S|\bar{S}} =\frac{1}{n!} \sum_{\pi \in \mathcal{S}_n} P_{\pi} \sigma_{\pi^{-1}(S)|\pi^{-1}(\bar{S})} P_{\pi}^{\dagger}$. Since the twirled operators are such that $\tilde{\sigma}_{\tau(S)|\tau(\bar{S})} = P_{\tau} \tilde{\sigma}_{S|\bar{S}} P_{\tau}^{\dagger}$, there is only one independent variable for each cut $(S|\bar{S})$ with subsets of the same size $|S|$. Therefore, the total number of independent operators reduce to $\floor{\frac{n}{2}}$. Denote the independent operators as $\tilde{\sigma}_{m}$, for $m = 1,\dots,\floor{\frac{n}{2}}$, and call $g_{S}$ the permutations which lead to different bipartitions $(S|\bar{S})$ with subsets of the same size $|S|$. The first constraint in \eqref{SDP} becomes
\begin{equation}
	v\, S_n + \frac{1-v}{d^n} \openone = \sum_{m=1}^{\floor{\frac{n}{2}}} \sum_{g_{S}} P_{g_{S}} \tilde{\sigma}_{m} P_{g_{S}}^{\dagger}.
\end{equation}
Using the invariance under the group $SU(n)^{\otimes n}$, we can also twirl the constraint as
\begin{equation}
	\int \, d\mu_{\text{H}}(U) \, U^{\otimes n} \left(v \, S_n + \frac{1-v}{d^n} \openone\right) U^{\otimes n \dagger} = \sum_{m=1}^{\floor{\frac{n}{2}}} \sum_{g_{S}} P_{g_{S}} \int \, d\mu_{\text{H}}(U) \, U^{\otimes n} \tilde{\sigma}_{m} U^{\otimes n \dagger} P_{g_{S}}^{\dagger},
\end{equation}
where $d\mu_{H}(U)$ is the Haar measure over all possible unitaries $U \in SU(n)$ and we have used that $[P_{\pi},U^{\otimes n}] = 0$ for each $\pi$ and $U$. Define the new variables $\bar{\sigma}_{m} = \int \,d\mu_{\text{H}}(U) \, U^{\otimes n} \tilde{\sigma}_{m} U^{\otimes n \dagger}$. By Schur-Weyl duality \cite{fulton1991representation}, any element of the group can be written as
\begin{equation}
	U^{\otimes n} = \bigoplus_{\lambda} U_{\lambda}(U) \otimes I_{V_{\lambda}},
\end{equation}
where $\lambda$ denotes all possible irreducible representations, $U_{\lambda}(U)$ is the matrix representing $U^{\otimes n}$ in the associated Schur space, and $I_{V_\lambda}$ is the identity on the associated irreducible representation space of $\mathcal{S}_{n}$. Therefore, since $[\bar{\sigma}_m, U^{\otimes n}] = 0$, for each $U$, using Schur's lemma,
\begin{equation}\label{Schur-Weyl_sigma}
	\bar{\sigma}_m = \bigoplus_{\lambda} I_{U_{\lambda}} \otimes M_{\lambda}(m),
\end{equation}
with $I_{U_{\lambda}}$ being the identity on the Schur space and $M_{\lambda}(m)$ being generic matrices in $V_{\lambda}$. Similarly, since $[P_{\pi},U^{\otimes n}] = 0$ for any permutation operator $P_{\pi}$, using Schur's lemma,
\begin{equation}\label{perm_young_rep}
	P_{\pi} = \bigoplus_{\lambda} I_{U_{\lambda}} \otimes \rho_{\lambda}(\pi),
\end{equation}
where $\rho_{\lambda}(\pi)$ is the matrix representation of the permutation $\pi \in \mathcal{S}_{n}$ in the irreducible representation space labelled by $\lambda$. For each $\lambda$, the set $\lbrace \rho_{\lambda}(\pi) \rbrace_{\pi}$ generates an algebra of operators spanning the space $V_{\lambda}$ \cite{JamesKerber1984}, i.e. any matrix $M_{\lambda}(m) = \sum_{\pi} x_{\pi}^{(m)} \rho_{\lambda}(\pi)$, for a set of coefficients $\lbrace x_{\pi}^{(m)} \rbrace_{\pi}$. Therefore, we can rewrite the operators in \eqref{Schur-Weyl_sigma} as
\begin{equation}\label{barsigma_permutation}
	\bar{\sigma}_m = \bigoplus_{\lambda} I_{U_{\lambda}} \otimes M_{\lambda}(m) = \bigoplus_{\lambda} I_{U_{\lambda}} \otimes \left(\sum_{\pi} x_{\pi}^{(m)} \rho_{\lambda}(\pi)\right) = \sum_{\pi} x_{\pi}^{(m)} \left(\bigoplus_{\lambda} I_{U_{\lambda}} \otimes \rho_{\lambda}(\pi)\right) = \sum_{\pi} x_{\pi}^{(m)} P_{\pi}.
\end{equation}
The constraints in the symmetrised version of \eqref{SDP} become:
\begin{enumerate}
	\item 
	$\displaystyle v\, S_n + \frac{1-v}{d^n} \openone = \sum_{m=1}^{\floor{\frac{n}{2}}} \sum_{g_{S}} P_{g_{S}} \bar{\sigma}_{m} P_{g_{S}}^{\dagger}$
	
	Using \eqref{supersinglet_permutation} and \eqref{barsigma_permutation}, we get
	\begin{equation}
		\frac{v}{n!}\sum_{\pi\in \mathcal{S}_{n}} \text{sgn}(\pi) P_{\pi} + \frac{1-v}{d^n} \openone = \sum_{\pi \in \mathcal{S}_n}\sum_{m=1}^{\floor{\frac{n}{2}}} \sum_{g_{S}} x_{\pi}^{(m)} P_{g_{S} \circ \pi \circ g_{S}^{-1}} = \sum_{\pi \in \mathcal{S}_n}\sum_{m=1}^{\floor{\frac{n}{2}}} \sum_{g_{S}} x_{g_{S}^{-1} \circ \pi \circ g_{S}}^{(m)} P_{\pi},
	\end{equation}
	that becomes, after factorising the $P_{\pi}$,
	\begin{equation}
		\frac{v}{n!}\text{sgn}(\pi) + \frac{1-v}{d^n} \delta_{e,\pi} = \sum_{m=1}^{\floor{\frac{n}{2}}} \sum_{g_{S}} x_{g_{S}^{-1} \circ \pi \circ g_{S}}^{(m)}, \quad \forall \pi \in \mathcal{S}_n,
	\end{equation}
	where we have denoted as $P_{e}$ the identity permutation.
	
	\item $\displaystyle \bar{\sigma}_{m} \succeq 0, \quad \forall m = 1,\dots,\left\lfloor\frac{n}{2}\right\rfloor$
	
	Using \eqref{barsigma_permutation}, we get
	\begin{equation}
		\bar{\sigma}_m = \sum_{\pi} x_{\pi}^{(m)} \left(\bigoplus_{\lambda} I_{U_{\lambda}} \otimes \rho_{\lambda}(\pi)\right) \succeq 0 \quad \implies \quad \sum_{\pi} x_{\pi}^{(m)}\rho_{\lambda}(\pi) \succeq 0, \quad \forall m,\lambda.
	\end{equation}
	
	\item $\displaystyle (\Lambda_S\otimes \openone_{\bar{S}})[\bar{\sigma}_{m}] \succeq 0, \quad \forall m = 1,\dots,\left\lfloor\frac{n}{2}\right\rfloor$
	
	Using \eqref{barsigma_permutation}, we can write the action of the generalised reduction map as
	\begin{equation}\label{red_map}
		(\Lambda_S\otimes \openone_{\bar{S}})[\bar{\sigma}_{m}] = \openone_{S} \otimes \tr_{S}(\bar{\sigma}_m) - \frac{1}{k}\, \bar{\sigma}_m = \sum_{\pi} x_{\pi}^{(m)} \left(\openone_{S} \otimes \tr_{S}(P_{\pi}) - \frac{1}{k}P_{\pi}\right).
	\end{equation}
	The partial trace of a permutation operator is
	\begin{equation}
		\tr_{S}(P_{\pi}) = \sum_{[i_{S}],[i_{\bar{S}}]} \tr_{S}(\ketbra{\pi([i_{S}][i_{\bar{S}}])}{[i_{S}][i_{\bar{S}}]}) = \sum_{[i_{S}]} \delta_{\pi([i_{S}]),[i_{S}]} \sum_{[i_{\bar{S}}]} \ketbra{\pi([i_{\bar{S}}])}{[i_{\bar{S}}]} = \sum_{[i_{S}]} \delta_{\pi([i_{S}]),[i_{S}]} P_{\tilde{\pi}_{\bar{S}}},
	\end{equation}
	where $[i_{S}]$ ($[i_{\bar{S}}]$) denotes the indices contained in the subset $S$ ($\bar{S}$), and $P_{\tilde{\pi}_{\bar{S}}}$ is the permutation operator acting only on the elements of the subset $\bar{S}$. Any permutation of $n$ elements can be decomposed as a product of cycles $\pi = \Pi_{k}C_{k}$, where $C_{k} = (a_1a_2\dots a_l)$ is a cycle of length $l \in [1,n]$. Assume that the permutation is already cyclic, i.e. $\pi([i_{S}]) = (i_{1} i_{2}\dots i_{|S|})$. Then,
	\begin{equation}
		\tr_{S}(P_{\pi}) = \sum_{i_1,\dots,i_{|S|}} \delta_{\pi(i_{1}),i_{1}}\delta_{\pi(i_{2}),i_{2}}\dots \delta_{\pi(i_{|S|}),i_{|S|}} P_{\tilde{\pi}} = \sum_{i_1,\dots,i_{|S|}} \delta_{i_{2},i_{1}}\delta_{i_{3},i_{2}}\dots \delta_{i_{1},i_{|S|}} = \sum_{i_{S}=1}^{n} P_{\tilde{\pi}_{\bar{S}}} = n P_{\tilde{\pi}_{\bar{S}}},
	\end{equation}
	where we have used that for a cyclic permutation of length $|S|$, $\pi(i_k) = i_{k+1}$ for $k=1,\dots,|S|-1$ and $\pi(i_{|S|}) = i_1$. In this case, the product of Kronecker delta's cancels all the sums except one. For a general permutation, call $c_{S}(\pi)$ the number of times the subset $S$ contains a full cycle in the decomposition of $\pi$. Then, we have
	\begin{equation}\label{partial_tr_perm}
		\tr_{S}(P_{\pi}) = n^{c_{S}(\pi)}P_{\tilde{\pi}_{\bar{S}}}.
	\end{equation}
	Using \eqref{partial_tr_perm} in \eqref{red_map}, we can write
	\begin{equation}
		\sum_{\pi} x_{\pi}^{(m)} \openone_{S} \otimes \tr_{S}(P_{\pi}) = \sum_{\pi} x_{\pi}^{(m)} \openone_{S} \otimes n^{c_{S}(\pi)} P_{\tilde{\pi}_{\bar{S}}} = \sum_{\pi} x_{\pi}^{(m)} n^{c_{S}(\pi)} P_{\tilde{\pi}},
	\end{equation}
	where, via the cycle notation, we have defined $\tilde{\pi}_{\bar{S}} = \Pi_{k}C_{k}$ and $\tilde{\pi} = \Pi_{k=1}^{|S|}(i_{k}) \tilde{\pi}_{\bar{S}}$, i.e. the new permutation acts trivially on the elements of the $S$ subset. Using \eqref{perm_young_rep}, we can rewrite the constraint in \eqref{red_map} as
	\begin{equation}
		\begin{aligned}
			(\Lambda_S\otimes \openone_{\bar{S}})[\bar{\sigma}_{m}] &= \sum_{\pi} x_{\pi}^{(m)} \left[n^{c_{S}(\pi)}\left(\bigoplus_{\lambda} I_{U_{\lambda}} \otimes \rho_{\lambda}(\tilde{\pi})\right) - \frac{1}{k}\left(\bigoplus_{\lambda} I_{U_{\lambda}} \otimes \rho_{\lambda}(\pi)\right)\right] \succeq 0\\
			&\quad \implies \quad \sum_{\pi} x_{\pi}^{(m)} \left[n^{c_{S}(\pi)}\rho_{\lambda}(\tilde{\pi}) - \frac{1}{k}\rho_{\lambda}(\pi) \right] \succeq 0, \quad \forall m,\lambda.
		\end{aligned}
	\end{equation}
\end{enumerate}

The SDP in \eqref{SDP} becomes
\begin{equation}\label{SDP_symm}
	\begin{aligned}
		\max_{v,\sigma} & \quad v\\
		\text{s.t.}&\quad \frac{v}{n!}\text{sgn}(\pi) + \frac{1-v}{d^n} \delta_{e,\pi} = \sum_{m=1}^{\floor{\frac{n}{2}}} \sum_{g_{S}} x_{g_{S}^{-1} \circ \pi \circ g_{S}}^{(m)}, \quad \forall \pi \in \mathcal{S}_n,\\
		& \quad \sum_{\pi} x_{\pi}^{(m)}\rho_{\lambda}(\pi) \succeq 0, \quad \forall m,\lambda,\\
		& \quad \sum_{\pi} x_{\pi}^{(m)} \left[n^{c_{S}(\pi)}\rho_{\lambda}(\tilde{\pi}) - \frac{1}{k}\rho_{\lambda}(\pi) \right] \succeq 0, \quad \forall m,\lambda.
	\end{aligned}
\end{equation}
Since this problem has less variables than the original one, it makes possible the computation of the critical visibility for higher values of $n$. On a standard machine, we have been able to solve \eqref{SDP_symm} up to $(n,k) = (7,6)$ (see Table \ref{table:visibility_supersinglet}). Our implementation in Matlab is available at \cite{CobucciSupersingletGME}.

\begin{table}[t!]
	\begin{tabular}{|c|c|c|c|c|c|c|}
		\hline
		$n \backslash k$ & $1$ & $2$ & $3$ & $4$ & $5$ & $6$ \\
		\hline
		$3$ & $0.2059$ & $0.5345$ & $-$ & $-$ & $-$ & $-$ \\
		\hline
		$4$ & $0.1218$ & $0.2763$ & $0.4854$ & $-$ & $-$ & $-$ \\
		\hline
		$5$ & $0.0454$ & $0.0996$ & $0.1673$ & $0.2668$ & $-$ & $-$ \\
		\hline
		$6$ & $0.0210$ & $0.0438$ & $0.0690$ & $0.0976$ & $0.1338$ & $-$\\
		\hline
		$7$ & $0.0066$ & $0.0136$ & $0.0210$ & $0.0295$ & $0.0376$ & $0.0470$\\
		\hline
	\end{tabular}
	\caption{Upper bounds on the critical visibility $v$ for supersinglet up to $(n,k) = (7,6)$ using the symmetrised SDP \eqref{SDP_symm}.}
	\label{table:visibility_supersinglet}
\end{table}

\section{Proof of witness bound}\label{App_GME-dimension_witness_bound}

Consider the witness
\begin{align}\label{Witness_operator}
	&W=\int dU \ U^{\otimes n} D \ (U^\dagger)^{\otimes n}, \\
	&\text{where} \quad D= \sum_{k_1,\ldots, k_n=0}^{n-1}\ketbra{k_1}\otimes\ldots \otimes \ketbra{k_n}\delta_{\vec{k}},
\end{align}
where $\delta_{\vec{k}}=1$ if all elements in $\vec{k}=(k_1,\ldots,k_n)$ are different and $\delta_{\vec{k}}=0$ otherwise. In this appendix, we prove that the following bound holds:
\begin{equation}
	W \preceq S_n + \frac{\openone - S_n}{n+1}.
\end{equation}
Since $W$ is defined as a twirl over all possible unitaries, it follows that $[W,U^{\otimes n}] = 0$ for each $U \in SU(n)$. Furthermore, since $D$ is obtained by summing over all $\vec{k}$ with distinct elements, we also have that $[W,P_{\pi}] = 0$ for each permutation operator $P_{\pi}$. Therefore, using Schur-Weyl duality and Schur's lemma \cite{fulton1991representation}, we can write
\begin{equation}\label{W_decomposition_irreducible}
	W = \bigoplus_{\substack{\lambda \vdash n\\ l(\lambda) \leq n}} w_{\lambda} \Pi_{\lambda},
\end{equation}
where we have denoted as $\Pi_{\lambda}$ the projector over the space of each irreducible representation. Following the notation in \eqref{Qlambda_dim} and \eqref{Plambda_dim} for the dimension of the irreducible representations, we get that 
\begin{equation}\label{coefficients_expression}
	\tr(W \Pi_{\lambda}) = \bigoplus_{\substack{\lambda' \vdash n\\ l(\lambda') \leq n}} w_{\lambda'} \tr(\Pi_{\lambda'}\Pi_{\lambda}) = w_{\lambda} \tr(\Pi_{\lambda}) = w_{\lambda} \dim \mathcal{Q}_{\lambda}\dim \mathcal{P_{\lambda}} \quad \implies \quad w_{\lambda} = \frac{\tr(W \Pi_{\lambda})}{\dim \mathcal{Q}_{\lambda}\dim \mathcal{P_{\lambda}}}.
\end{equation}
Since $\Pi_{\lambda}$ commutes with every unitary $U^{\otimes n}$ by Schur's lemma, we get
\begin{equation}
	\tr(W\Pi_{\lambda}) = \int dU \ \tr(U^{\otimes n} D \ (U^\dagger)^{\otimes n} \Pi_{\lambda}) = \tr(\Pi_{\lambda}D).
\end{equation}
Moreover, since $D$ is the projector over a permutation invariant subspace, we can also write that
\begin{equation}
	\text{supp } D \simeq \bigoplus_{\lambda \vdash n} (\dim \mathcal{P}_{\lambda}) \mathcal{P}_{\lambda},
\end{equation}
from which $\tr(\Pi_{\lambda}D) = (\dim \mathcal{P}_{\lambda})^2$. Using this into \eqref{coefficients_expression}, we get that the coefficients
\begin{equation}
	w_{\lambda} = \frac{\dim\mathcal{P}_{\lambda}}{\dim\mathcal{Q}_{\lambda}} = \frac{n!}{\prod_{(i,j)\in \lambda} (n+j-i)}.
\end{equation}
When $\lambda = (1,1,\dots,1)$, i.e. the supersinglet representation, we get:
\begin{equation}
	\prod_{(i,j)\in (1,1,\dots,1)} (n+j-i) = \prod_{i=1}^{n}(n+1-i) = n! \quad \implies \quad w_{(1,1,\dots,1)} = \frac{n!}{n!} = 1.
\end{equation}
From this expression, it follows that the values of $w_{\lambda}$ decrease as we move the boxes in the Young tableau from the lower rows to the higher ones, since the denominator increases. Therefore, the largest $w_{\lambda}$ occurs for $\lambda = (2,1,\dots,1)$, where
\begin{equation}
	\prod_{(i,j)\in (2,1,\dots,1)} (n+j-i) = \prod_{i=1}^{n-1}(n+1-i) (n+2-1) = n!(n+1) \quad \implies \quad w_{(2,1,\dots,1)} = \frac{1}{n+1}.
\end{equation}
Therefore, since $w_{\lambda} \leq w_{(2,1,\dots,1)}$ for each $\lambda \neq (1,1,\dots,1)$, we get that
\begin{equation}
	W \preceq S_n + \frac{\openone - S_n}{n+1}.
\end{equation}

\section{Finite-setting certification}
\label{app:finite_settings}
In this appendix, we show how to use finite statistics to perform practical detection of GME-dimension using the witness presented in Section \ref{random_witness}. Let us consider the finite scenario, in which we draw a list of $N$ independent Haar-random unitaries once and retain the associated witness
\begin{equation}\label{app_eq:fs_list}
	W_N=\frac1N\sum_{i=1}^N U_i^{\otimes n}D(U_i^\dagger)^{\otimes n}.
\end{equation}
Then, we measure independent copies of one fixed state $\rho$, choosing a unitary from the list uniformly and independently with replacement on every trial.  We denote the total number of trials by $T$, which we fix in advance, so each setting is used $R=T/N$ times on average.  Measurements are trusted and all trials are retained. If $K$ trials yield distinct outcomes on all particles, the observed witness is $w=W_{\mathrm{obs}}=K/T$.

The null hypothesis is $\rho\in\mathcal G_r$, with witness bound $b_r=(r+1)/(n+1)$ from Eq.~\eqref{witnessRandomised}. Now, we address both the problems of defining the witness with a finite number of measurement settings and using a finite number of trials to estimate its expectation value.

\begin{enumerate}
	\item Finite number of measurement settings $N$.
	
	In general a witness defined as the one in \eqref{app_eq:fs_list}, with a finite number of settings, does not need to obey the bound of Eq.~\eqref{witnessRandomised}. However, we now show that for any $b_r\leq x\leq1$ it satisfies the certificate
	\begin{equation}\label{eq:fs_list_tail}
		\begin{aligned}
			&\mathbb P_{\mathrm{list}}\!\left\{
			\sup_{\rho\in\mathcal G_r}\tr(\rho W_N)>x\right\} \leq(n^n-1)\exp\!\left[-ND_{\text{KL}}\!\left(
			\frac{nx-r}{n-r}\middle\Vert\frac1{n+1}\right)\right],
		\end{aligned}
	\end{equation}
	where $D_{\text{KL}}(a\Vert b)=a\ln(a/b)+(1-a)\ln[(1-a)/(1-b)]$ is the Kullback-Leibler divergence, with endpoint limits.
	
	To prove this, write $D_U=U^{\otimes n}D(U^\dagger)^{\otimes n}$. Since the range of $D_U$ contains the supersinglet $\ket{S_n}$, then $A_U=D_U-S_n$ is the projector onto the remaining part of that range. To see this explicitly, $D_U\ket{S_n}=\ket{S_n}$ and Hermiticity imply
	\begin{equation}
		D_US_n=S_nD_U=S_n.
	\end{equation}
	Using $D_U^2=D_U$ and $S_n^2=S_n$, we obtain
	\begin{equation}
		A_U^\dagger=A_U,\qquad
		A_U^2=D_U^2-D_US_n-S_nD_U+S_n^2=D_U-S_n=A_U.
	\end{equation}
	Moreover, $A_US_n=S_nA_U=0$, so
	\begin{equation}
		0\preceq A_U\preceq\openone-S_n.
	\end{equation}
	Thus, on the supersinglet's orthogonal complement, $A_U$ is a positive contraction: all its eigenvalues lie in $[0,1]$ (in fact, they are zero or one). The spectral bound proved in Appendix~\ref{App_GME-dimension_witness_bound} gives
	\begin{equation}\label{eq:fs_mean_A}
		0\preceq\mathbb E [A_U]=W-S_n\preceq\frac{\openone-S_n}{n+1}.
	\end{equation}
	Let $a=\|W_N-S_n\|_\infty\leq1$. For any $\rho\in\mathcal G_r$, the supersinglet fidelity bound $\tr(\rho S_n)\leq r/n$ implies the deterministic inequality
	\begin{equation}\label{eq:fs_actual_bound}
		\tr(\rho W_N)\leq\tr(\rho S_n)+a[1-\tr(\rho S_n)]
		\leq\frac rn+\left(1-\frac rn\right)\|W_N-S_n\|_\infty.
	\end{equation}
	Thus a violation of $\tr(\rho W_N)\leq x$ by any state in $\mathcal{G}_r$ requires $a>(nx-r)/(n-r)$. Applying the binary-relative-entropy form of matrix Chernoff \cite[Theorem~5.1]{FS_Tropp2012} to the $N$ independent contractions $A_{U_i}$, in dimension $n^n-1$, gives
	\begin{equation}\label{eq:fs_matrix_chernoff}
		\mathbb P\!\left\{\left\|\frac1N\sum_i A_{U_i}\right\|_\infty>z\right\}
		\leq(n^n-1)e^{-N D_{\text{KL}}(z\Vert1/(n+1))},
		\qquad \frac1{n+1}\leq z\leq1.
	\end{equation}
	Substituting $z=(nx-r)/(n-r)$ proves Eq.~\eqref{eq:fs_list_tail}.
	
	\item Finite number of trials $T$ to estimate expectation value.
	
	Assume that the realised list satisfies $\tr(\rho W_N)\leq x$ for every $\rho\in\mathcal G_r$. In order to correctly estimate the expectation value, we would in principle need an infinte number of trials $T$. However, using the scalar Chernoff bound \cite{FS_Chernoff1952}, we can provide an upper bound on the $\text{p-value}$:
	\begin{equation}\label{eq:fs_shot_tail_main}
		\mathbb P\{K/T\geq w\mid\mathrm{list},\rho\}
		\leq e^{-T D_{\text{KL}}(w\Vert x)},\qquad w>x.
	\end{equation}

\end{enumerate}

Under the null hypothesis, a score at least $w$ can arise either from a list that permits a null-state mean above $x$, or from measurement fluctuations on a list whose null-state means are bounded by $x$. To combine the two sources of uncertainty, fix a deterministic $x\in[b_r,w]$ and define
\begin{equation}
	F_x=\left\{\sup_{\sigma\in\mathcal G_r}\tr(\sigma W_N)>x\right\},
	\qquad E_w=\{K/T\geq w\}.
\end{equation}
The event $F_x$ means that the list fails to bound every null-state mean by $x$.  Splitting the observed-score event gives
\begin{equation}
	\mathbb P(E_w)=\mathbb P(E_w\cap F_x)+\mathbb P(E_w\cap F_x^c).
\end{equation}
The first term is at most $\mathbb P(F_x)$, bounded by Eq.~\eqref{eq:fs_list_tail}. For the second, condition on the list and apply Eq.~\eqref{eq:fs_shot_tail_main} on $F_x^c$:
\begin{equation}
	\begin{aligned}
		\mathbb P(E_w\cap F_x^c)
		&=\mathbb E_{\mathrm{list}}\!\left[
		\mathbf 1_{F_x^c}\,\mathbb P(E_w\mid\mathrm{list},\rho)\right] \leq e^{-T D_{\text{KL}}(w\Vert x)}\mathbb P(F_x^c)
		\leq e^{-T D_{\text{KL}}(w\Vert x)}.
	\end{aligned}
\end{equation}
Here $\mathbf 1_{F_x^c}$ is the indicator of $F_x^c$, and $\rho$ is fixed conditional on the list.  Combining the two bounds, with $T=NR$, yields
\begin{equation}
	\mathbb P\{K/T\geq w\}\leq
	(n^n-1)e^{-N D_{\text{KL}}((nx-r)/(n-r)\Vert1/(n+1))}
	+e^{-NR D_{\text{KL}}(w\Vert x)}.
\end{equation}
Every deterministic $x\in[b_r,w]$ bounds the same upper-tail probability. Taking the infimum and setting $p_r(w)=1$ for $w\leq b_r$
\begin{equation}\label{eq:fs_pvalue}
	p_r(w)=\inf_{b_r\leq x\leq w}\left\{
	(n^n-1)e^{-N D_{\text{KL}}(\frac{nx-r}{n-r}\Vert\frac1{n+1})}+e^{-NR D_{\text{KL}}(w\Vert x)}\right\}.
\end{equation}
At $x=w$, the measurement bound is the trivial value one. The quantity $p_r(w)$ includes the randomness in both the list and the measurements.  For a significance level $\eta\in(0,1)$ chosen in advance, rejecting $\rho\in\mathcal G_r$ when $p_r(w)\leq\eta$ certifies GME-dimension at least $r+1$, with false-certification probability at most $\eta$ under the null hypothesis. To see this,  consider that a rejection at an $\eta$-level for fixed $\eta\in(0,1)$ is possible,  or equivalently that there exists a minimal integer $k_*$ with $p_r(k_*/T)\leq\eta$. Then
\begin{equation}
	\mathbb P\{p_r(K/T)\leq\eta\}
	=\mathbb P\{K\geq k_*\}=\mathbb P\{K/T\geq k_* /T\}
	\leq p_r(k_*/T)\leq\eta.
\end{equation}
The optimisation over $x$ balances the two uncertainty contributions.  Notably,  at fixed $N$ and $w$, increasing $R$ improves only the measurement contribution; while the list term at any fixed $x$ is unchanged.

As a benchmark to illustrate the finite-setting effect,  we have taken the white-noise supersinglet $\rho_v=vS_n+(1-v)\openone/n^n$, whose exact score is $w(v)=v+(1-v)n!/n^n$. Figure~\ref{fig:fs_visibility} then shows the minimum integer $N$ for which Eq.~\eqref{eq:fs_pvalue}, evaluated at this score, is at most $10^{-4}$ with $n=8$ and $R=10^3$. No finite $N$ suffices when the visibility is lower than the continuum bound
\begin{equation}\label{eq:fs_visibility}
	v\leq v_r
\end{equation}
Above this threshold the required list size is finite and diverges as $v\to v_r$. 

\begin{figure}[t]
	\centering
	\includegraphics[width=0.5\textwidth]{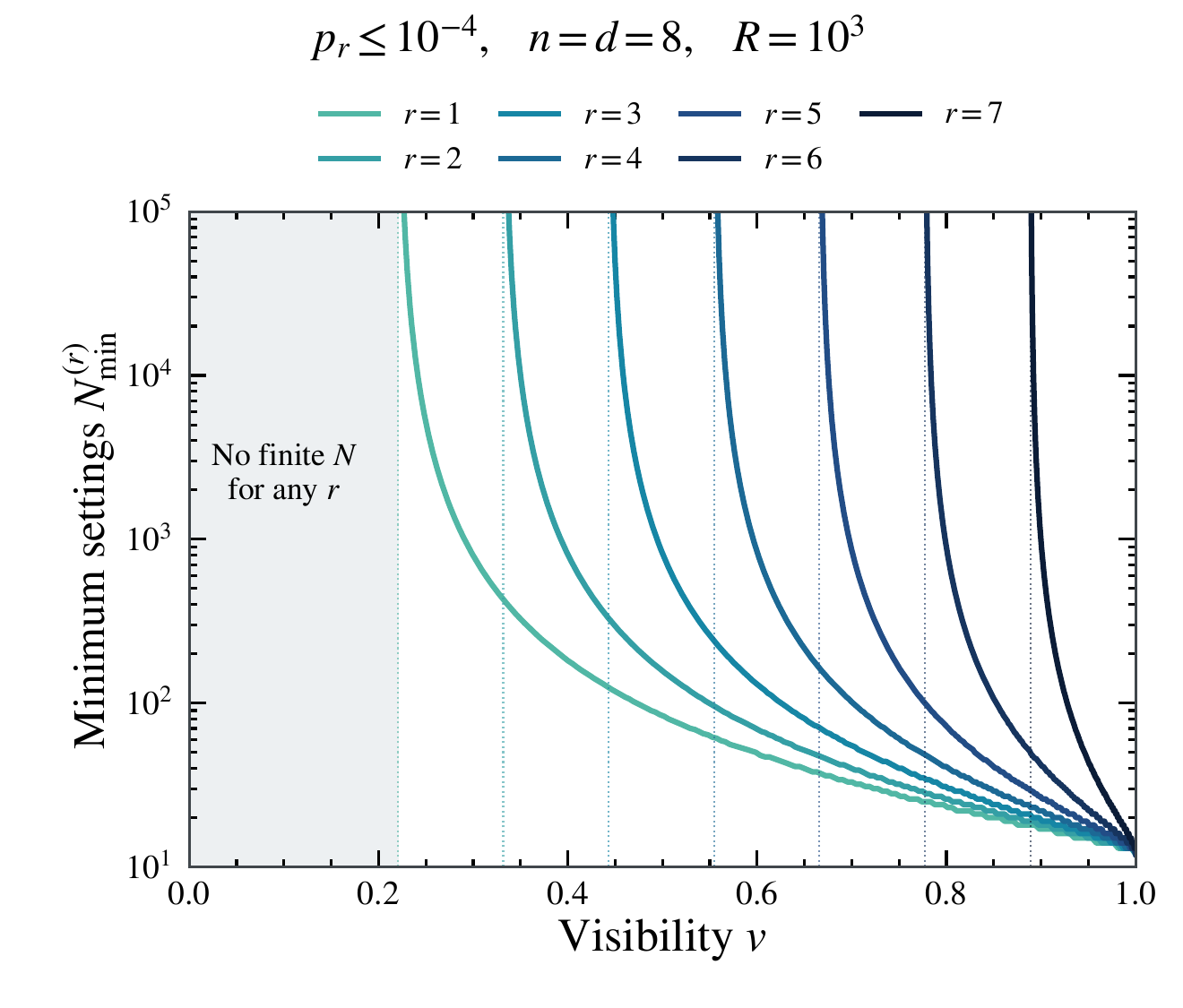}
	\caption{Minimum number of settings required for $p_r\leq10^{-4}$ at the exact white-noise supersinglet score, with $n=8$ and $R=10^3$. Curves $r=1,\ldots,7$ certify GME-dimension at least $r+1$. Dotted lines mark the continuum bound $v_r$, so no finite $N$ suffices to certify $r$ at their left. }
	\label{fig:fs_visibility}
\end{figure}

\section{Proof of Lemma \ref{lemma1}}
\label{app:uniqueground}

In this appendix, we prove that the Casimir operator has a single zero-mode which corresponds to the supersinglet, i.e. $\mathcal{C}\ket{\psi} = \sum_{j=1}^{n^2-1}R_{j}^2 \ket{\psi} = 0$ iff $\ket{\psi}$ is the supersinglet state (up to a phase factor). Let us start showing that the set of operators $R_j = \sum_{i=1}^{n} \openone^{\otimes (i-1)} \otimes \lambda_{j} \otimes \openone^{\otimes (n-i)} = \sum_{i=1}^{n} \lambda_{j}^{(i)}$ is a valid set of generators for the subgroup of unitaries $\lbrace U^{\otimes n}: \, U \in SU(n) \rbrace$. To prove this, let us compute
\begin{equation}\label{generators_J}
	R_{a}R_{b} = \sum_{p=1}^{n}\sum_{q=1}^{n}\lambda_{a}^{(p)}\lambda_{b}^{(q)} = \sum_{p=1}^{n} \lambda_{a}^{(p)}\lambda_{b}^{(p)} + \sum_{p\neq q} \lambda_{a}^{(p)}\lambda_{b}^{(q)}, \quad R_{b}R_{a} = \sum_{p=1}^{n}\sum_{q=1}^{n}\lambda_{b}^{(p)}\lambda_{a}^{(q)} = \sum_{p=1}^{n} \lambda_{b}^{(p)}\lambda_{a}^{(p)} + \sum_{p\neq q} \lambda_{b}^{(p)}\lambda_{a}^{(q)},
\end{equation}
and notice that by relabeling of the indices we can write
\begin{equation}
	\sum_{p\neq q} \lambda_{a}^{(p)}\lambda_{b}^{(q)} = \sum_{p < q} \lambda_{a}^{(p)}\lambda_{b}^{(q)} + \sum_{p > q} \lambda_{a}^{(q)}\lambda_{b}^{(p)} = \sum_{q < p} \lambda_{a}^{(q)}\lambda_{b}^{(p)} + \sum_{q > p} \lambda_{a}^{(p)}\lambda_{b}^{(q)},
\end{equation}
which is exactly the second term appearing in the product $R_{b}R_{a}$ in \eqref{generators_J}. Using that $[\lambda_a,\lambda_b] = i \sum_{c=1}^{n^2-1}f_{abc} \lambda_c$, where $f_{abc}$ is the structure constant of the $su(n)$ algebra, the commutator is:
\begin{equation}
	[R_{a},R_{b}] = \sum_{p=1}^{n} [\lambda_{a}^{(p)},\lambda_{b}^{(p)}] = i \sum_{c=1}^{n^2-1} f_{abc} \sum_{p=1}^{n} \lambda_{c}^{(p)} = i \sum_{c=1}^{n^2-1} f_{abc} R_{c},
\end{equation}
which is the condition that a proper algebra of $\lbrace U^{\otimes n}: \, U \in SU(n) \rbrace$ would satisfy. Therefore, we can write any of those collective unitaries as $U^{\otimes n} = e^{-i\sum_{j} \theta_j R_j}$ for a set of coefficients $\theta_j$. From the definition of supersinglet state, we get
\begin{equation}\label{invariance_Lie}
	U^{\otimes n} \ket{\psi} = \ket{\psi} \quad \iff \quad e^{-i\sum_{j} \theta_j R_j} \ket{\psi} = \ket{\psi},
\end{equation}
Differentiating \eqref{invariance_Lie} with respect to $\theta_j$ on both sides, we get $R_{j}\ket{\psi} = 0$ for each $j$. Therefore, also $\sum_{j=1}^{n^2-1}R_{j}^2 \ket{\psi} = 0$ holds. For the uniqueness proof, we notice that, since all $R_{j}^2$ are positive semidefinite, $\mathcal{C}\ket{\psi} = 0$ implies $R_{j} \ket{\psi} = 0$ for each $j$, from which \eqref{invariance_Lie} follows. Since from Appendix \ref{app:unique} we know that $\ket{S_n}$ is the unique state for which \eqref{invariance_Lie} holds, we also conclude that the supersinglet is the unique state corresponding to the zero-mode of the Casimir.

\section{Proof of entanglement depth witness}\label{entdeptwit_app}
Consider the Casimir operator
\begin{equation}\label{Casimir_operator_app}
	\mathcal{C} = \sum_{\alpha=1}^{n^2-1} R_{\alpha}^{2}, \quad \text{with } R_{\alpha} = \sum_{i=1}^{n} \lambda_{\alpha}^{(i)},
\end{equation}
where $\lbrace \lambda_{\alpha} \rbrace_{\alpha=1}^{n^2-1}$ are the traceless Hermitian generators of the $SU(n)$ group that satisfy $\tr(\lambda_{\alpha}\lambda_{\beta}) = 2\delta_{\alpha,\beta}$, and $\lambda_{\alpha}^{(i)} = \openone^{\otimes (i-1)} \otimes \lambda_{\alpha}^{(i)} \otimes \openone^{\otimes (n-i)}$. In this appendix, we prove that for any $k$-producible state $\rho$ of $n$ particles with $n$ local levels each, it holds that
\begin{equation}\label{witness_ent_depth_bound}
	\expect{\mathcal{C}}_\rho \geq 2\left(n^2-\left\lfloor \frac{n}{k}\right\rfloor k^2-s^2\right),
\end{equation}
where $s = n (\text{mod }k)$. The proof consists of three steps:
\begin{enumerate}
	\item Rewriting the Casimir operator.
	
	By expanding the square in \eqref{Casimir_operator_app}, we get
	\begin{equation}
		C = \sum_{\alpha = 1}^{n^2-1} \left[\sum_{i=1}^{n}(\lambda_{\alpha}^{(i)})^{2} + \sum_{i\neq j} \lambda_{\alpha}^{(i)} \otimes \lambda_{\alpha}^{(j)}\right] = \sum_{\alpha = 1}^{n^2-1} \left[\sum_{i=1}^{n}(\lambda_{\alpha}^{(i)})^{2} + 2\sum_{i< j} \lambda_{\alpha}^{(i)} \otimes \lambda_{\alpha}^{(j)}\right].
	\end{equation}
	Define $\Gamma_{1} \coloneqq \sum_{\alpha=1}^{n^2-1}(\lambda_{\alpha})^{2}$ and $\Gamma_{2}^{i,j} \coloneqq \sum_{\alpha=1}^{n^2-1}(\lambda_{\alpha}^{i} \otimes \lambda_{\alpha}^{j})$, where in $\Gamma_{2}^{i,j}$ we are discarding all the identity operators in the positions different from $i,j$, i.e. $\Gamma_{2}^{i,j} \in \mathcal{H}_{i} \otimes \mathcal{H}_{j}$. Notice that
	\begin{equation}
		\sum_{\alpha=1}^{n^2-1}[\lambda_{\alpha}^2,\lambda_{\beta}] = \sum_{\alpha=1}^{n^2-1}(\lambda_{\alpha}[\lambda_{\alpha},\lambda_{\beta}] + [\lambda_{\alpha},\lambda_{\beta}]\lambda_{\alpha}) = \sum_{\alpha=1}^{n^2-1} (i \, \sum_{c=1}^{n^2-1} f_{\alpha \beta c} \, \lambda_{\alpha}\lambda_{c} + i \, \sum_{c=1}^{n^2-1} f_{\alpha \beta c} \, \lambda_{c} \lambda_{\alpha}) = 0,
	\end{equation}
	where we have used the antisymmetry of the structure constant $f_{\alpha \beta c}$. Therefore, from the Lie expansion $U = e^{-i \sum_{\beta}\theta_{\beta} \lambda_{\beta}}$, it also follows that $[\Gamma_1, U] = 0$ for each $U \in SU(n)$.
	
	Therefore, from Schur's lemma \cite{fulton1991representation}, $\Gamma_{1} = \gamma \openone$, for some coefficient $\gamma$. In order to find the coefficient, compute
	\begin{equation}
		\tr(\Gamma_{1}) = \gamma \tr(\openone) = \gamma n = \sum_{\alpha=1}^{n^2-1}\tr(\lambda_{\alpha}^2) = 2(n^2-1) \quad \implies \quad \gamma = \frac{2(n^2-1)}{n}.
	\end{equation}
	Similarly, notice that $[\Gamma_{2}^{i,j},U \otimes U] = 0$ for each $U \in SU(n)$. From Schur-Weyl duality and Schur's lemma, we get
	\begin{equation}
		\Gamma_{2}^{i,j} = a \openone \otimes \openone + b S_{ij},
	\end{equation}
	where $S_{ij}$ is the SWAP operator on subsystems $\mathcal{H}_{i}$ and $\mathcal{H}_{j}$ such that $S_{ij} \ket{pq} = \ket{qp}$. To find the coefficients $a,b$, we can compute $\tr\small(\Gamma_{2}^{i,j})$ and $\tr\small(S_{ij}\Gamma_{2}^{i,j})$, i.e.
	\begin{equation}
		\begin{aligned}
			&\tr\small(\Gamma_{2}^{i,j}) = a \tr(\openone \otimes \openone) + b \tr(S_{ij}) = a n^2 + b n = \sum_{\alpha=1}^{n^2-1}\tr(\lambda_{\alpha}^{i} \otimes \lambda_{\alpha}^{j}) = 0 \quad \implies \quad a = -\frac{b}{n},\\
			&\tr\small(S_{ij}\Gamma_{2}^{i,j}) = a\tr(S_{ij}) + b\tr(\openone \otimes \openone) = a n + b n^2 = \sum_{\alpha=1}^{n^2-1}\tr(S_{ij}(\lambda_{\alpha}^{i} \otimes \lambda_{\alpha}^{j})) = \sum_{\alpha=1}^{n^2-1}\tr(\lambda_{\alpha}^2) = 2(n^2-1) \quad \implies \quad b = 2,
		\end{aligned}
	\end{equation}
	where we have used that $\tr(S_{ij}(A \otimes B)) = \tr(AB)$. We get
	\begin{equation}
		\begin{aligned}
			\Gamma_{1} &= \frac{2(n^2-1)}{n}\openone, \quad \Gamma_{2}^{i,j} = -\frac{2}{n} \openone + 2 S_{ij}.
		\end{aligned}
	\end{equation}
	Therefore, inserting $\Gamma_{1}$ and $\Gamma_{2}$ into the corresponding locations in $\mathcal{H}^{\otimes n}$, the Casimir operator becomes
	\begin{equation}\label{Casimir_operator_2app}
		\begin{aligned}
			\mathcal{C} &= \sum_{i=1}^{n} \frac{2(n^2-1)}{n}\openone + 2\sum_{i<j} \left(-\frac{2}{n} \openone + 2 \mathbf{S}_{ij}\right)\\
			&= 2(n^2-1)\openone - 2(n-1) \openone + 4\sum_{i <j} \mathbf{S}_{ij} = 2n(n-1) \openone + 4\sum_{i<j}\mathbf{S}_{ij},
		\end{aligned}
	\end{equation}
	where the operator $\mathbf{S}_{ij}$ acts as a SWAP on subsystems $\mathcal{H}_{i} \otimes \mathcal{H}_{j}$ and as an identity over the remaining subsystems of $\mathcal{H}^{\otimes n}$.
	
	\item Expectation value over $k$-producible states.
	
	Since any $k$-producible mixed state is a convex combination of $k$-producible pure states, $\rho = \sum_{i} q_i \ketbra{\psi_i}{\psi_i}$, we have that $\tr(\mathcal{C}\rho) \geq \tr(\mathcal{C} \ketbra{\psi}{\psi})$. Therefore, to find the bound we only need to consider the expectation value over a pure $k$-producible state, $\ket{\psi} = \bigotimes_{p=1}^{m} \ket{\phi}_{S_p}$ with $|S_p| = c(p) \leq k$,
	\begin{equation}\label{Casimir_expect}
		\expect{\mathcal{C}}_{\ket{\psi}} = 2n(n-1) + 4 \sum_{i < j} \expect{\mathbf{S}_{ij}}_{\ket{\psi}}.
	\end{equation}
	Depending on whether the indices $(i,j)$ live in the same set $S_{p}$, we get a different bound on $\mathbf{S}_{ij}$.
	\begin{itemize}
		\item Indices in different sets: in this case, the action of the operator can be seen as
		\begin{equation}
			\expect{\mathbf{S}_{ij}}_{\ket{\psi}} = \tr(S_{ij}(\rho_i \otimes \rho_j)) = \tr(\rho_{i}\rho_{j}) \geq 0,
		\end{equation}
		where $\rho_{i}$ and $\rho_{j}$ are the reductions of the state $\ket{\psi}$ over the subsystems $\mathcal{H}_i$ and $\mathcal{H}_{j}$, respectively.
		
		\item Indices in the same set: in this case, the state over the subsystems $\mathcal{H}_i \otimes \mathcal{H}_{j}$ can be entangled, therefore, since the eigenvalues of $S_{ij}$ are $\pm 1$, $\expect{\mathbf{S}_{ij}}_{\ket{\psi}} \geq -1$. We get
		\begin{equation}
			\sum_{i < j} \expect{\mathbf{S}_{ij}}_{\ket{\psi}} = \sum_{p=1}^{m} \sum_{p_j=1}^{c(p)}\sum_{p_i < p_j} \expect{\mathbf{S}_{ij}}_{\ket{\psi}} \geq - \sum_{p=1}^{m}\frac{c(p)(c(p)-1)}{2},
		\end{equation}
		where we have divided the sum over all indices $i,j$ into the sum over the indices $p_{i},p_{j}$ which are contained in the same set $S_{p}$.
	\end{itemize}
	Using the bounds, we get
	\begin{equation}
		\expect{\mathcal{C}} \geq 2n(n-1) - 4\sum_{p=1}^{m}\frac{c(p)(c(p)-1)}{2} = 2\left(n^2 - \sum_{p=1}^{m}c(p)^2\right).
	\end{equation}
	
	\item Maximisation of the squared lenghts sum.
	
	To prove the result, we need to solve the following maximisation problem:
	\begin{equation}
		\begin{aligned}
			\max_{c(p)} \,& \sum_{p=1}^{m}c(p)^{2}\\
			\quad \text{s.t. }&c(p) \leq k, \quad \forall p,\\
			\quad \quad& \sum_{p=1}^{m} c(p) = n.
		\end{aligned}
	\end{equation}
	The maximum is achieved when each set contains the largest possible number of elements. We can divide the total number of subspaces $n$ into $\left\lfloor \frac{n}{k} \right\rfloor$ groups of $k$ particles, such that $m = \left\lfloor \frac{n}{k} \right\rfloor + 1$ and $n = \sum_{p=1}^{m} c(p) = \left\lfloor \frac{n}{k} \right\rfloor k + s$, where $s = n(\text{mod }k)$ is the number of particles left. Therefore,
	\begin{equation}
		\expect{\mathcal{C}} \geq 2\left(n^2 - \sum_{p=1}^{m}c(p)^2\right) \geq 2\left(n^2 - \left\lfloor \frac{n}{k} \right\rfloor k^2 - s^2\right).
	\end{equation}
\end{enumerate}
Now we can show that this bound is tight. Consider the case in which $c(p) = k$ for $p=1,\dots,\left\lfloor \frac{n}{k} \right\rfloor$ and $c(\left\lfloor \frac{n}{k} \right\rfloor+1) = s$ and construct the state
\begin{equation}
	\ket{\psi^{\text{opt}}} = \underbrace{\ket{S_{k}} \oplus\dots \oplus \ket{S_{k}}}_{\left\lfloor \frac{n}{k} \right\rfloor \text{ times}} \oplus \ket{S_s},
\end{equation}
where $\ket{S_k}$ and $\ket{S_s}$ are the supersinglet states in dimension $k$ and $s$, respectively. When the indices $i,j$ are in the same block, $\mathbf{S}_{ij}\ket{\psi^{\text{opt}}} = -\ket{\psi^{\text{opt}}}$ by definition. Therefore, $\expect{\mathbf{S}_{ij}}_{\ket{\psi^{\text{opt}}}} = -1$. When they are not in the same block, the reduced one-particle states $\rho_{i}$ and $\rho_{j}$ live in two orthogonal vector spaces; then, $\expect{\mathbf{S}_{ij}} = \tr(S_{ij} (\rho_{i} \otimes \rho_{j})) = \tr(\rho_{i}\rho_{j}) = 0$. Therefore, this state saturates the bound in \eqref{witness_ent_depth_bound}.

\section{Proof of GME-dimension witness bound}\label{App_GME-dim_Casimir}

In this appendix, we prove that the Casimir operator $\mathcal{C} = \sum_{j=1}^{n^2-1}R_{j}^2$ with $R_{j} = \sum_{i=1}^{n} \lambda_{j}^{(i)}$, is such that
\begin{equation}
	\mathcal{C} \succeq 4n(\openone - S_n),
\end{equation}
where $\ket{S_n}$ is the supersinglet state. This bound was used in Result \ref{GME-dimension_witness} to derive the GME-dimension witness. Since $[R_a,R_b] = i f_{abc}\sum_{c=1}^{n^2-1} R_{c}$ (see Appendix \ref{app:uniqueground}), we have that
\begin{equation}
	\left[\sum_{a=1}^{n^2-1}R_{a}^2,R_{b}\right] = \sum_{a=1}^{n^2-1} [R_{a}^2,R_{b}] = \sum_{a=1}^{n^2-1} \left([R_{a},R_{b}] R_{a} + R_{a}[R_{a},R_{b}]\right) = \sum_{a=1}^{n^2-1} \left(i\sum_{c=1}^{n^2-1}f_{abc}R_{c}R_{a} + i\sum_{c=1}^{n^2-1}f_{abc}R_{a}R_{c}\right) = 0,
\end{equation}
where we have used that the structure constant $f_{abc}$ is fully antisymmetric in its indices. Therefore, the Casimir operator is such that $[\mathcal{C},U^{\otimes n}] = 0$. Furthermore, since each $R_{j}$ is a collective observable of the same spin on each particle, for any permutation operator $P_{\pi}$, we have that $[R_{j},P_{\pi}] = 0$ for all $j=1,\dots,n^2-1$. Therefore, also $[\mathcal{C},P_{\pi}] = 0$ for all possible permutations. From Schur-Weyl duality, any element $U^{\otimes n}$ can be written as
\begin{equation}
	U^{\otimes n} = \bigoplus_{\substack{\lambda \vdash n\\ l(\lambda) \leq n}} U_{\lambda}(U) \otimes I_{V_{\lambda}},
\end{equation}
where $\lambda$ denotes all possible integer partitions of $n$ with length $l(\lambda) \leq n$, $U_{\lambda}(U)$ is the matrix representation of the group element on the space of the irreducible representation of $SU(n)$, and $I_{V_{\lambda}}$ is the identiy on the associated irreducible representation space, $V_{\lambda}$, of the symmetric group $\mathcal{S}_n$. Therefore, using Schur's lemma \cite{fulton1991representation}, also the Casimir can be written in a block-diagonal form as
\begin{equation}
	\mathcal{C} = \bigoplus_{\substack{\lambda \vdash n\\ l(\lambda) \leq n}} c_{\lambda} I_{Q_{\lambda}} \otimes I_{V_{\lambda}},
\end{equation}
where $I_{Q_{\lambda}}$ and $I_{V_{\lambda}}$ are the identity acting on the irreducible representation space of $SU(n)$ and on $V_{\lambda}$, respectively, and $c_{\lambda}$ are real coefficients. These coefficients are the spectrum of the Casimir operator and can be computed using
\begin{equation}\label{Casimir_eigenvalues}
	c_{\lambda} = 2(\vec{\mu}_{\lambda} \cdot \vec{\mu}_{\lambda} + 2 \vec{\mu}_{\lambda} \cdot \vec{\sigma}),
\end{equation}
where $\vec{\mu}_{\lambda} = (\mu_{1,\lambda}^{\max},\dots,\mu_{n-1,\lambda}^{\max})$ is the highest weight in a given representation $\lambda$ (i.e. the vector of the maximum eigenvalues $\mu_{i}^{\max}$ of the $n-1$ generators of the Cartan subalgebra) and $\vec{\sigma}$ is the Weyl vector \cite{mkrtchyan2025casimireigenvaluesadotimesk,GotoGrosshans1978,DiFrancescoMathieuSenechal1997}. Define as $\lbrace \vec{\omega}_{i} \rbrace_{i=1}^{n-1}$ the basis of the fundamental weights, such that
\begin{equation}\label{fundweight_basis}
	\vec{\omega}_i = \sum_{k=1}^{i} \vec{e}_{k} - \frac{i}{n} \sum_{k=1}^{n} \vec{e}_{k}, \quad \text{s.t. } \vec{\omega}_{i}\cdot \vec{\omega}_{j} = \min(i,j) - \frac{ij}{n},
\end{equation}
where $\lbrace \vec{e}_{k} \rbrace_{k=1}^{n}$ is an orthonormal basis \cite{huang2006dirac}. Then, both $\vec{\mu}$ and $\vec{\sigma}$ can be expressed in this basis as:
\begin{equation}\label{high_weight_weyl_fund}
	\vec{\mu}_{\lambda} = \sum_{k=1}^{n-1} a_{k,\lambda} \vec{\omega}_k, \quad \vec{\sigma} = \sum_{k=1}^{n-1}\vec{\omega}_k,
\end{equation}
where $\lbrace a_{k,\lambda} \rbrace_{k=1}^{n-1}$ are the Dynkin labels \cite{Georgi:1999wka}. They are strictly connected to the geometry of the Young tableau associated with a given representation, since they are defined as $a_{k,\lambda} = \lambda_{k} - \lambda_{k+1}$, i.e. the difference between the number of boxes in two subsequent rows. Using \eqref{fundweight_basis}, we can express both $\vec{\mu}_{\lambda}$ and $\vec{\sigma}$ as a function of the geometry of the Young tableaux:
\begin{equation}\label{high_weight_Young}
	\begin{aligned}
		\vec{\mu}_{\lambda} &= \sum_{i=1}^{n-1} a_{i,\lambda} \vec{\omega}_{i} = \sum_{i=1}^{n-1}(\lambda_i - \lambda_{i+1})\vec{\omega}_i = \lambda_1 \vec{\omega}_1 + \sum_{j=2}^{n} \lambda_{j}(\vec{\omega}_j - \vec{\omega}_{j-1})\\
		&= \lambda_1\left(\vec{e}_1 - \frac{1}{n}\sum_{k=1}^{n}\vec{e}_{k}\right) + \sum_{j=2}^{n}\lambda_{j}\left(\sum_{k=1}^{j}\vec{e}_{k} - \frac{j}{n}\sum_{k=1}^{n}\vec{e}_{k} - \sum_{k=1}^{j-1}\vec{e}_{k} + \frac{j-1}{n}\sum_{k=1}^{n}\vec{e}_{k}\right)\\
		&= \sum_{k=1}^{n}\lambda_{k}\vec{e}_{k}-\frac{1}{n}\sum_{k=1}^{n}\left(\sum_{j=1}^{n}\lambda_{j}\right)\vec{e}_{k} = \sum_{k=1}^{n}\left(\lambda_{k}-1\right)\vec{e}_{k},
	\end{aligned}
\end{equation}
where we have used that $\sum_{j=1}^{n} \lambda_j = n$, i.e. the total number of boxes in the Young tableau is $n$. In the same way,
\begin{equation}\label{Weyl_Young}
	\begin{aligned}
		\vec{\sigma} &= \sum_{i=1}^{n-1}\vec{\omega}_i = \sum_{i=1}^{n}\left(\sum_{k=1}^{i}\vec{e}_{k} - \frac{i}{n}\sum_{k=1}^{n}\vec{e}_{k}\right)\\
		&= \vec{e}_1 + \sum_{k=1}^{2}\vec{e}_2 + \dots + \sum_{k=1}^{n-1}\vec{e}_{k} - \frac{1}{n}(1+2+\dots+n-1)\sum_{k=1}^{n}\vec{e}_{k}\\
		&= (n-1)\vec{e}_{1} + (n-2)\vec{e}_{2} + \dots + \vec{e}_{n-1} - \frac{(n-1)n}{2n}\sum_{k=1}^{n}\vec{e}_{k}\\
		&= \sum_{k=1}^{n-1}(n-k)\vec{e}_{k} - \frac{n-1}{2}\sum_{k=1}^{n}\vec{e}_{k} = \sum_{k=1}^{n}\left(\frac{n+1}{2} - k\right)\vec{e}_{k}.
	\end{aligned}
\end{equation}
Therefore, using \eqref{high_weight_Young} and \eqref{Weyl_Young}, we get:
\begin{equation}
	\vec{\mu}_{\lambda}+ 2\vec{\sigma} = \sum_{k=1}^{n}\left(\lambda_k - 1 + n+1 - 2k\right)\vec{e}_{k},
\end{equation}
from which \eqref{Casimir_eigenvalues} becomes:
\begin{equation}\label{Casimir_eigenvalues_Young}
	c_{\lambda} = 2(\vec{\mu}_{\lambda} \cdot \vec{\mu}_{\lambda} + 2 \vec{\mu}_{\lambda} \cdot \vec{\sigma}) = 2\vec{\mu}_{\lambda} \cdot (\vec{\mu}_{\lambda} + 2\vec{\sigma}) = 2 \sum_{k=1}^{n}(\lambda_k - 1)\left(\lambda_k - 1 + n+1 - 2k\right).
\end{equation}
By looking at all possible irreducible representations of $SU(n)$, we have:
\begin{equation}
	\lambda = \lbrace (1,1,\dots,1), (2,1,\dots,1), \dots, (n-1,1), (n) \rbrace,
\end{equation}
we can compute the associated eigenvalues $c_{\lambda}$. The minimum eigenvalue is obtained for the supersinglet representation $(1,1,\dots,1)$, indeed $c_{(1,1,\dots,1)} = 0$. Since from \eqref{Casimir_eigenvalues_Young}, the eigenvalues increase with the length of the rows in the Young tableaux, the second smallest eigenvalue is obtained for the representation $(2,1,\dots,1)$, i.e.
\begin{equation}
	c_{(2,1,\dots,1)} = 2 (2-1)(2+n-2) + 2 (0-1)(0+n-2n) = 2n + 2n = 4n.
\end{equation}
Therefore, since $c_{\lambda} \geq 4n$ for all $\lambda \neq (1,1,\dots,1)$, we can write that
\begin{equation}
	\mathcal{C} \succeq 0 \ketbra{S_n}{S_n} + c_{\lambda}^{\min}(\openone - \ketbra{S_n}{S_n}) = 4n(\openone - \ketbra{S_n}{S_n}).
\end{equation}

\section{Metrology with a single supersinglet state}

\label{app:metsingle}


\noindent Since all one-particle reduced states are maximally mixed, $\rho_i = \mathds 1/d$, we have that $\expval{\xi^{(i)}}=0$ and $\expval{(\xi^{(i)})^2}=\tr(\xi^2)/d$. The evolved state is given by 
\begin{equation}
	\ket{\psi_g} =e^{-itH_g}\ket{S_n}.
\end{equation}We define the rescaled parameter $\theta=gt$ and the Hamiltonian $H'= H_{g}/g$. The QFI for estimating $\theta$ with $\ket{\psi_g}$ is given by 
\begin{equation}\label{eq:fish}
	F_Q =4 \Delta H'^2 = 4\left(\expval{H'^2}_{\ket{S_n}} - \expval{H'}^2_{\ket{S_n}}\right) = 4\expval{H'^2}_{\ket{S_n}},
\end{equation}where we have used the fact that the expectation value $\expval{H'}_{\ket{S_n}}=0$ as the observables $\xi^{(i)}$ are traceless and the one-particle reduced state is maximally mixed. To calculate $\expval{H'^2}_{\ket{S_n}}$, note that 
\begin{equation}\label{eq:h2}
	H'^2 = \sum_i\sum_j x_i x_j\xi^{(i)}\xi^{(j)} = \sum_i x_i^2 (\xi^{(i)})^2 + \sum_{i\neq j} x_ix_j \xi^{(i)}\xi^{(j)}.
\end{equation}The left term leads to a single particle expectation value, $\expval{(\xi^{(i)})^2}_{\rho_i} = \tr(\xi^2)/d$, but the right one involves a two-particle expectation value and requires calculating a two-particle reduced state. We have 
\begin{equation}\label{eq:2part}
	\rho_{ij} =\frac{1}{{n\choose 2}} \Pi^{(2)}_\text{asym} =\frac{1}{{n\choose 2}} \frac{\mathds 1-S}{2} = \frac{\mathds 1-S}{d(d-1)},
\end{equation}where $S$ is the swap operator. Then, the two-particle expectation value,
\begin{equation}
	\expval{\xi^{(i)}\xi^{(j)}}_{\rho_{ij}} = \frac{\tr(\xi^{(i)})\tr(\xi^{(j)}) - \tr(\xi^{(i)}\xi^{(j)}) }{d(d-1)} =- \frac{\tr(\xi^{(i)}\xi^{(j)}) }{d(d-1)}  = -\frac{\tr(\xi^2)}{d(d-1)},
\end{equation} where we have utilised the swap identity $\tr(S(A\otimes B))=\tr(AB)$. Therefore, 
\begin{equation}\label{eq:h1}
	\expval{H'^2}_{\ket{S_n}} = \frac{\tr(\xi^2)}{d}\sum_i x_i^2 -\frac{\tr(\xi^2)}{d(d-1)} \sum_{i\neq j} x_i x_j.
\end{equation}Now, as we mentioned previously, uniform spatial translation has no effect on the estimation problem. Therefore, we may replace $x_i\to x'_i = x_i-\bar x$, where $\bar x = \frac{1}{d}\sum_i x_i$ is the mean location of the particles. Note that $\sum_i x'_i =0$. So,
\begin{equation}\label{eq:xx}
	0 = \left(\sum_i x'_i\right)^2 = \sum_i x_i^{'2} + \sum_{i\neq j} x'_i x'_j \implies   \sum_i x_i^{'2} =- \sum_{i\neq j} x'_i x'_j.
\end{equation}
Substituting the above in Eq. \eqref{eq:h1}, simplifying and using Eq. \eqref{eq:fish}, we obtain the QFI
\begin{equation}\label{eq:fqapp}
	F_Q =4 \frac{\tr(\xi^2)}{d-1}\sum_i x_i^{'2} = 4\frac{\tr(\xi^2)}{d-1}\sum_i (x_i-\bar x)^2 = 4\frac{d\tr(\xi^2)}{d-1}\sigma_x^2.
\end{equation}
where we have defined the spatial variance of the particles $\sigma_x^2 = \frac{1}{d}\sum_i(x_i-\bar x)^2$.

\subsubsection{Equally spaced array}
If the particles are equally spaced with distance $x_0$, then 
\begin{align}
	\sigma_x^2 = x_0^2 \frac{(d^2-1)}{12}.
\end{align}Then, the QFI is
\begin{align}
	F_Q =  4\frac{d\tr(\xi^2)}{d-1}x_0^2 \frac{(d^2-1)}{12} =  \frac{1}{3}d(d+1)\tr(\xi^2)x_0^2, 
\end{align}which seems to grow as $\sim d^2$ with the dimension. However if the total length of the chain $L = (d-1)x_0$ is fixed, then
\begin{align}
	F_Q = \frac{1}{3}\frac{d(d+1)}{(d-1)^2}\tr(\xi^2)L^2,
\end{align}which is fixed for large $d$ and grows only with the total size of the chain.

.

\section{Metrology with macroscopic supersinglet states}

We now consider a macroscopic variant of the original problem. Consider $d$ spatial sites labelled $i=1,\dots,d$ and $k$ qudits on each site labelled $\alpha =1,\dots,k$. There are therefore a total of $n=kd$ qudits. We refer to the index of the qudits at each site as the "layer". Let the field be $f(x) = f_0 + g x$, where $f_0$ is a constant and $g$ is some gradient that needs to be estimated. Let the corresponding Hamiltonian be 
\begin{align}
	H = \sum_{i=1}^d\sum_{\alpha=1}^k (f_0+gx_i)\xi^{(i,\alpha)} \equiv H_0 + H_g
\end{align}where $\xi^{(i,\alpha)}$ is some traceless operator of particle $(i,\alpha)$ and $H_0$ and $H_g$ are the Hamiltonians due to the background field and gradient field, respectively. As before, we define $\theta=gt$ and 
\begin{align}\label{eq:Hp}
	H'=H_g/g
\end{align}

To estimate $\theta$, we consider two types of states
\begin{enumerate}
	\item{A tensor product of supersinglet states,  $\ket{S^k_d} = \bigotimes_{\alpha=1}^k \ket{S_d}_\alpha$. This is a macroscopic supersinglet state, with each layer in a supersinglet state.}
	
	\item{At site $i$, let $V_\pi^{(i)}$ permute $k$ layers and let
		\begin{align}
			P_i = \frac{1}{k!}\sum_{\sigma \in S_k}V_\sigma^{(i)}, \quad P= \bigotimes_{i=1}^d P_i.
		\end{align}Define the state \cite{Arovas2008}
		\begin{align}
			\ket{\mathds S_d^k}  = \frac{P\ket{S_d^k}}{\sqrt{\bra{S_d^k}P\ket{S_d^k}}},
		\end{align} and $\rho^{\mathds S} = \ketbra{\mathds S_d^k}{\mathds S_d^k}$. It is also easy to see why $\ket{S_d^k}$ and $\ket{\mathds S_d^k}$ remain invariant under uniform collective unitary transformations. Each minimal supersinglet is invariant under $U^{\otimes d}$. Therefore, $\ket{S_d^k}$ is invariant under $U^{\otimes dk} = \bigotimes_{i=1}^d\bigotimes_{\alpha=1}^k U^{(i,\alpha)}$. Particle permutations commute with $U^{\otimes dk}$ because the same unitary $U$ acts locally on all particles. Therefore, $U^{\otimes dk}\ket{\mathds S_d^k} =\ket{\mathds S_d^k}$. For any traceless $A$, this implies that $\sum_{i=1}^d\sum_{\alpha=1}^k A^{(i,\alpha)}\ket{\mathds S_d^k} =0$. As a direct consequence of Schur's lemma, the corresponding one-particle reduced states $\rho_1^{S} = \rho_1^{\mathds S} = \mathds 1/d$.}
\end{enumerate}

\subsection{GME-dimension of $\ket{\mathds S_d^k}$}\label{app:GME-dim_macroscopsupersinglet}

In this appendix, we prove that the state $\ket{\mathds S_d^k}$ has GME-dimension $d$. From \cite{Cobucci2024}, we know that the GME-dimension of a pure state is the minimum Schmidt rank across all possible bipartitions, i.e. $\mathcal{D}_{\text{GME}}(\ket{\psi}) = \min_{(B|\bar{B})} \text{SR}(\ket{\psi}_{B|\bar{B}})$. Using that the one-particle reduced state is $\rho_{1}^{\mathds{S}} = \frac{\mathds{1}}{d}$, we already know that $\text{SR}(\ket{\mathds{S}_{d}^{k}}_{1|\text{rest}}) = d$. Now, we need to show that the Schmidt rank across all other bipartitions is larger or equal to $d$.

Since $U^{\otimes n} \ket{\mathds{S}_{d}^{k}} = \ket{\mathds{S}_{d}^{k}}$ for all $U \in SU(d)$, we have that for a generic bipartition $(B|\bar{B})$, $[U_{B},\rho_{B}^{\mathds{S}}] = 0$, where $\rho_{B}^{\mathds{S}} = \tr_{\bar{B}}(\mathds{S}_{d}^{k})$ and $U_{B}$ denotes the unitaries applied on the particles in the subset $B$. Then, using Schur-Weyl duality and Schur's lemma \cite{fulton1991representation}, we can write
\begin{equation}\label{Schur_lemma_rho_B}
	\rho_{B}^{\mathds{S}} = \bigoplus_{\lambda} I_{\mathcal{Q}_{\lambda}} \otimes \mathcal{M}_{\lambda},
\end{equation}
where $\lambda$ denotes all possible irreducible representations of $SU(d)$, $I_{\mathcal{Q}_{\lambda}}$ is the identity on the Schur space $\mathcal{Q}_{\lambda}$ and $\mathcal{M}_{\lambda}$ is an operator on the associated irreducible representation space of $\mathcal{S}_{|B|}$. From \eqref{Qlambda_dim}, we know that for any non-trivial representations, $\dim \mathcal{Q}_{\lambda} > d$. Therefore, we only need to show that $\rho_{B}^{\mathds{S}}$ does not have support on the trivial representation to prove the final result. Assume that $\rank(\rho_{B}^{\mathds{S}}) < d$ for some bipartition $(B|\bar{B})$. Then, since $\rho_{B}^{\mathds{S}}$ has support on the trivial representation, we have that:
\begin{equation}\label{invariance_trivial_rep}
	U_{B} \otimes \openone_{\bar{B}} \ket{\mathds{S}_{d}^{k}} = \ket{\mathds{S}_{d}^{k}}, \quad \forall\, U \in SU(d).
\end{equation}
Consider the operator $\sum_{\alpha=1}^{d^2-1}\lambda_{\alpha}^{(p)}\lambda_{\alpha}^{(q)}$,where $\lbrace \lambda_{\alpha} \rbrace_{\alpha=1}^{n^2-1}$ are the traceless Hermitian generators of the $SU(d)$ group that satisfy $\tr(\lambda_{\alpha}\lambda_{\beta}) = 2\delta_{\alpha,\beta}$, and $\lambda_{\alpha}^{(i)} = \openone^{\otimes (i-1)} \otimes \lambda_{\alpha}^{(i)} \otimes \openone^{\otimes (n-i)}$. Fix $p$ to be a particle in the set $B$ and $q$ a particle in the set $\bar{B}$. Then, since any state $\rho_{pq}$ such that $[U_{p} \otimes \openone_{q},\rho_{pq}] = 0$ for each $U \in SU(d)$ can be written as $\rho_{pq} = \frac{\mathds{1}}{d}\otimes \rho_{q}$ by Schur's lemma, we have that
\begin{equation}\label{unitary_invariant_state_exp_val}
	\sum_{\alpha=1}^{d^2-1} \tr(\rho_{pq} \lambda_{\alpha} \otimes \lambda_{\alpha}) = \frac{1}{d}\tr(\lambda_{\alpha}) \tr(\rho_{q}\lambda_\alpha) = 0,
\end{equation}
However, this cannot happen for the state $\ket{\mathds{S}_{d}^{k}}$. To show this, consider the Casimir operator
\begin{equation}\label{Casimir_dk_supersinglet}
	\begin{aligned}
		\mathcal{C} &= \sum_{\alpha=1}^{d^2-1} \left(\sum_{p=1}^{dk} \lambda_{\alpha}^{(p)}\right)^2 = \sum_{p=1}^{dk} \sum_{\alpha=1}^{d^2-1} \left(\lambda_{\alpha}^{(p)}\right)^2 + 2\sum_{p < q} \sum_{\alpha=1}^{d^2-1} \lambda_{\alpha}^{(p)}\lambda_{\alpha}^{(q)}\\
		&= \sum_{p=1}^{dk} \sum_{\alpha=1}^{d^2-1} \left(\lambda_{\alpha}^{(p)}\right)^2 + 2\sum_{\substack{p < q\\ \text{different copies}}}\sum_{\alpha=1}^{d^2-1} \lambda_{\alpha}^{(p)}\lambda_{\alpha}^{(q)} + 2\sum_{\substack{p < q\\ \text{same copy}}}\sum_{\alpha=1}^{d^2-1} \lambda_{\alpha}^{(p)}\lambda_{\alpha}^{(q)},
	\end{aligned}
\end{equation}
where we have distinguished the case in which $(p,q)$ are two particles in different copies of the supersinglet $\ket{S_{d}}$ from the case in which they live in the same copy. By computing the single expectation values in \eqref{Casimir_dk_supersinglet}, we get
\begin{itemize}
	\item One-particle terms.
	\begin{equation}
		\sum_{p=1}^{dk} \sum_{\alpha=1}^{d^2-1} \bracket{\mathds{S}_{d}^{k}}{\left(\lambda_{\alpha}^{(p)}\right)^2}{\mathds{S}_{d}^{k}} = \sum_{p=1}^{dk} \sum_{\alpha=1}^{d^2-1} \tr(\lambda_{\alpha}^2\rho_{1}^{\mathds{S}}) = \sum_{p=1}^{dk} \sum_{\alpha=1}^{d^2-1} \frac{1}{d}\tr(\lambda_{\alpha}^2) = 2k(d^2-1)
	\end{equation}
	where we have used the normalisation of the $SU(d)$ generators, $\tr(\lambda_{\alpha}\lambda_{\beta}) = 2\delta_{\alpha,\beta}$.
	
	\item Different copies terms.
	
	As shown in Appendix \ref{entdeptwit_app}, we can always write
	\begin{equation}
		\sum_{\alpha=1}^{d^2-1}\lambda_{\alpha}^{(p)}\lambda_{\alpha}^{(q)} = 2 \mathbf{S}_{pq} - \frac{2}{d}\openone.
	\end{equation}
	Therefore,
	\begin{equation}\label{diff_copies_exp_val}
		2\sum_{\substack{p < q\\ \text{different copies}}}\sum_{\alpha=1}^{d^2-1} \bracket{\mathds{S}_{d}^{k}}{\lambda_{\alpha}^{(p)}\lambda_{\alpha}^{(q)}}{\mathds{S}_{d}^{k}} = 2\sum_{\substack{p < q\\ \text{different copies}}}\bracket{\mathds{S}_{d}^{k}}{2 \mathbf{S}_{pq} - \frac{2}{d}}{\mathds{S}_{d}^{k}} = 2\sum_{\substack{p < q\\ \text{different copies}}} \left(2-\frac{2}{d}\right) = 2 d \frac{k(k-1)}{2} \left(2-\frac{2}{d}\right),
	\end{equation}
	where we have used that $\mathbf{S}_{pq}\ket{\mathds{S}_{d}^{k}} = \ket{\mathds{S}_{d}^{k}}$ for $p$ and $q$ being in two different copies of supersinglets (since the state is symmetrised over the copies by definition). Note that this already tells us that for such a choice of $(p,q)$, $\ket{\mathds{S}_{d}^{k}}$ does not satisfy \eqref{unitary_invariant_state_exp_val}.
	
	\item Same copy terms.
	
	Since $U^{\otimes kd} \ket{\mathds{S}_{d}^{k}} = \ket{\mathds{S}_{d}^{k}}$ for all $U \in SU(d)$, using the same procedure shown in Appendix \ref{app:uniqueground}, we can prove that $\mathcal{C}\ket{\mathds{S}_{d}^{k}} = 0$. Therefore,
	\begin{equation}\label{same_copy_exp_val}
		\begin{aligned}
			2\sum_{\substack{p < q\\ \text{same copy}}}\sum_{\alpha=1}^{d^2-1} \bracket{\mathds{S}_{d}^{k}}{\lambda_{\alpha}^{(p)}\lambda_{\alpha}^{(q)}}{\mathds{S}_{d}^{k}} &= - \sum_{p=1}^{dk} \sum_{\alpha=1}^{d^2-1} \bracket{\mathds{S}_{d}^{k}}{\left(\lambda_{\alpha}^{(p)}\right)^2}{\mathds{S}_{d}^{k}} - 2\sum_{\substack{p < q\\ \text{different copies}}}\sum_{\alpha=1}^{d^2-1} \bracket{\mathds{S}_{d}^{k}}{\lambda_{\alpha}^{(p)}\lambda_{\alpha}^{(q)}}{\mathds{S}_{d}^{k}}\\
			&= - 2k(d^2-1) - 2k(k-1)(d-1) \neq 0.
		\end{aligned}
	\end{equation}
\end{itemize}

From \eqref{diff_copies_exp_val} and \eqref{same_copy_exp_val}, we get that the expectation value of $\sum_{\alpha=1}^{d^2-1}\lambda_{\alpha}^{(p)}\lambda_{\alpha}^{(q)}$ on $\ket{\mathds{S}_{d}^{k}}$ is never zero. This is in contradiction with \eqref{unitary_invariant_state_exp_val}. Therefore, $\ket{\mathds{S}_{d}^{k}}$ does not have support on the trivial representation, i.e. $\dim\mathcal{Q}_{\lambda} > d$ for any representation in \eqref{Schur_lemma_rho_B} and $\rank(\rho_{B}^{\mathds{S}}) \geq d$ for any bipartition $(B|\bar{B})$.


\subsection{Metrology with $\ket{\mathds S_d^k}$}\label{app:symsing}

\subsubsection{One and two-particle reduced states}
As before, the one-particle reduced state is $\rho^{\mathds S}_1 = \mathds 1/d$. For the two-particle reduced state, we again take the decomposition in the basis $\{\lambda_a\}_a$,
\begin{align}\label{eq:rho21}
	\rho_2^{\mathds S} = \frac{\mathds 1}{d^2} + \frac{1}{4} \sum_{a,b=1}^{d^2-1} \tr(\rho_2^{\mathds S}\lambda_a\otimes \lambda_b)\lambda_a\otimes \lambda_b
\end{align}Now, note that $\rho_2^{\mathds S}$ is invariant under $U\otimes U$. Equivalently, $[\rho_2^{\mathds S},U\otimes U]=0$. Further note that $\mathcal H\otimes \mathcal H = \mathcal H_{\text{sym}}\oplus \mathcal H_{\text{asym}}$. Therefore, by Schur's lemma,  $\rho_2^{\mathds S}$ must be proportional to the identity separately in these sectors. Hence,
\begin{align}
	\rho_2^{\mathds S} = c_s \left(\frac{1+S}{2} \right) + c_a\left(\frac{1-S}{2} \right) = \frac{c_s+c_a}{2}\mathds 1 + \frac{c_s-c_a}{2} S,
\end{align}The above form of the reduced state directly leads to $\tr(\rho_2^{\mathds S}\lambda_a\otimes \lambda_b)=c\delta_{ab}$, where $c$ is some constant. Substituting in Eq. \eqref{eq:rho21} we have 
\begin{align}\label{eq:state2}
	\rho_2^{\mathds S} = \frac{\mathds 1}{d^2} + 4c \sum_{a}\lambda_a\otimes \lambda_a = \frac{\mathds 1}{d^2} + \frac{c}{2} \left( S-\frac{\mathds 1}{d}\right)
\end{align}There are now two possibilities. Either the reduction is taken on particles belonging to the same site or two different sites.

Consider particles at the same site. The projector $P_i$ projects all $k$ particles at a particular site $i$ to its completely symmetric subspace. Therefore, 
\begin{align}
	S\rho_{2,\text{same}}^{\mathds S}
	= \rho_{2,\text{same}}^{\mathds S}
	\implies \expval{S}_{\rho_{2,\text{same}}^{\mathds S}} =1.
\end{align}Consider 
\begin{align}
	\sum_a \expval{\lambda_a\otimes \lambda_a} = 2\left(1-\frac{1}{d}\right) \implies (d^2-1)c_{\text{same}} = 2\left(1-\frac{1}{d}\right) 
\end{align}Therefore, we have 
\begin{align}
	c_{\text{same}} = \frac{2}{d(d+1)},
\end{align}which gives us 
\begin{align}
	\rho_{2,\text{same}}^{\mathds S} = \frac{\mathds 1+S}{d(d+1)}.
\end{align}
Now consider particles at different sites. Let $A$ and $B$ be traceless single-particle operators. Let 
\begin{align}
	c_{\text{diff}}(A,B) = \expval{A^{(i,\alpha)}B^{(j,\beta)}},\quad i\neq j.
\end{align}Note that for one particle $\expval{A^{(i,\alpha)}B^{(i,\alpha)}} = \tr(AB)/d$, and for different particles at the same site,
\begin{align}
	c_{\text{same}}(A,B) = \tr(\rho_{2,\text{same}}^{\mathds S}A\otimes B) = \frac{\tr(AB)}{d(d+1)}.
\end{align}Consider the collective operators $Q_A = \sum_{i}\sum_\alpha A^{(i,\alpha)}$ and $Q_B = \sum_{i}\sum_\alpha B^{(i,\alpha)}$. Due to the unitary invariance of $\ket{\mathds S_d^k}$, we have $Q_A\ket{\mathds S_d^k}=Q_B\ket{\mathds S_d^k}=0$. Therefore, $\expval{Q_AQ_B}_{\ket{\mathds S_d^k}}=0$. Expanding this expectation value, we get three types of terms: i) $dk$ one-particle terms, ii) $dk(k-1)$ 2-particle terms with the same site, and iii) $d(d-1)k^2$ 2-particle terms with different sites. It can be seen that 
\begin{align}
	\expval{Q_AQ_B}
	= dk \frac{\tr(AB)}{d}
	+ dk(k-1)\frac{\tr(AB)}{d(d+1)}
	+ d(d-1)k^2 c_{\text{diff}}(A,B)=0.
\end{align}Using the above equation and setting $A=\lambda_a$ and $B=\lambda_b$, we obtain
\begin{align}
	c_{\text{diff}} = -\frac{2(d+k)}{dk(d^2-1)} \delta_{ab}.
\end{align}Using Eq. \eqref{eq:state2}, we finally obtain
\begin{align}
	\rho_{2,\text{diff}}^{\mathds S} = \frac{(kd+1)\mathds 1 - (d+k)S}{kd(d^2-1)}.
\end{align}

\subsubsection{Final form of the QFI}
We write the variance separating same-site and different site terms
\begin{align}
	\tr(\rho^{\mathds S}H'^2) = \sum_i \sum_{\alpha,\beta} x^2_i \expval{\xi^{(i,\alpha)}\xi^{(i,\beta)}} + \sum_{i\neq j}\sum_{\alpha,\beta}x_ix_j \expval{\xi^{(i,\alpha)}\xi^{(j,\beta)}}.
\end{align}We therefore require the following three types of expectation values 
\begin{align}
	&\expval{(\xi^{(i,\alpha)})^2} = \frac{\tr(\xi^2)}{d} \quad (k\text{ terms})\\
	&\expval{\xi^{(i,\alpha)}\xi^{(i,\beta)}} = \frac{\tr(\xi^2)}{d(d+1)}\quad (k(k-1)\text{ terms})\\
	&\expval{\xi^{(i,\alpha)}\xi^{(j,\beta)}} = -\frac{d+k}{kd(d^2-1)}\tr(\xi^2),\,\,(i\neq j)\quad (k^2\text{ terms})
\end{align}Using the above, counting the summations and again using the same procedure with $x_i\to x_i-\bar x$, we obtain 
\begin{align}\label{eq:appqfifu}
	F_Q^{\mathds S} = 4  \tr(\rho^{\mathds S}H'^2) = \frac{4kd(d+k)\tr(\xi^2)}{d^2-1}\sigma_x^2 = \frac{4\tr(\xi^2)}{d(d^2-1)} n(n+d^2)\sigma_x^2.
\end{align}For fixed $d$ and large $n$, $F_Q^{\mathds S}$ grows as $\mathcal O(n^2)$, corresponding to a Heisenberg scaling.

\subsection{Optimal state for macroscopic gradient metrology}

\subsubsection{Proof of Result 6}\label{app:th1}

Collective unitary invariance of $\ket{\psi}$ means that for any traceless Hermitian operator $A$,
\begin{align}	\label{eq:pure_darkness}
	Q_A\ket{\psi}=0,
	\qquad
	Q_A=\sum_{i=1}^d\sum_{\alpha=1}^k A^{(i,\alpha)}.
\end{align}
Since Eq.~\eqref{eq:pure_darkness} holds for $A=\xi$, a uniform shift of all
positions has no effect,
\begin{align}
	H'\ket{\psi}
	=\sum_{i,\alpha}x_i\xi^{(i,\alpha)}\ket{\psi}=
	\sum_{i,\alpha}(x_i-\bar x)\xi^{(i,\alpha)}\ket{\psi}.
\end{align}
We define the $n$ real coefficients $w_{i,\alpha}=x_i-\bar x$. They satisfy
\begin{align}
	\sum_{i,\alpha}w_{i,\alpha}=0,
	\qquad
	\sum_{i,\alpha}w_{i,\alpha}^2
	=
	k\sum_i(x_i-\bar x)^2
	=
	kd\sigma_x^2.
	\label{eq:w_norm}
\end{align}Returning to the metrology problem, we have
\begin{align}
	F_Q(\ket{\psi},H')
	=
	4\Delta H'^2.
\end{align}
Moreover, $\expval{H'}_{\psi}=0$ by collective $SU(d)$ invariance and tracelessness of $\xi$, so
\begin{align}
	F_Q(\ket{\psi},H')
	=
	4\|H'\ket{\psi}\|^2 = 4 \|\sum_{p=1}^n w_p \xi^{(p)} \ket{\psi}\|^2,
\end{align}where we have used the combined index $p\equiv (i,\alpha)$. Applying Lemma~\ref{lem:1} to the coefficients in
Eq.~\eqref{eq:w_norm},
\begin{align}
	F_Q(\ket{\psi},H')
	&\leq
	4\frac{d+k}{d^2-1}\tr(\xi^2)
	\,kd\sigma_x^2
	\nonumber\\
	&=
	\frac{4kd(d+k)\tr(\xi^2)}{d^2-1}\sigma_x^2.
\end{align}
Comparing with Eq. \eqref{eq:appqfifu}, we see that $\ket{\mathds S_d^k}$ saturates the bound for any spatial profile.

\qed \\
Note that we do not claim uniqueness for every spatial profile.  Special spatial configurations
may have additional pure states which attain the same QFI.

\subsubsection{Proof of Result 7}

\label{app:th2}

For fixed $\rho$, the unitary QFI is a quadratic function of the spatial
profile, so we may write \cite{Liu_2019}
\begin{align}
	F_Q[\rho,H'(\boldsymbol{x})]
	=
	\boldsymbol{x}^{T}J_\rho\boldsymbol{x},
	\label{eq:Jrho_definition}
\end{align}
where $J_\rho$ is a real, symmetric and positive-semidefinite $d\times d$ matrix.  For any traceless operator $A$, we define $L_i(A) = \sum_\alpha A^{(i,\alpha)}$. Then, collective unitary invariance of $\rho$ implies
\begin{align}
	\left[\sum_iL_i(A),\rho\right]=0
\end{align}
for every traceless Hermitian $A$.  Since
$H'(\boldsymbol{x}+c\boldsymbol{1})
=H'(\boldsymbol{x})+c\sum_iL_i(\xi)$, a uniform shift of all positions leaves
the encoded state (and the associated QFI) unchanged. Stated in another way, it is simple to see that the uniform spatial configuration $\boldsymbol x = \boldsymbol{1}$ gives $F_Q[\rho,H'(\boldsymbol{1})] =0 = \boldsymbol{1}^TJ_\rho \boldsymbol{1}$. As $J_\rho\geq 0$,  this directly gives $J_\rho \boldsymbol{1}=0$. Next, we consider the spatial profile $\boldsymbol x = \boldsymbol{e_i}$, where $\boldsymbol{e_i}$ contains 1 at position $i$, and 0 in others. This gives $H'(\boldsymbol{e_i}) = L_i(\xi)$ and therefore, $F_Q[\rho,L_i(\xi)] = \boldsymbol{e_i}^T J_\rho \boldsymbol{e_i} = [J_\rho]_{ii}$. This gives us $\sum_i F_Q[\rho,L_i(\xi)] = \tr(J_\rho)$. 

Due to the collective unitary invariance of $\rho$, we have $\expval{L_i(\xi)}_\rho=0$ and $\expval{(\xi^{(i,\alpha)})^2} =\tr(\xi^2)/d$. Moreover, as a consequence of Schur-Weyl duality, any two-particle reduced state is of the form $\rho_2 = a\mathds 1 + b S$. For two different particles $\alpha\neq\beta$ at the same site $i$, let
$f_{i,\alpha\beta}=\expval{S}_{\rho_2}$.  Using normalization and the swap
	identity gives
\begin{align}
	\expval{\xi^{(i,\alpha)}\xi^{(i,\beta)}}_\rho
	=
	\frac{df_{i,\alpha\beta}-1}{d(d^2-1)}\tr(\xi^2)
	\leq
	\frac{\tr(\xi^2)}{d(d+1)},
	\label{eq:same_site_bound}
\end{align}
because $f_{i,\alpha\beta}\leq1$.  Consequently, summing over $\alpha,\beta$, we obtain
\begin{align}
	\expval{L_i(\xi)^2}_\rho
	&\leq
	\frac{k\tr(\xi^2)}{d}
	+
	\frac{k(k-1)\tr(\xi^2)}{d(d+1)}
	\nonumber\\
	&=
	\frac{k(d+k)\tr(\xi^2)}{d(d+1)} \equiv A_{\text{max}}
\end{align}Using $F_Q(\rho,H)\leq4\Delta H^2$ and
$\expval{L_i(\xi)}_\rho=0$,
\begin{align}
\tr J_\rho
\leq
4\sum_i\expval{L_i(\xi)^2}_\rho
\leq
4dA_{\max}.
\label{eq:J_trace_bound_notes}
\end{align}Since $J_\rho\boldsymbol{1}=0$, only its restriction to the
$(d-1)$-dimensional subspace orthogonal to $\boldsymbol{1}$ is relevant. Moreover, note that we can write $\boldsymbol{x}^T J_\rho \boldsymbol{x} = \boldsymbol{y}^TJ_\rho \boldsymbol{y} $, with $\boldsymbol{y} = \boldsymbol{x}-\bar x\boldsymbol{1}$ (such that $\boldsymbol{y}\perp \boldsymbol{1}$). Also, $\sigma^2_x  =\frac{1}{d}\sum_i (x_i-\bar x)^2 = \frac{1}{d}\norm{y}^2$. Therefore, 
\begin{align}
	\frac{F_Q[\rho,H']}{\sigma_x^2} = d\frac{\boldsymbol{y}^TJ_\rho \boldsymbol{y}}{\boldsymbol{y}^T\boldsymbol{y}}, \quad \boldsymbol{y}\in \boldsymbol{1}^\perp, \, \boldsymbol{y}\neq 0.
\end{align}Ignoring the $d$ factor, the RHS is a Rayleigh quotient of a real-symmetric matrix. Therefore, 
\begin{align}
	\lambda_{\text{min}} = \min_{\boldsymbol y\neq0}\frac{\boldsymbol{y}^TJ_\rho \boldsymbol{y}}{\boldsymbol{y}^T\boldsymbol{y}},
\end{align}where $\lambda_{\text{min}}$ is the minimum eigenvalue of $J_\rho$ on the $\boldsymbol{1}^\perp$ subspace. Therefore,
\begin{align}
\min_{\boldsymbol{x}:\sigma_x>0}
\frac{F_Q[\rho,H'(\boldsymbol{x})]}{\sigma_x^2}
&=
d\lambda_{\min}\!\left(J_\rho\big|_{\boldsymbol{1}^{\perp}}\right)
\nonumber\\
&\leq
\frac{d\tr J_\rho}{d-1}
\leq
\frac{4kd(d+k)\tr(\xi^2)}{d^2-1},
\label{eq:min_eigenvalue_bound}
\end{align}where we have used the trivial identity $\lambda_{\text{min}}\leq \frac{1}{d-1}\sum_{i=1}^{d-1}\lambda_i = \frac{1}{d-1}\tr(J_\rho)$ (noting that one eigenvalue is zero and the others are non-negative). The state $\ket{\mathds S_d^k}$ saturates the bound. 

It remains to prove uniqueness. If a state saturates
Eq.~\eqref{eq:min_eigenvalue_bound}, it must also saturate
Eq.~\eqref{eq:J_trace_bound_notes}.  Equality in
Eq.~\eqref{eq:same_site_bound} requires
$f_{i,\alpha\beta}=1$ for every pair of particles at the same site. This implies that
every vector in the support of $\rho$ must be completely symmetric among the $k$
particles at each site. 

Let $P$ be the projector onto the support of $\rho$.  Equality in
$F_Q(\rho,L_i(\xi))\leq4\Delta L_i(\xi)^2$ implies
\begin{align}
	PL_i(\xi)P=0.
\end{align}Summing over $i$, we obtain
\begin{align}
	PQ_{\xi}P = 0, \,\,\text{with}\,\,Q_\xi = \sum_i L_i(\xi).
\end{align}Due to the unitary invariance of $\rho$, we have $[\rho,Q_\xi] =0$. Therefore, the support projector also commutes, $[P,Q_\xi]=0$. This means that $PQ_\xi P = P Q_\xi$ = 0. Using Hermiticity, $Q_\xi P =0$. Assuming $\xi\neq0$, collective invariance of $P$ also gives
$Q_{U\xi U^\dagger}P=0$ for every $U\in SU(d)$.
Since these unitary conjugates span all traceless Hermitian operators,
the support of $\rho$ consists of global $SU(d)$ singlets.

Therefore, the support of $\rho$ lies in the global singlet sector of
$\bigl(\operatorname{Sym}^k\mathbb C^d\bigr)^{\otimes d}$.
This sector is one-dimensional. Indeed, represent these symmetric
tensors by polynomials in $d$ vector variables, homogeneous of degree $k$ in each variable. The $SU(d)$-invariant polynomials, equivalently the $SL(d,\mathbb C)$-invariant polynomials, are generated by the determinant of these $d$ vectors. The prescribed degrees therefore leave only the $k$-th power of this determinant. Since $\ket{\mathds S_d^k}$ is a nonzero singlet in this space, it spans the sector. Consequently, $\rho=\ketbra{\mathds S_d^k}{\mathds S_d^k}$.

\qed

\subsection{Useful lemma}
\begin{lemma}\label{lem:1}
	Let $n=kd$, and let $U^{\otimes n}\ket{\psi} =\ket{\psi}$ $\forall \, U\in SU(d)$. Let $\xi$ be traceless and Hermitian, and let $w_1,\ldots,w_n\in\mathbb R$ satisfy $	\sum_{p=1}^n w_p=0$. Then
	\begin{equation}
			\left\|\sum_{p=1}^{n}w_p\,\xi^{(p)}\ket{\psi}\right\|^2
			\leq
			\frac{d+k}{d^2-1}\,\tr(\xi^2)
			\sum_{p=1}^{n}w_p^2 .
	\end{equation}
\end{lemma}

\begin{proof}
	
	Define $O_{\bm w}:=\sum_{p=1}^n w_p\,\xi^{(p)}$, where the set of coefficients $\bm w = (w_1,\dots,w_n)$ is such that $\sum_p w_p=0$. Such set can be seen as an $n$-dimensional permutation representation of $n$ coordinates. For such a representation of the symmetric group $\mathcal{S}_n$, there are only two irreducible representations, i.e. the vector space decomposes as $\mathbf{V}_{\text{triv}} \oplus \mathbf{V}_{\text{std}}$. The first one is the trivial representation $\bm 1 = (1,\dots,1)$, while the second one is the $(n-1)$-dimensional representation orthogonal to it, i.e. $\bm{1}^{T}\cdot \bm{w} = 0$, called the standard representation. Since the coefficients of $O_{\bm w}$ satisfy $\sum_p w_p=0$, each $\bm w$ belongs to the standard representation of $\mathcal{S}_n$.
	
	Furthermore, as shown in Appendix \ref{app:unique}, any state $\ket{\psi} \in \mathcal{H}_{d}^{\otimes n}$ such that $U^{\otimes n}\ket{\psi} = \ket{\psi}$, for each $U \in SU(d)$ and $n = kd$, belongs to the singlet irreducible representation associated with the Young tableaux $\lambda = (k,\dots,k)$. Applying $O_{\bm w}$ to $\ket{\psi}$ generates an object $O_{\bm{w}}\ket{\psi}$ which transforms as  $O_{\bm{w}}(\xi)\ket{\psi} \longmapsto O_{\bm{w}}(U\xi U^{\dagger})\ket{\psi}$, i.e. it transforms in the adjoint representation, whose Young tableaux is $\mu = (k+1,k,\dots,k,k-1)$ \cite{fulton1991representation}. This representation appears only once in the decomposition of $\mathbf{V}_{\text{std}}\otimes \mathbf{V}_{\lambda}$ (i.e. the ones associated with the operation $O_{\bm{w}}\ket{\psi}$) \cite{Stanley_2023}.  Therefore, the map $\Phi_{\xi}(\bm{w} \otimes \ket{\psi}) = O_{\bm{w}}(\xi)\ket{\psi}$, once restricted on the vector space labelled by the Young tableaux $\mu$, $\mathbf{V}_\mu$, is such that $\Phi_{\xi}|_{\mathbf{V}_\mu} \,  \rho_{\mu}(\pi) = \rho(\pi)\Phi_{\xi}|_{\mathbf{V}_\mu}$, where $\rho_{\mu}(\pi)$ is the matrix representation of the permutation $\pi$ in $\mathbf{V}_\mu$ and $\rho(\pi)$ is the matrix representation on the codomain of $\Phi_{\xi}$. By Schur's lemma, $\Phi_{\xi}^{\dagger}\Phi_{\xi} = \gamma_{\xi}P_{\mu}$, where $\gamma_{\xi}$ is a scalar and $P_{\mu}$ is the projector onto the $\mathbf{V}_{\mu}$ vector space. We get
	\begin{equation}
			\|O_{\bm w}\ket{\psi}\|^2 = \gamma_{\xi}\|P_{\mu}(\bm{w} \otimes \ket{\psi})\|^2
		\leq \gamma_\xi\,\|\bm w\|^2,
		\qquad
		\|\bm w\|^2=\sum_p w_p^2,
		\label{eq:gamma-bound}
	\end{equation}
	where we have used that $\|P_{\mu}(\bm{w} \otimes \ket{\psi})\|^2 \leq \|\bm w\otimes \ket{\psi}\|^2 = \|\bm w\|^2$, since $\ket{\psi}$ is normalised.

%
	It remains only to determine $\gamma_\xi$. 	Let $\{\bm w^{(r)}\}_{r=1}^{n-1}$ be an orthonormal basis of $\mathbf{V}_{\text{std}}$. Then
	\begin{equation}
		\sum_{r=1}^{n-1}w_p^{(r)}w_q^{(r)}
		=
		\delta_{pq}-\frac1n.
		\label{eq:standard-projector}
	\end{equation}
	Let $\{\ket{\psi_m}\}_{m=1}^{f^\lambda}$ be an orthonormal basis of the singlet sector, where $f^\lambda$ is the dimension of the irreducible representation of $\mathcal{S}_n$ labelled by $\lambda$. Using \eqref{eq:standard-projector},
	\begin{align}
		\sum_{r,m}\|O_{\bm w^{(r)}}\ket{\psi_m}\|^2
		&=
		\sum_m\bra{\psi_m}
		\left[
		\sum_{p,q}
		\left(\delta_{pq}-\frac1n\right)
		\xi^{(p)}\xi^{(q)}
		\right]
		\ket{\psi_m}
		\\
		&=
		\sum_m\bra{\psi_m}
		\left[
		\sum_p (\xi^{(p)})^2
		-
		\frac1n\left(\sum_p\xi^{(p)}\right)^2
		\right]
		\ket{\psi_m}.
	\end{align}
	For every singlet,
	\begin{equation}
		\left(\sum_p\xi^{(p)}\right)\ket{\psi_m}=0.
	\end{equation}
	Also every one-particle reduced state of a global singlet is maximally mixed, giving
	\begin{equation}
		\bra{\psi_m}(\xi^{(p)})^2\ket{\psi_m}
		=
		\frac{\tr(\xi^2)}{d}.
	\end{equation}
	Therefore
	\begin{equation}
		\sum_{r,m}\|O_{\bm w^{(r)}}\ket{\psi_m}\|^2
		=
		f^\lambda\,\frac{n}{d}\,\tr(\xi^2).
		\label{eq:trace-one}
	\end{equation}
	On the other hand, using \eqref{eq:gamma-bound},
	\begin{equation}
		\sum_{r,m}\|O_{\bm w^{(r)}}\ket{\psi_m}\|^2
		=
		\gamma_\xi f^\mu,
		\label{eq:trace-two}
	\end{equation}where $f^\mu$ is the dimension of the irreducible representation of $\mathcal{S}_n$ labelled by $\mu$. Combining \eqref{eq:trace-one} and \eqref{eq:trace-two},
	\begin{equation}
		\gamma_\xi
		=
		\frac{nf^\lambda}{d f^\mu}\,\tr(\xi^2).
	\end{equation}
	Both $f^{\lambda}$ and $f^{\mu}$ can be computed with the hook-length formula to get
	\begin{equation}
		\frac{f^\mu}{f^\lambda}
		=
		\frac{k(d^2-1)}{d+k}.
	\end{equation}
	Since $n=kd$,
	\begin{align}
		\gamma_\xi
		=
		\frac{kd}{d}
		\frac{d+k}{k(d^2-1)}\tr(\xi^2)=
		\frac{d+k}{d^2-1}\tr(\xi^2).
	\end{align}
	Substituting this into \eqref{eq:gamma-bound} yields
	\begin{equation}
			\left\|\sum_{p=1}^{n}w_p\,\xi^{(p)}\ket{\psi}\right\|^2
			\leq
			\frac{d+k}{d^2-1}\,\tr(\xi^2)
			\sum_{p=1}^{n}w_p^2 .
		\end{equation}
\end{proof}

\subsection{Proof of Result 8}
\label{app:optmeas}
We first show that $H'$ has no elements within the singlet sector. We rewrite $H' = \sum_{i=1}^n c_i\xi^{(i)}$ and let $\ket{\phi}$ and $\ket{\chi}$ be arbitrary global $SU(d)$ singlets. For any $U\in SU(d)$, we can therefore write, 
\begin{equation}
	\begin{aligned}
		\bra{\phi} \xi^{(i)}\ket{\chi} &= \bra{\phi} (U^{\otimes n})^\dagger \xi^{(i)} U^{\otimes n} \ket{\chi} \\
		& = \bra{\phi} \left(\int dU U^\dagger \xi U \right)^{(i)}\ket{\chi} \\
		& =  \frac{\tr(\xi)}{d}\braket{\phi}{\chi} =0,
	\end{aligned}
\end{equation}where in the second line, we have averaged over all unitaries and in the third line have used the standard Haar-twirling identity $\int dU U^\dagger X U = \tr(X)\mathds 1/d$ and $\tr(\xi)=0$. Therefore, $P_0 \xi^{(i)} P_0 =0$ $\forall\, i$, and consequently, $P_0 H' P_0 =0$.
Now, since $P_0 \ket{\psi}=\ket{\psi}$ by assumption, we have
\begin{align}
	P_0 H' \ket\psi =0.
\end{align}As previously noted, $\expval{H'}_\psi =0$ is implied by the tracelessness of $\xi$ and the unitary invariance of $\ket\psi$. Therefore, the QFI $F_Q[\ket\psi,H'] = 4\Delta H'^2 = 4\expval{H'^2}_\psi$.

Now let us consider the classical Fisher information of the time-evolved state $\ket{\psi_\theta} = e^{-iH'\theta}\ket\psi$ with the binary Casimir measurement $\{P_0,\openone - P_0\}$. Let the probability corresponding to the non-singlet outcome be $p_\perp (\theta) = \bra{\psi_\theta} (\openone -P_0)\ket{\psi_\theta}$. The Fisher information is given by 
\begin{align}\label{eq:cfi1}
F_{\text{Cas}}(\theta) = \frac{(\partial_\theta p_\perp (\theta))^2}{p_\perp (\theta)(1-p_\perp (\theta))}.
\end{align}Writing $\ket{\psi_\theta} = \ket\psi - i\theta H'\ket\psi
- \frac{\theta^2}{2}H'^2\ket\psi + \mathcal O(\theta^3)$, and using $(\openone - P_0)\ket\psi =0$ and $(\openone - P_0)H'\ket\psi = H'\ket\psi$, we obtain
\begin{align}
	p_\perp(\theta) = \theta^2 \expval{H'^2}_\psi + \mathcal O(\theta^3),
\end{align}and $\partial_\theta p_\perp(\theta) = 2\theta\expval{H'^2}_\psi + \mathcal O(\theta^2)$. Substituting in Eq. \eqref{eq:cfi1}, we have
\begin{align}
	F_{\text{Cas}} (\theta) = \frac{4\theta^2\expval{H'^2}_\psi^2 + \mathcal O(\theta^3)}{ \theta^2 \expval{H'^2}_\psi + \mathcal O(\theta^3)}.
	\end{align}
Clearly, 
\begin{align}
\lim_{\theta\to0} F_{\text{Cas}}(\theta) = 4\expval{H'^2}_\psi = F_Q[\ket\psi, H'].
\end{align}
\qed

\end{document}